\documentclass[twoside,final,11pt]{entics}
\usepackage{graphicx}
\usepackage{microtype}
\usepackage{amsmath}
\usepackage{stmaryrd}

\usepackage{ottalt}
\usepackage{pgffor}

\newenvironment{grassdefnblock}[3][]{ \framebox{\mbox{#2}} \quad #3 \\[0pt]}{}

\newcommand{\grassnt}[1]{\mathit{#1}}
\newcommand{\grassmv}[1]{\mathit{#1}}

\newcommand{\grasssym}[1]{#1}

  \renewottcommands[grass]

\newenvironment{grassSdefnblock}[3][]{ \framebox{\mbox{#2}} \quad #3 \\[0pt]}{}

\newcommand{\grassSnt}[1]{\mathit{#1}}
\newcommand{\grassSmv}[1]{\mathit{#1}}

\newcommand{\grassSsym}[1]{#1}

  \renewottcommands[grassS]

\makeatletter
\newcommand*{\namedrule}[3]{%
  \def\ou@t{\csname #1drule#2XX#3\endcsname}%
  \mprset{flushleft}%
  \ou@t{}%
}

\newcommand*{\namedrules}[3]{%
  \newif\ifprevious%
  \global\previousfalse%
    \foreach \@iter in {#3}{%
      \ifprevious \and\fi\global\previoustrue%
      \namedrule{#1}{#2}{\@iter}%
    }%
}

\newcommand*{\rulename}[3]{%
  \def\ou@t{\csname #1RenameRule#2XX#3\endcsname}%
  \text{\ou@t{}}}

\newcommand*{\rulenames}[3]{%
  \newif\ifprevious%
  \global\previousfalse%
  \foreach \@iter in {#3}{%
    \ifprevious, \fi\global\previoustrue%
    \rulename{#1}{#2}{\@iter}%
  }}

\newcommand*{\renamerule}[4]{%
    \expandafter\renewcommand\csname #1RenameRule#2XX#3\endcsname{#4}}
\makeatother

\renamerule{grassS}{term}{var}{\textsc{S.Var}}
\renamerule{grassS}{term}{unitI}{\textsc{S.Unit.i}}
\renamerule{grassS}{term}{unitE}{\textsc{S.Unit.e}}
\renamerule{grassS}{term}{funI}{\textsc{S.Fun.i}}
\renamerule{grassS}{term}{funE}{\textsc{S.Fun.e}}
\renamerule{grassS}{term}{pairI}{\textsc{S.Pair.i}}
\renamerule{grassS}{term}{pairE}{\textsc{S.Pair.e}}
\renamerule{grassS}{term}{sumIL}{\textsc{S.Sum.il}}
\renamerule{grassS}{term}{sumIR}{\textsc{S.Sum.ir}}
\renamerule{grassS}{term}{sumE}{\textsc{S.Sum.e}}
\renamerule{grassS}{term}{weakE}{\textsc{Weak}}
\renamerule{grassS}{term}{contE}{\textsc{Cont}}
\renamerule{grassS}{term}{weak}{\textsc{S.Weak}}
\renamerule{grassS}{term}{cont}{\textsc{S.Cont}}
\renamerule{grassS}{term}{sub}{\textsc{S.Sub}}
\renamerule{grassS}{term}{boxI}{\textsc{S.Box.i}}
\renamerule{grassS}{term}{boxE}{\textsc{S.Box.e}}

\renamerule{grassS}{beta}{unit}{\textsc{S.Unit-$\beta$}}
\renamerule{grassS}{beta}{pair}{\textsc{S.Pair-$\beta$}}
\renamerule{grassS}{beta}{arrow}{\textsc{S.Fun-$\beta$}}
\renamerule{grassS}{beta}{box}{\textsc{S.Box-$\beta$}}

\renamerule{grassS}{eta}{unit}{\textsc{S.Unit-$\eta$}}
\renamerule{grassS}{eta}{pair}{\textsc{S.Pair-$\eta$}}
\renamerule{grassS}{eta}{arrow}{\textsc{S.Fun-$\eta$}}
\renamerule{grassS}{eta}{box}{\textsc{S.Box-$\eta$}}

\renamerule{grass}{term}{var}{\textsc{Var}}
\renamerule{grass}{term}{unitI}{\textsc{Unit.i}}
\renamerule{grass}{term}{unitE}{\textsc{Unit.e}}
\renamerule{grass}{term}{arrowI}{\textsc{Fun.i}}
\renamerule{grass}{term}{arrowE}{\textsc{Fun.e}}
\renamerule{grass}{term}{pairI}{\textsc{Pair.i}}
\renamerule{grass}{term}{pairE}{\textsc{Pair.e}}
\renamerule{grass}{term}{sumIL}{\textsc{Sum.il}}
\renamerule{grass}{term}{sumIR}{\textsc{Sum.ir}}
\renamerule{grass}{term}{sumE}{\textsc{Sum.e}}
\renamerule{grass}{term}{weak}{\textsc{Weak}}
\renamerule{grass}{term}{cont}{\textsc{Cont}}
\renamerule{grass}{term}{sub}{\textsc{Sub}}

\renamerule{grass}{term}{dropI}{\textsc{Drop.i}}
\renamerule{grass}{term}{dropE}{\textsc{Drop.e}}
\renamerule{grass}{term}{raiseI}{\textsc{Raise.i}}
\renamerule{grass}{term}{raiseE}{\textsc{Raise.e}}

\renamerule{grass}{type}{unit}{\textsc{Ty.Unit}}
\renamerule{grass}{type}{pair}{\textsc{Ty.Pair}}
\renamerule{grass}{type}{arrow}{\textsc{Ty.Fun}}
\renamerule{grass}{type}{drop}{\textsc{Ty.Drop}}
\renamerule{grass}{type}{raise}{\textsc{Ty.Raise}}
\renamerule{grass}{type}{sum}{\textsc{Ty.Sum}}

\renamerule{grass}{beta}{unit}{\textsc{Unit.$\beta$}}
\renamerule{grass}{beta}{pair}{\textsc{Pair.$\beta$}}
\renamerule{grass}{beta}{arrow}{\textsc{Fun.$\beta$}}
\renamerule{grass}{beta}{raise}{\textsc{Raise.$\beta$}}
\renamerule{grass}{beta}{drop}{\textsc{Drop.$\beta$}}
\renamerule{grass}{beta}{sumL}{\textsc{Sum.$\beta$.L}}
\renamerule{grass}{beta}{sumR}{\textsc{Sum.$\beta$.R}}

\renamerule{grass}{eta}{unit}{\textsc{Unit.$\eta$}}
\renamerule{grass}{eta}{pair}{\textsc{Pair.$\eta$}}
\renamerule{grass}{eta}{arrow}{\textsc{Fun.$\eta$}}
\renamerule{grass}{eta}{raise}{\textsc{Raise.$\eta$}}
\renamerule{grass}{eta}{drop}{\textsc{Drop.$\eta$}}
\renamerule{grass}{eta}{sum}{\textsc{Sum.$\eta$}}

\newcommand*{\mb}{\mathbf}
\newcommand*{\mc}{\mathcal}

\newcommand*{\ms}{\mathsf}
\newcommand*{\mbb}{\mathbb}
\newcommand*{\Nat}{\mbb N}

\newcommand*{\from}{\colon}
\newcommand*{\at}{\odot}
\newcommand*{\tensor}{\otimes}
\newcommand*{\ox}{\otimes}
\newcommand*{\op}{{\mathrm{op}}}
\newcommand*{\iso}{\cong}

\newcommand*{\proves}{\vdash}
\DeclareMathOperator{\id}{id}

\newcommand*{\interp}[1]{\llbracket #1 \rrbracket}
\newcommand*{\interpCtx}[4]{\llbracket #1 \mid #2 \at #3 \rrbracket_{#4}}
\newcommand*{\lto}{\multimap}
\makeatletter
\renewcommand*{\Vec}{\old@vec}
\makeatother
\DeclareMathOperator{\SMC}{\mb{SMC}}
\newcommand*{\hox}{\mathbin{\hat\ox}}

\newcommand*{\mode}{\ms}
\renewcommand{\phi}{\varphi}
\newcommand*{\grass}{\textsc{Grass}}

\DeclareMathOperator{\Cont}{\ms{Cont}}
\DeclareMathOperator{\Weak}{\ms{Weak}}
\newcommand*{\Type}{\ms{Type}}
\newcommand*{\Mode}{\bf{Mode}}

\newcommand*{\ExpAct}{\bf{ExpAct}}

\usepackage{tikz}
\usetikzlibrary{cd}
\tikzcdset
  { arrow style = tikz
  , diagrams = {> = {Straight Barb[scale = 0.8]}}
  }

\newcommand{\mlnode}[1]{\begin{tabular}{c}#1\end{tabular}}

\usepackage{enticsmacro}
\newif{\ifextended}
\extendedfalse

\def\conf{MFPS 2026}

\volume{6}

\def\volu{6}

\def\papno{8}
\def\mydoi{19721}

\def\lastname{Hanukaev, Eades III}

\def\runtitle{A unification of graded and substructural logics}

\def\copynames{P. Hanukaev, H. Eades III}

\begin{document}

\begin{frontmatter}

  \title{A Unification of Graded and Substructural Logics}%
  \thanks[ALL]{This work was supported by NSF grant \#2104535
    ``SHF: SMALL: Semantically and Practically Generalizing Graded Modal Types''}

  \author{Peter Hanukaev\thanksref{phanukaev}}
  \author{Harley Eades III\thanksref{heades}}
  \address{%
    School of Computer and Cyber Sciences\\
    Augusta University\\
    Augusta, Georgia, USA%
  }
  \thanks[phanukaev]{Email: \href{mailto:phanukaev@augusta.edu}
    {\texttt{\normalshape phanukaev@augusta.edu}}}
  \thanks[heades]{Email: \href{mailto:harley.eades@pm.me}
    {\texttt{\normalshape harley.eades@pm.me}}}
  \begin{abstract}
    Type systems which account for resource sensitive computations
    can generally be split into two styles:
    First, substructural logics such as Linear Logic
    which seek to restrict weakening and contraction
    and reintroduce them in a controlled manner;
    And second, graded systems which allow weakening and contraction by default,
    but track the use of variables quantitatively in some algebraic structure --
    usually a semiring.
    We present \grass{} (\textbf{Gra}ded and \textbf{s}ub\textbf{s}tructural),
    a type system which incorporates mechanisms from both of these approaches,
    thus allowing maximally flexible control over variable usage.
    Furthermore, \grass{} allows grades from an arbitrary collection of grade
    algebras to coexist in the same system,
    thus allowing different variables to be controlled with respect
    to different notions of resource within the same program.
    We develop the categorical semantics of \grass{}, and find that,
    on the level of categorical semantics,
    it subsumes multiple established systems such as
    LNL~\cite{benton_mixed_1995},
    Adjoint Logic~\cite{pruiksma_adjointlogic_2018},
    and \textsf{mGL}~\cite{vollmer_mixed_2025}.
  \end{abstract}
\end{frontmatter}

\section{Introduction}
\label{sec:introduction}
The treatment of program variables as resources has seen a rise over the last few
decades.
While classically variables were treated like propositions
which could be reused or discarded at will,
the resourceful treatment acknowledges that this is not true in
many common computing use cases.
Treating variables as resources allows programmers to verify interesting
correctness properties,
like the absence of memory leaks or use-after-free errors.

Two distinct approaches to the resourceful analysis of programs
exist in the literature.
\emph{Substructural logics}
restrict weakening and contraction,%
\footnote{Exchange can
also be restricted to create non-commutative logics~\cite{Lambek:1958}, but we do
not consider this here.}
prohibiting variables from being duplicated or discarded by default,
then selectively re-introducing these capabilities where needed.
In Linear Logic \cite{girard_linear_1987}, this is done via the of-course modality.
Benton generalized Linear Logic to Linear-non-Linear Logic \cite{benton_mixed_1995},
a system which features two sub-logics, one where assumptions are intuitionistic
and one where assumptions can be either intuitionistic or linear.
Pruiksma et al.~\cite{pruiksma_adjointlogic_2018} further generalized this system to
Adjoint Logic, allowing any number of logics, each with its own set of
permitted structural rules, to coexist in the same system.
This allows programmers to specify exactly which structural rules are applicable
to which variable in a program, independently of the other variables.

The second approach is to allow structural rules by default, but to
annotate assumptions with so called \emph{grades}, usually drawn from a semiring.
In this approach, the control over the usage of
assumptions stems from the algebraic structure of the grades with graded contraction being
modeled by addition and weakening being modeled by an explicit zero element.
Graded systems are commonly parametrized by their grade structure and can
therefore model a variety of resource notions.
Examples of applications are:
Garbage collection~\cite{choudhury_graded_2021},
liveness analysis~\cite{petricek_coeffects_2014},
security~\cite{%
  abel_unified_2020,%
  gaboardi_combining_2016,%
  liepelt_graded_2026,%
  moon_graded_2021,%
  orchard_quantitative_2019%
  },
dataflow analysis
\cite{petricek_coeffects_2014,uustalu_comonadic_2008,uustalu_signals_2005},
numerical sensitivity tracking~\cite{%
  azevedo_de_amorim_semantic_2017,de_amorim_really_2014%
},
and reasoning about probabilistic programs \cite{rajani_modal_2024}.

In this work we present \grass{}, a type system which unifies these two approaches.
Our system allows grades drawn from any number of grade algebras
to be used simultaneously.
For example, we can type a function
$
  \mathtt{read\_decrypt}
  \from (k : ^{\ms{Hi}} \mathtt{Key})
  \to (f :^1 \mathtt{FileHandle})
  \to \mathtt{Text}
$,
which reads the contents of an encrypted file $ f $ using the decryption key $ k $.
The grade annotation $ \ms{Hi} $ ensures that the key is used in a secure context,
while the grade annotation $ 1 $ ensures that the file handle is properly
accessed and closed.
Our type system further allows control over graded contraction and weakening.
For example two file handles used linearly should not be contractible
into one which can be used twice.

We develop a categorical semantics for our type system,
based on graded linear exponential comonads~\cite{katsumata_double_2018}.
Our type system is assembled from a collection of graded logics connected by morphisms,
and this design is mirrored in the categorical semantics.
We develop a novel notion of morphisms between graded comonads
based on morphisms of actegories [sic] \cite{capucci_actegories_2024}.
We show that known categorical models for substructural logics
give rise to models of \grass{} in a straightforward way.
Thus, on the level of categorical semantics,
\grass{} is capable of expressing both graded and substructural logics,
as well as combinations thereof.

\section{Single-mode system}
\label{sec:control}
  \grass{} is paratemeterized by a collection of so called \emph{modes},
each of which admits different graded structural rules.
In this section we explain our novel weakening and contraction rules
in the situation when this collection is instantiated to consist of only one mode.
In this case, \grass{} simplifies to be an ordinary graded type theory,
nearly identical to the purely graded fragment of
\textsf{mGL} \cite{vollmer_mixed_2025},
except with modified weakening and contraction rules.
We will focus on the structural rules here;
the full set of typing rules is given in
\ifextended
  \S\ref{asec:single-mode-spec} of the appendix.
\else
  the extended version of this paper \cite{extendedVersion}.
\fi
Before we begin with our treatment of the structural rules,
we review weakening and contraction in existing graded systems,
then explain how we control them further using modes.
For contraction, the idea is to restrict it to a set of mutually contractible grades,
subject to some conditions.
For weakening the situation is simpler: It can just be toggled on or off.
We end this section with a discussion of morphisms of modes,
which capture the notion that one mode has more structure than another.

\subsection{Weakening and contraction in existing systems}

We review the behavior of graded weakening and contraction in existing systems.
These structural rules are present in most graded type systems,
either as explicit rules~\cite{%
  ghica_bounded_2014,%
  petricek_coeffects_2013,%
  petricek_coeffects_2014,%
  vollmer_mixed_2025},
as admissible rules
\cite{choudhury_graded_2021,%
  gaboardi_combining_2016,%
  moon_graded_2021},
or even a mix of the two with an explicit weakening rule,
but admissible graded contraction
\cite{orchard_quantitative_2019}.
In this section we will consider typing judgments of the form $  \rho  \odot  \Gamma  \vdash  \grassSnt{t}  :  \grassSnt{T}  $
where $ t $ is a term, $ T $ is a type and~$ \Gamma $ a typing context.
Furthermore, $ \rho $ is a list of grades of the same length as $ \Gamma $.
Each entry of $ \rho $ dictates the capabilties with which
the variable in the corresponding position in $ \Gamma $ may be used by $ t $.
The graded contraction and weakening rules are now
\begin{mathpar}
  \namedrules{grassS}{term}{contE,weakE}.
\end{mathpar}
The contraction rule allows sharing a variable
across two sub-expressions in a term,
but also requires its capabilities to be split between the two expressions.
The weakening rule allows the introduction of additional variables into the context,
marked by the grade $ 0 $.
The graded structural rules can be seen as generalizations of the structural
rules available using the of-course modality $ !A $ in Linear Logic,
where they correspond to maps
$ !A \lto !A \ox !A $ and $ !A \lto I $ respectively.
The graded structural rules generalize these to
$ \square_{q_1 + q_2} A \lto \square_{q_1} A \ox \square_{q_2} A $ and
$ \square_0 A \lto I $,
where the type $ \square_q A $ indicates use at grade $ q $.
This allows fine grained tracking of the usage of each hypothesis.
Depending on the grade algebra being used,
the graded structural rules may simply correspond to the ones of linear logic,
but by choosing different kinds of semirings we can enforce different styles
of capability tracking.
For example,
if we use the natural numbers to grade our types we can track the precise number
of times variables are being used rather than simply one or zero.

\subsection{Modes}
In this section we formalize a novel structure used to control weakening and
contraction in the graded setting, which we call \emph{modes}.
The main insight is that contraction can be controlled by specifying
a set of mutually contractible grades called an \emph{ideal}.
The notion of ideals is borrowed from ring theory,
where ideals are a well-known concept (see e.g.\@ Lang~\cite{Lang_2002}).

\begin{definition}
  A \emph{semiring} is a set $ R $ together with binary operations
  $ (+), (\cdot) \from R \times R \to R $
  and elements $ 0, 1 \in R $ such that
  $ (R, +, 0) $ is a commutative monoid,
  $ (R, \cdot, 1) $ is a monoid
  and such that
  for all $ x, y, z \in R $
  we have
  $ x \cdot (y + z) = (x \cdot y) + (x \cdot z) $ and
  $ (x + y) \cdot z = (x \cdot z) + (y \cdot z) $ and
  $ 0 \cdot x = 0 = x \cdot 0 $.
  A \emph{preordered semiring} is a semiring equipped with a
  preorder $ (\le) $ such that $ (+) $ and $ (\cdot) $
  are monotone with respect to $ (\le) $, i.e.\@
  whenever $ x \le x' $ and $ y \le y' $, then $ x + y \le x' + y' $
  and likewise for $ (\cdot) $.
  Preordered semirings are also called \emph{grade algebras}
  and their elements are called \emph{grades}.
\end{definition}

\begin{example}
  \label{exmp:grade-algebras}
  \begin{enumerate}
  \item
    The set of natural numbers $ \Nat $ with the usual arithmetic operations
    forms a semiring.
    There are several natural choices of preorder on $ \Nat $:
    The usual ordering $ 0 \le 1 \le \mathellipsis $,
    its opposite $ 0 \ge 1 \ge \mathellipsis $,
    and the \emph{discrete ordering} $ n \le m \iff n = m $.

  \item
    The set $ \{0, 1\} $ can be equipped with two distinct semiring structures,
    by either defining $ 1 + 1 = 0 $ or $ 1 + 1 = 1 $.
    Only the latter structure will be of interest to us.
    There are three interesting preorders,
    defined by either $ 0 \le 1 $ or $ 0 \ge 1 $,
    or the discrete ordering.

  \item
    The \emph{none-one-tons} semiring is the set $ \{ 0, 1, \omega \} $.
    Addition and multiplication are fixed by the semiring axioms and the equations
    $ 1 + 1 = 1 + \omega = \omega + \omega = \omega \cdot \omega = \omega $.
    Again, there are multiple possible preorders,
    for example, the preorder generated by $ 0 \le \omega $
    or the one generated by $ 0 \le \omega \ge 1 $.

  \item
    There is a unique semiring with $ 0 = 1 $,
    which we denote $ \top $.
  \end{enumerate}
\end{example}

\begin{definition}
  Let $ R $ be a semiring.
  A \emph{(two sided) ideal} in $ R $ is a subset $ I \subseteq R $ such that
  $ 0 \in I $,
  and whenever $ x, y \in I $, then $ x + y \in I $,
  and for any $ x \in I $ and $ r \in R $,
  we have $ r \cdot x \in I $ and $ x \cdot r \in I $.
\end{definition}

\begin{example}
  \label{exmp:ideals}
  \begin{enumerate}
  \item
    Any semiring $ R $ admits the trivial ideals $ \{ 0 \} \subseteq R $
    and $ R \subseteq R $.

  \item
    The sets $ \{0, \omega \} \subseteq \{0, 1, \omega \} $
    and $ \Nat \setminus \{ 1 \} \subseteq \Nat $
    are ideals.

  \item
    \label{point:generating-ideals}
    Given a family of ideals $ I_j $,
    the intersection $ \bigcap_j I_j $ is again an ideal.
    Hence, given any subset $ X \subseteq R $ there exists a smallest ideal
    containing $ X $:
    There exists an ideal containing $ X $
    since $ R \subseteq R $ is an ideal itself.
    Now consider the set of all such ideals and take their intersection.
  \end{enumerate}
\end{example}

\begin{definition}
  A mode $ \mode m $
  is a triple $ (R_{\mode m},  \mathsf{Cont}( \mathsf{m} ) ,  \mathsf{Weak} (  \mathsf{m}  ) ) $
  where $ R_{\mode m} $ is a grade algebra,
  $  \mathsf{Cont}( \mathsf{m} )  \subseteq R_{\mode m} $ is an ideal
  and $  \mathsf{Weak} (  \mathsf{m}  )  \in \{\ms{true}, \ms{false} \} $.
\end{definition}

\begin{remark}
  \label{rem:mode-control-details}
  Our notion of modes is a generalization of the notion of modes given in
  Adjoint Logic~\cite{pruiksma_adjointlogic_2018}.
  There, a mode is equipped with two boolean values,
  indicating whether weakening and contraction are permitted.
  In the graded setting,
  we can control weakening and contraction more granularly
  by enabling them only for a subset of grades.
  Contraction is modeled by a set of mutually contractible grades.
  We will see below how the conditions that this set is an ideal arise
  naturally in the syntax.
  While it may appear that weakening is still controlled by a boolean,
  the situation is a bit more subtle:
  Most graded systems, including \grass{}, have a subsumption rule,
  which states that grades need only be upper bounds on the actual usage
  of variables.
  Thus, when weakening is enabled,
  unused variables can be introduced at any grade $ q \ge 0 $,
  by applying weakening followed by subsumption.
  We do not require $ 0 $ to be the bottom element of a grade algebra,
  so weakening is controlled by the predicate $  \mathsf{Weak} (  \mathsf{m}  )  $
  and the choice of preorder.
\end{remark}

\begin{example}
  The following modes will serve as running examples:
  $ \ms U = (\top, \top, \ms{true}) $,
  $ \ms R = (\top, \top, \ms{false}) $,
  $ \ms A = (\{0, 1\}, \{ 0 \}, \ms{true}) $
  and $ \ms L = (\Nat, \{ 0 \}, \ms{false}) $.
  We will see in Example \ref{exmp:non-graded-logics}
  how these modes can recover intuitionistic, relevant,
  affine and linear logic respectively.
\end{example}

\subsection{Restricted structural rules}

We can now describe the controlled structural of \grass{}.
Fix a mode $ \mode m = (R_{\mode m},  \mathsf{Cont}( \mathsf{m} ) ,  \mathsf{Weak} (  \mathsf{m}  ) ) $.
The new rules for contraction and weakening are as follows:
\begin{mathpar}
  \namedrules{grassS}{term}{cont,weak}.
\end{mathpar}
The weakening rule is straightforwardly modified:
It can be toggled on or off by the mode.
Contraction is only allowed when both grades are elements of $  \mathsf{Cont}( \mathsf{m} )  $.
The requirement that $  \mathsf{Cont}( \mathsf{m} )  $ is an ideal arises naturally as follows.
First, consider a judgment
$
   \rho  \grassSsym{,}  \grassSnt{q_{{\mathrm{1}}}}  \grassSsym{,}  \grassSnt{q_{{\mathrm{2}}}}  \grassSsym{,}  \grassSnt{q_{{\mathrm{3}}}}  \odot  \Gamma  \grassSsym{,}  \grassSmv{x_{{\mathrm{1}}}}  \grassSsym{:}  \grassSnt{A}  \grassSsym{,}  \grassSmv{x_{{\mathrm{2}}}}  \grassSsym{:}  \grassSnt{A}  \grassSsym{,}  \grassSmv{x_{{\mathrm{3}}}}  \grassSsym{:}  \grassSnt{A}  \vdash  \grassSnt{t}  :  \grassSnt{B} 
$
where $  \grassSnt{q_{{\mathrm{1}}}}  \grassSsym{,}  \grassSnt{q_{{\mathrm{2}}}}  \grassSsym{,}  \grassSnt{q_{{\mathrm{3}}}}  \in \mathsf{Cont}( \mathsf{m} )  $.
We want $  \mathsf{Cont}( \mathsf{m} )  $ to act as a set of mutually contractible grades,
so it should be possible to contract this into a term
$
   \rho  \grassSsym{,}  \grassSnt{q_{{\mathrm{1}}}}  \grassSsym{+}  \grassSnt{q_{{\mathrm{2}}}}  \grassSsym{+}  \grassSnt{q_{{\mathrm{3}}}}  \odot  \Gamma  \grassSsym{,}  \grassSmv{z}  \grassSsym{:}  \grassSnt{A}  \vdash  \grassSsym{[}  \grassSmv{z}  \grassSsym{/}  \grassSmv{x_{{\mathrm{1}}}}  \grassSsym{,}  \grassSmv{z}  \grassSsym{/}  \grassSmv{x_{{\mathrm{2}}}}  \grassSsym{,}  \grassSmv{z}  \grassSsym{/}  \grassSmv{x_{{\mathrm{3}}}}  \grassSsym{]}  \grassSnt{t}  :  \grassSnt{B} 
$.
Since the contraction rule only allows contracting two variables at a time,
we must pick an order:
We either contract $ x_1 $ with $ x_2 $ and then the result with $ x_3 $,
or $ x_2 $ and $ x_3 $ first and then the result with $ x_1 $.
The condition that $  \mathsf{Cont}( \mathsf{m} )  $ is closed under addition means that
both orderings are always possible,
and in fact the resulting derivations should be considered equal.
Similarly,
the condition that $   0   \in \mathsf{Cont}( \mathsf{m} )  $ ensures that
contraction interacts well with weakening in the sense that
for $  \grassSnt{q}  \in \mathsf{Cont}( \mathsf{m} )  $ we have the derivation
\[
  \inferrule*
  [Right = \rulename{grassS}{term}{cont}]
  {
    \inferrule*
    [Right = \rulename{grassS}{term}{weak}]
    {
       \rho  \grassSsym{,}  \grassSnt{q}  \odot  \Gamma  \grassSsym{,}  \grassSmv{x}  \grassSsym{:}  \grassSnt{A}  \vdash  \grassSnt{e}  :  \grassSnt{B} 
    }
    {
       \rho  \grassSsym{,}  \grassSnt{q}  \grassSsym{,}   0   \odot  \Gamma  \grassSsym{,}  \grassSmv{x}  \grassSsym{:}  \grassSnt{A}  \grassSsym{,}  \grassSmv{y}  \grassSsym{:}  \grassSnt{A}  \vdash  \grassSnt{e}  :  \grassSnt{B} 
    }
  }
  {
     \rho  \grassSsym{,}  \grassSnt{q}  \grassSsym{+}   0   \odot  \Gamma  \grassSsym{,}  \grassSmv{z}  \grassSsym{:}  \grassSnt{A}  \vdash  \grassSsym{[}  \grassSmv{z}  \grassSsym{/}  \grassSmv{x}  \grassSsym{]}  \grassSnt{e}  :  \grassSnt{B} 
  }
\]
and the resulting term should be considered equal to $ e $.
Finally, contraction should also interact well with substitution.
In graded type systems, substitution takes the form
\[
  \inferrule*
  {
     \rho  \grassSsym{,}  \grassSnt{r}  \odot  \Gamma  \grassSsym{,}  \grassSmv{x}  \grassSsym{:}  \grassSnt{A}  \vdash  \grassSnt{e}  :  \grassSnt{B} 
    \and
     \sigma  \odot  \Delta  \vdash  \grassSnt{t}  :  \grassSnt{A} 
  }
  {
     \rho  \grassSsym{,}   \grassSnt{r}  \cdot  \sigma   \odot  \Gamma  \grassSsym{,}  \Delta  \vdash  \grassSsym{[}  \grassSnt{t}  \grassSsym{/}  \grassSmv{x}  \grassSsym{]}  \grassSnt{e}  :  \grassSnt{B} 
  }.
\]
Given terms
$  \rho  \grassSsym{,}  \grassSnt{r}  \odot  \Gamma  \grassSsym{,}  \grassSmv{x}  \grassSsym{:}  \grassSnt{A}  \vdash  \grassSnt{e}  :  \grassSnt{B}  $
and
$  \sigma  \grassSsym{,}  \grassSnt{q_{{\mathrm{1}}}}  \grassSsym{,}  \grassSnt{q_{{\mathrm{2}}}}  \odot  \Delta  \grassSsym{,}  \grassSmv{z_{{\mathrm{1}}}}  \grassSsym{:}  \grassSnt{T}  \grassSsym{,}  \grassSmv{z_{{\mathrm{2}}}}  \grassSsym{:}  \grassSnt{T}  \vdash  \grassSnt{t}  :  \grassSnt{A}  $
with $  \grassSnt{q_{{\mathrm{1}}}}  \grassSsym{,}  \grassSnt{q_{{\mathrm{2}}}}  \in \mathsf{Cont}( \mathsf{m} )  $
we can apply contraction, obtaining
$
   \sigma  \grassSsym{,}  \grassSnt{q_{{\mathrm{1}}}}  \grassSsym{+}  \grassSnt{q_{{\mathrm{2}}}}  \odot  \Delta  \grassSsym{,}  \grassSmv{z}  \grassSsym{:}  \grassSnt{T}  \vdash  \grassSsym{[}  \grassSmv{z}  \grassSsym{/}  \grassSmv{z_{{\mathrm{1}}}}  \grassSsym{,}  \grassSmv{z}  \grassSsym{/}  \grassSmv{z_{{\mathrm{2}}}}  \grassSsym{]}  \grassSnt{t}  :  \grassSnt{A} 
$
then perform substitution.
The condition that $  \mathsf{Cont}( \mathsf{m} )  $ is an ideal ensures that the other ordering
is possible too:
If we substitute first, obtaining
$
   \rho  \grassSsym{,}   \grassSnt{r}  \cdot  \sigma   \grassSsym{,}   \grassSnt{r}  \cdot  \grassSnt{q_{{\mathrm{1}}}}   \grassSsym{,}   \grassSnt{r}  \cdot  \grassSnt{q_{{\mathrm{2}}}}   \odot  \Gamma  \grassSsym{,}  \Delta  \grassSsym{,}  \grassSmv{z_{{\mathrm{1}}}}  \grassSsym{:}  \grassSnt{T}  \grassSsym{,}  \grassSmv{z_{{\mathrm{2}}}}  \grassSsym{:}  \grassSnt{T}  \vdash  \grassSsym{[}  \grassSnt{t}  \grassSsym{/}  \grassSmv{x}  \grassSsym{]}  \grassSnt{e}  :  \grassSnt{B} 
$
then we may still apply contraction.
Note that both orders of these operations result in the same term.

\begin{example}
  \label{exmp:non-graded-logics}
  We can recover non-graded logics by using appropriate modes.
  \begin{enumerate}
  \item
    Using the mode $ \ms U = (\top, \top, \ms{true}) $
    we recover intuitionistic logic:
    Contraction and weakening are permitted.

  \item
    Using mode $ \ms R = (\top, \top, \ms{false}) $
    recovers \emph{relevant logic}:
    There is only one grade that may be contracted with itself,
    so variables may be reused arbitrarily.
    However, weakening is disabled, so unused variables are not allowed.

  \item
    Using the mode $ \ms A = (\{ 0, 1 \}, \{ 0 \}, \ms{true}) $
    with ordering $ 0 \le 1 $, we recover \emph{affine logic}.
    Variables graded $ 1 $ may be unused by applying a combination of the weakening
    and subsumption rules.
    Contraction is only allowed for unused variables graded $ 0 $,
    so reuse is impossible.

  \item
    Using the mode $ (\{0, 1 \}, \{ 0 \}, \ms{false}) $ we almost recover linear logic.
    Variables graded $ 1 $ are guaranteed to have linear usage.
    The system we obtain here is not exactly linear logic (hence ``almost''),
    since variables graded $ 0 $ can still be introduced explicitly using
    the graded modality $ \square_0 A $,
    or by applying the elimination rule for the unit type.
    (The unit type has no computational content so its terms may be eliminated
    at any grade.)

  \item
    Similarly, grading by the mode $ \ms L = (\Nat, \{ 0 \}, \ms{false}) $
    with discrete ordering
    tracks the number of free occurences of variables in a term exactly.
    Since contraction and weakening are disabled,
    variables graded by $ n \neq 1 $ must be introduced explicitly
    via the graded modality $ \square_n A $,
    or the unit type's elimination rule.
    Thus, we recover another variant on linear logic.
  \end{enumerate}
\end{example}

\begin{example}
  \label{exmp:file-handle-mode}
  Using the mode $ (R = \{0, 1, \omega\}, \{ 0, \omega \}, \ms{true}) $,
  the grades mean unused $ (0) $, linear $ (1) $, and unrestricted~$ (\omega) $.
  Furthermore, two variables graded $ 1 $ may not be contracted into one variable
  graded $ 1 + 1 = \omega $.
  This is useful, for example, when dealing with file handles:
  Two file handles used linearly cannot be replaced by one file handle with
  unrestricted use.
  Similar considerations apply to the mode
  $ (\Nat, \Nat \setminus \{ 1 \}, \ms{true}) $:
  Contraction in only allowed for variables which are known to be used non-linearly.
\end{example}

\subsection{Comparing modes}
\label{sec:mode-morphisms}

\begin{definition}
  A morphism of semirings $ R \to S $ is a function
  $ f \from R \to S $ such that
  $ f(0) = 0 $, $ f(1) = 1 $, and for all $ x, y \in R $,
  $ f(x + y) = f(x) + f(y) $ and $ f(x \cdot y) = f(x) \cdot f(y) $.
  A morphism of grade algebras $ f \from R \to S $ is a morphism of
  semirings which is also monotone, i.e.\@~$ x \le y $ implies $ f(x) \le f(y) $.
  A morphism of modes
  $ \mode m \to \mode n $ is a morphism
  $ f \from R_{\mode m} \to R_{\mode n} $ of the underlying grade algebras
  such that for each $ x \in \Cont(\mode m) $, we have $ f(x) \in \Cont(\mode n) $
  and such that the boolean formula $ \Weak(\mode m) \to \Weak(\mode n) $ is true.
  This definition yields a category of modes, denoted $ \Mode $.
\end{definition}

A morphism of modes $ \phi \from \mode m \to \mode n $ implies
that $ \mode n $ has ``more structure'' than $ \mode m $ in an appropriate way.
We give some intuition for this:
Given a derivable judgment
$  \rho  \odot  \Gamma  \vdash  \grassSnt{t}  :  \grassSnt{A}  $ under mode $ \mode m $
there exists a corresponding judgment
$ \phi(\rho) \at \Gamma \proves t : A $
which is derivable under mode $ \mode n $.
Here, $ \phi(\rho) $ is the result of applying $ \phi $ to each element of $ \rho $.
In other words, there is a translation function.
This translation is not quite the identity on types and terms,
since some types and terms feature grade annotations,
which also need to be transported along $ \phi $.
The proof that this translation is sound is straightforward by induction,
since morphisms of modes preserve all structure carried by a mode.
In particular, each application of the weakening or contraction rules
at mode $ \mode m $ can be mirrored at $ \mode n $,
which we take to mean that $ \mode n $ has more structure than $ \mode m $.
This intuition has two interesing special cases:
The category $ \Mode $ has a terminal object $ \ms U $
and an initial object $ \ms L $.
We already saw in Example \ref{exmp:non-graded-logics}
that these modes recover intuitionistic and linear logic respectively.
Therefore it makes sense to think of the terminal and initial modes as
having maximal and minimal structure respectively.

\section{Syntax of \grass{}}
\label{sec:syntax}
  In this section we explain the syntax and typing rules of \grass{}.
Our system is parametrized by a family of modes,
with each type belonging to one mode.
Like with Adjoint Logic~\cite{pruiksma_adjointlogic_2018},
the modes are arranged in a preorder,
with modes higher in the order admitting more structural rules than those
lower in the order, as discussed in \S\ref{sec:mode-morphisms}.
First, we present several definitions used throughout the paper.

\begin{definition}
  \label{defn:parametrizing-data}
  \grass{} is parametrized by a functor
  $ S \to \Mode $, where $ S $ is some preordered set.
  We usually elide mention of this functor, instead thinking of its image
  as a preordered set of modes $ \mathsf{m}  \grasssym{,}  \mathsf{n}  \grasssym{,}  \mathsf{l}  \grasssym{,}  \mathellipsis $ together with a coherent
  system of morphisms $ \phi_{\mode m, \mode n} \from \mode m \to \mode n $,
  whenever $ \mathsf{m}  \le  \mathsf{n} $.
  We fix such data for the remainder of this text.
\end{definition}

\begin{definition}
  Grades are denoted by the letters $ \grassnt{q}  \grasssym{,}  \grassnt{r}  \grasssym{,}  \grassnt{s} $.
  They may be drawn from any of the grade algebras $ R_{\mode m} $.
  \emph{Grade vectors} are finite lists of grades, denoted by the letters
  $ \rho $, $ \sigma $, $ \tau $.
  Grade vectors are \emph{heterogenous}, i.e.\@ entries of one grade vector
  may be drawn from more than one grade algebra.
\end{definition}

\begin{definition}
  \label{defn:algebra-multiplication}
  Let $ \mode m \le \mode n $ be modes
  and let $ q \in R_{\mode m} $ and $ r \in R_{\mode n} $.
  We define
  $  \grassnt{q}   \grassnt{r}  = \phi_{\mode m, \mode n} (q) \cdot r \in R_{\mode n} $.
  We extend this notation to scalar multiplication:
  If $ q \in R_{\mode m} $ and $ \rho = \grasssym{(}  \grassnt{r_{{\mathrm{1}}}}  \grasssym{,}  \mathellipsis  \grasssym{,}  \grassnt{r_{\grassmv{k}}}  \grasssym{)} $
  with $ r_i \in R_{\mode m_i} $ and $ \mathsf{m}_{\grassmv{i}}  \ge  \mathsf{m} $ for $ i \in \{1, \mathellipsis, k\} $,
  then we write
  $  \grassnt{q}   \rho  = \grasssym{(}   \grassnt{q}   \grassnt{r_{{\mathrm{1}}}}   \grasssym{,}  \mathellipsis  \grasssym{,}   \grassnt{q}   \grassnt{r_{\grassmv{k}}}   \grasssym{)}
    = (
    \phi_{\mode m,\mode m_1}(q) \cdot r_1,
    \mathellipsis,
    \phi_{\mode m,\mode m_k}(q) \cdot r_k) $.
\end{definition}

\begin{definition}
  Types are denoted $ A, B, C, T $.
  Each type in \grass{} belongs to a unique mode and
  we write $  \grassnt{A}  \in \mathsf{Type}( \mathsf{m} )  $
  to indicate that the type $ A $ is well-formed and belongs to mode $ \mode m $.
  In this case, variables of type $ A $ are graded by grades drawn from $ R_{\mode m} $,
  and structural rules may be applied to them as prescribed by $ \Weak(\mode m) $
  and $ \Cont(\mode m) $.
  The rules for well-formed types are given in Figure \ref{fig:types-well-formed}.
\end{definition}

\begin{figure}
  \begin{mathpar}
    \small
    \namedrules{grass}{type}{unit,pair,arrow,sum,drop,raise}
  \end{mathpar}
  \caption{\grass{} rules for well-formed types}
  \label{fig:types-well-formed}
\end{figure}

\begin{figure*}
  \begin{mathpar}
    \small
    \namedrules{grass}{term}{var,weak,cont,sub,%
      unitI,unitE,arrowI,arrowE,pairI,pairE,sumIL,sumIR,sumE,%
      dropI,dropE,raiseI,raiseE}
  \end{mathpar}
  \caption{\grass{} typing rules}
  \label{fig:full-typing-rules}
\end{figure*}
\begin{figure*}
  \small
  \begin{mathpar}
    \namedrules{grass}{beta}{unit,pair,arrow,sumL,sumR,raise,drop}
  \end{mathpar}
  \begin{mathpar}
    \namedrules{grass}{eta}{unit,pair,arrow,sum,raise,drop}
  \end{mathpar}
  \caption{\grass{} $ \beta\eta $-conversions}
  \label{fig:grass-reductions}
\end{figure*}

\begin{definition}
  Typing judgments in \grass{} have the form
  $
     \rho  \mid  \mathsf{M}   \odot   \Gamma  \vdash_{ \mathsf{m} }  \grassnt{t}  \colon  \grassnt{T} 
  $.
  The roles of $ \Gamma $, $ \grassnt{t} $, and $ \grassnt{T} $ are as before.
  There are two new components in the judgment:
  The annotation $ \mathsf{m} $ on the turnstile indicates that $  \grassnt{T}  \in \mathsf{Type}( \mathsf{m} )  $,
  and we say that this typing judgment is \emph{made at mode $ \mode m $}.
  The other new component is $ \mathsf{M} $, which is a list of modes,
  also called a \emph{mode vector}.
  If $ \rho = \grasssym{(}  \grassnt{r_{{\mathrm{1}}}}  \grasssym{,}  \mathellipsis  \grasssym{,}  \grassnt{r_{\grassmv{k}}}  \grasssym{)} $
  and $ \mathsf{M} = \grasssym{(}  \mathsf{m}_{{\mathrm{1}}}  \grasssym{,}  \mathellipsis  \grasssym{,}  \mathsf{m}_{\grassmv{k}}  \grasssym{)} $
  and $ \Gamma = \grassmv{x_{{\mathrm{1}}}}  \grasssym{:}  \grassnt{A_{{\mathrm{1}}}}  \grasssym{,}  \mathellipsis  \grasssym{,}  \grassmv{x_{\grassmv{k}}}  \grasssym{:}  \grassnt{A_{\grassmv{k}}} $,
  then the above judgment indicates that $  \grassnt{A_{\grassmv{i}}}  \in \mathsf{Type}( \mathsf{m}_{\grassmv{i}} )  $
  and that the variable $ x_i $ is used with capabilities indicated by $ r_i $.
\end{definition}

\begin{definition}
  Typing contexts are denoted by $ \Gamma $, $ \Delta $.
  When $ \Gamma $ and $ \Delta $ are contexts, their concatenation is written
  $ \Gamma  \grasssym{,}  \Delta $ and is only well-defined when the contexts $ \Gamma $
  and $ \Delta $ have no variables in common.
  Whenever we use context concatenation, we implicitly assume that this condition
  is satisfied.
\end{definition}

In any derivable typing judgment $  \rho  \mid  \mathsf{M}   \odot   \Gamma  \vdash_{ \mathsf{m} }  \grassnt{t}  \colon  \grassnt{T}  $,
each mode $ \mathsf{n} $ occuring in $ \mathsf{M} $ must satisfy $ \mathsf{n}  \ge  \mathsf{m} $,
which we write concisely as $ \mathsf{M}  \ge  \mathsf{m} $.
A similar requirement is present in Linear-non-Linear Logic~\cite{benton_mixed_1995},
where linear terms may depend on non-linear variables, but not vice-versa,
and was also observed by Pruiksma et al.~\cite{pruiksma_adjointlogic_2018},
where it is called \emph{independence}.
The typing rules of \grass{} are set up to ensure that independence
holds in all derivable judgments.
In \grass{}, independence also ensures that all scalar multiplications
occuring in the typing rules are well-defined
(cf.\@ Definition \ref{defn:algebra-multiplication}).
We discuss the typing rules in detail in the following paragraphs.

We begin with the structural rules.
The grade $  1  $ corresponds to linear usage,
and is thus used in the variable rule (\rulename{grass}{term}{var}).
The rules for weakening and contraction are as before,
only modified to account for the presence of different modes.
Contraction (rule \rulename{grass}{term}{cont})
is only allowed between grades of the same mode.
In the weakening rule (\rulename{grass}{term}{weak})
the condition $ \mathsf{n}  \le  \mathsf{m} $ ensures that independence holds in the conclusion,
if it holds in the assumption.
We assume an implicit exchange rule,
which allows reordering of contexts.
When this occurs, the grade and mode vectors must be reordered accordingly.

The type $  \operatorname{\downarrow} ^{ \grassnt{q} }_{ \mathsf{n}  \le  \mathsf{m} }  \grassnt{A}  $ is defined for $  \grassnt{A}  \in \mathsf{Type}( \mathsf{m} )  $,
$ q \in R_{\mode m} $ and $ \mathsf{n}  \le  \mathsf{m} $ (rule \rulename{grass}{type}{drop}).
Terms of this type behave like terms of type $ A $,
with the added constraint that they must be used at grade $ q $.
Furthermore, the type $  \operatorname{\downarrow} ^{ \grassnt{q} }_{ \mathsf{n}  \le  \mathsf{m} }  \grassnt{A}  $ belongs to mode $ \mode n $ and is therefore
subject to the structural rules imposed by $ \mode n $.
The introduction rule (\rulename{grass}{term}{dropI})
promotes a term $ t $ of type $ A $ to a term $  \operatorname{\downarrow} ^{ \grassnt{q} }_{ \mathsf{n}  \le  \mathsf{m} }  \grassnt{t}  $ of
type $  \operatorname{\downarrow} ^{ \grassnt{q} }_{ \mathsf{n}  \le  \mathsf{m} }  \grassnt{A}  $.
The resources used to produce~$ t $ are multiplied by $ q $ during this promotion,
reflecting the $ q $-fold reuse later.
The elimination rule (\rulename{grass}{term}{dropE})
uses a let binding $  \mathbin{\mathsf{let} } _{@  \grassnt{q} }   \operatorname{\downarrow} _{ \mathsf{n}  \le  \mathsf{m} }  \grassmv{x}   =  \grassnt{e}   \mathbin{\mathsf{in} }   \grassnt{t}  $
to extract a term of type $ A $ from $ e $,
and make it available to the body~$ t $,
which then must use that term at grade $ q $.
This type combines two distinct feature sets from existing literature:
First, some graded type theories
(e.g.~\cite{%
  abel_unified_2020,%
  brunel_core_2014,%
  choudhury_graded_2021,%
  ghica_bounded_2014,%
  orchard_quantitative_2019%
})
include a graded modality, commonly written $ \square_q A $.
We can recover such a type in \grass{} by setting
$ \square_q A :=  \operatorname{\downarrow} ^{ \grassnt{q} }_{ \mathsf{m}  \le  \mathsf{m} }  \grassnt{A}  $.
Second, $  \operatorname{\downarrow} ^{ \grassnt{q} }_{ \mathsf{n}  \le  \mathsf{m} }  \grassnt{A}  $ is
an analogue of $ \ms F A $ from LNL \cite{benton_mixed_1995},
and $ \ms{Grd}_q A $ from \textsf{mGL} \cite{vollmer_mixed_2025},
since these constructs move types to a mode with fewer structural rules.
Our notation mirrors that of Pruiksma et al.~\cite{pruiksma_adjointlogic_2018}.

The type $  \operatorname{\uparrow} _{ \mathsf{n}  \le  \mathsf{m} }  \grassnt{A}  $ is defined for $  \grassnt{A}  \in \mathsf{Type}( \mathsf{n} )  $
and $ \mathsf{n}  \le  \mathsf{m} $ (rule \rulename{grass}{type}{raise}).
The introduction rule (\rulename{grass}{term}{raiseI}) constructs
a term $  \operatorname{\uparrow} _{ \mathsf{n}  \le  \mathsf{m} }  \grassnt{e}  :  \operatorname{\uparrow} _{ \mathsf{n}  \le  \mathsf{m} }  \grassnt{A}  $ from a term $ e : A $.
The produced term and its type belong to mode $ \mode m $ and thus have fewer
structural restrictions imposed upon them.
The elimination rule (\rulename{grass}{term}{raiseE}) allows us to extract
a term of type $ A $ from a term of type $  \operatorname{\uparrow} _{ \mathsf{n}  \le  \mathsf{m} }  \grassnt{A}  $ directly.
The type $  \operatorname{\uparrow} _{ \mathsf{n}  \le  \mathsf{m} }  \grassnt{A}  $ is the analogue of $ \ms G A $ from LNL
and $ \ms{Lin} $ from \textsf{mGL}.

\grass{} includes a subsumption rule \rulename{grass}{term}{sub}.
This rule means that a grade vector $ \rho $
in a judgment $  \rho  \mid  \mathsf{M}   \odot   \Gamma  \vdash_{ \mathsf{m} }  \grassnt{t}  \colon  \grassnt{A}  $ need only be an upper bound on the
actual usage of the variables in $ \Gamma $ within the term $ \grassnt{e} $.
The relation $  \rho  \le  \sigma  $ is defined inductively by
\begin{mathpar}
  \inferrule*{
  }{
    \emptyset  \le  \emptyset
  }
  \and
  \inferrule*{
     \rho  \le  \sigma 
    \and
     \grassnt{r}  \in R_{ \mathsf{m} } 
    \and
     \grassnt{q}  \in R_{ \mathsf{m} } 
    \and
     \grassnt{r}  \le  \grassnt{q} 
  }
  {
     \grasssym{(}  \rho  \grasssym{,}  \grassnt{r}  \grasssym{)}  \le  \grasssym{(}  \sigma  \grasssym{,}  \grassnt{q}  \grasssym{)} 
  }.
\end{mathpar}
When the preorder is discrete, this rule is effectively disabled.

For each mode $ \mode m $, we have a unit type $  \mathbf{I}_ \mathsf{m}  $
(rule \rulename{grass}{type}{unit})
which admits a closed term $  \star_ \mathsf{m}  $
that can be produced without using any resources
(rule \rulename{grass}{term}{unitI}).
Terms of type $  \mathbf{I}_ \mathsf{m}  $ hold no computational information,
and may thus be eliminated (rule \rulename{grass}{term}{unitE}) at any grade $ q $
by the pattern matching construct $  \mathbin{\mathsf{let} } _{@  \grassnt{q} }   \star_ \mathsf{m}   =  \grassnt{e}   \mathbin{\mathsf{in} }   \grassnt{t}  $.

The tensor product type $  \grassnt{A_{{\mathrm{1}}}}  \otimes  \grassnt{A_{{\mathrm{2}}}}  $ is defined only when the types $ A_1 $
and $ A_2 $ belong to the same mode.
Its introduction rule (\rulename{grass}{term}{pairI}) is standard,
allowing the construction of a pair $ \grasssym{(}  \grassnt{t_{{\mathrm{1}}}}  \grasssym{,}  \grassnt{t_{{\mathrm{2}}}}  \grasssym{)} $ from terms $ \grassnt{t_{{\mathrm{1}}}} $
and $ \grassnt{t_{{\mathrm{2}}}} $.
The pair elimination rule (\rulename{grass}{term}{pairE}) functions by pattern
matching.
It asserts that using a pair at grade $ q $ is the same thing as using each of its
components separately at grade $ q $.
Since term reuse is modeled by multiplication,
the grades used to produce the pair are multiplied by $ q $
as part of the elimination rule.

The function type $  \grassnt{A} ^{ \grassnt{q} : \mathsf{m} } \multimap  \grassnt{B}  $ is defined when $  \grassnt{A}  \in \mathsf{Type}( \mathsf{m} )  $
and $  \grassnt{B}  \in \mathsf{Type}( \mathsf{n} )  $ with $ \mathsf{m}  \ge  \mathsf{n} $.
The requirement that $ \mathsf{m}  \ge  \mathsf{n} $ is an instance of independence.
The grade annotation $ q $ indicates how functions of this type use their parameters.
Functions are introduced (rule \rulename{grass}{term}{arrowI})
using $ \lambda $-abstraction.
The grade annotation $ q $ corresponds to the
the use of the abstacted variable in the function body.
Elimination (\rulename{grass}{term}{arrowE}) is done by function application.
Since grade multiplication corresponds to reuse,
the grades used to produce the parameter $ e $ are multiplied by $ q $
during application.

Like the product, the sum type $  \grassnt{A_{{\mathrm{1}}}}  \oplus  \grassnt{A_{{\mathrm{2}}}}  $ is defined between types
of the same mode. Its introduction forms (\rulenames{grass}{term}{sumIL,sumIR})
are standard, injecting a term $ e $ of type $ A_1 $ (resp. $ A_2 $) into the sum.
These injections do not affect the grades used to produce $ e $.
Elimination occurs by cases: to use a term of type $  \grassnt{A_{{\mathrm{1}}}}  \oplus  \grassnt{A_{{\mathrm{2}}}}  $ at grade $ q $,
we must specify two branches,
which use a term of $ A_1 $ (resp. $ A_2 $) at grade $ q $.

\begin{example}
  Let $ \mathsf{m}_{{\mathrm{1}}}  \grasssym{,}  \mathsf{m}_{{\mathrm{2}}} $ be two modes, and let $  \mathsf{L}   $ be the initial mode
  as discussed in \S\ref{sec:mode-morphisms}.
  Instantiate \grass{} with modes $ \mode m_1 >  \mathsf{L}  < \mode m_2 $.
  Write $ \square^{q}_{\mode m} A =  \operatorname{\downarrow} ^{ \grassnt{q} }_{  \mathsf{L}   \le  \mathsf{m} }  \grasssym{(}   \operatorname{\uparrow} _{  \mathsf{L}   \le  \mathsf{m} }  \grassnt{A}   \grasssym{)}  $.
  Then $  \mathsf{L}  $ behaves like linear logic,
  with two graded modalities $ \square^q_{\mode m_1} A $ and $ \square^q_{\mode m_2} A $
  which can account for two distinct notions of resource usage.
  This is a generalization of \textsf{mGL} allowing grades from two different
  grade algebras.
\end{example}

\begin{example}
  Using the modes $ \ms L < \ms U $
  of Example \ref{exmp:non-graded-logics},
  we obtain LNL \cite{benton_mixed_1995}.
  If we also use the modes~$ \ms A  $ and~$ \ms R $,
  we obtain a combined system in the style of Adjoint Logic
  \cite{pruiksma_adjointlogic_2018}.
\end{example}

\begin{theorem}[Substitution]
  If $  \rho  \grasssym{,}  \grassnt{r}  \mid  \mathsf{M}  \grasssym{,}  \mathsf{m}   \odot   \Gamma  \grasssym{,}  \grassmv{x}  \grasssym{:}  \grassnt{A}  \vdash_{ \mathsf{n} }  \grassnt{e}  \colon  \grassnt{B}  $
  and
  $  \sigma  \mid  \mathsf{N}   \odot   \Delta  \vdash_{ \mathsf{m} }  \grassnt{t}  \colon  \grassnt{A}  $,
  then
  \[
     \rho  \grasssym{,}   \grassnt{r}   \sigma   \mid  \mathsf{M}  \grasssym{,}  \mathsf{N}   \odot   \Gamma  \grasssym{,}  \Delta  \vdash_{ \mathsf{n} }  \grasssym{[}  \grassnt{t}  \grasssym{/}  \grassmv{x}  \grasssym{]}  \grassnt{e}  \colon  \grassnt{B} 
  \]
\end{theorem}

\begin{proof}
  See
  \ifextended
    \S \ref{asec:multi-mode-subst} of the appendix.
  \else
    the extended version \cite{extendedVersion}.
  \fi
\end{proof}

We include a set of $ \beta $- and $ \eta $-conversion rules
which are given in Figure \ref{fig:grass-reductions}.
The $ \beta $-rules are mostly standard:
Pattern matching eliminators reduce when the scrutinee of the match
is a constructor for the respective type.
In the pattern matching form for the unit type,
the grade annotation $ q $ gets ignored since the term $  \star_ \mathsf{m}  $
carries no computational information.
Function applications reduce by substitution into the function body,
when the function in question is a lambda.
The type $  \operatorname{\uparrow} _{ \mathsf{m}  \le  \mathsf{n} }  \grassnt{A}  $ is a negative type with an explicit destructor
$  \operatorname{\uparrow}^{-1} _{ \mathsf{m}  \le  \mathsf{n} }  \grassnt{e}  $, rather than a pattern matching one,
which acts as a left inverse to the constructor $  \operatorname{\uparrow} _{ \mathsf{m}  \le  \mathsf{n} }  \grassnt{e}  $.
The $ \eta $-rules are similarly standard: eliminating then reconstructing a term
recovers that term.
These conversions preserve typing:

\begin{theorem}[Type preservation]
  The $ \beta $- and $ \eta $-conversions preserve typing and grading, i.e.
  \begin{enumerate}
  \item
    Whenever $  \rho  \mid  \mathsf{M}  \odot  \Gamma  \vdash_{ \mathsf{m} }  \grassnt{t_{{\mathrm{1}}}}  \equiv_\beta  \grassnt{t_{{\mathrm{2}}}}  \colon  \grassnt{A}  $,
    then $  \rho  \mid  \mathsf{M}   \odot   \Gamma  \vdash_{ \mathsf{m} }  \grassnt{t_{{\mathrm{1}}}}  \colon  \grassnt{A}  $,
    and $  \rho  \mid  \mathsf{M}   \odot   \Gamma  \vdash_{ \mathsf{m} }  \grassnt{t_{{\mathrm{2}}}}  \colon  \grassnt{A}  $.
  \item
    Whenever $  \rho  \mid  \mathsf{M}  \odot  \Gamma  \vdash_{ \mathsf{m} }  \grassnt{t_{{\mathrm{1}}}}  \equiv_\eta  \grassnt{t_{{\mathrm{2}}}}  \colon  \grassnt{A}  $,
    then $  \rho  \mid  \mathsf{M}   \odot   \Gamma  \vdash_{ \mathsf{m} }  \grassnt{t_{{\mathrm{1}}}}  \colon  \grassnt{A}  $,
    and $  \rho  \mid  \mathsf{M}   \odot   \Gamma  \vdash_{ \mathsf{m} }  \grassnt{t_{{\mathrm{2}}}}  \colon  \grassnt{A}  $.
  \end{enumerate}
\end{theorem}

\begin{proof}
  See
  \ifextended
    \S\S \ref{asec:multi-mode-beta}, \ref{asec:multi-mode-eta} of the appendix.
  \else
    the extended version \cite{extendedVersion}.
  \fi
\end{proof}

\section{Categorical semantics}
\label{sec:cat-semantics}
  We present our categorical semantics.
The structure of the categorical models mirrors that of the syntax:
\grass{} is composed out of several graded logics, one for each mode,
tied together by a system of modalities $  \operatorname{\downarrow} ^{ \grassnt{q} }_{ \mathsf{n}  \le  \mathsf{m} }  \grassnt{A}  $
and $  \operatorname{\uparrow} _{ \mathsf{m}  \le  \mathsf{n} }  \grassnt{A}  $ mediating between them.
Similiarly, a model for \grass{} consists of several models for ordinary
graded logic, tied together by an appropriate system of morphisms,
here realized as adjunctions with some additional structure.
To appropriately capture the modified weakening and contraction rules,
we must make some changes relative to the existing literature.
Because of this, we set up all necessary terminology from scratch.
The categorical semantics presented here do not include the sum type;
Adding the necessary structure to model it is straightforward.
The following definitions are adapted from Katsumata~\cite{katsumata_double_2018}.

\begin{definition}
  Let $ \mathsf{m} = (R_ \mathsf{m},  \mathsf{Cont}( \mathsf{m} ) ,  \mathsf{Weak} (  \mathsf{m}  ) ) $ be a mode.
  The preorder on $ R_{\mode m} $ equips the sets $ R_{\mode m} $
  and $  \mathsf{Cont}( \mathsf{m} )  $ with categorical structures.
  We write $ (R_{\mode m}, \cdot, 1) $ for the symmetric monoidal structure
  on $ R_{\mode m} $ whose tensor product is multiplication
  and tensor unit is $ 1 \in R_{\mode m} $.
  Similarly, we write $ (  \mathsf{Cont}( \mathsf{m} ) , +, 0) $ for the symmetric monoidal
  structure on $  \mathsf{Cont}( \mathsf{m} )  $ whose tensor product is addition
  and tensor unit is $ 0 \in  \mathsf{Cont}( \mathsf{m} )  $.
\end{definition}

\begin{definition}
  Let $ (\mc C, \ox, I) $ be a symmetric monoidal category.
  Write $ \SMC(\mc C, \mc C) $ for the category of strong symmetric monoidal
  endofunctors on $ \mc C $ and monoidal natural transformations between them.
  This category can be equipped with a monoidal structure in two ways:
  \begin{enumerate}
  \item
    Pointwise:
    For $ F, G \in \SMC(\mc C, \mc C) $,
    their tensor product is
    $ (F \hox G)(X) = F(X) \ox G(X) $.
    The tensor unit is the constant functor sending everything to $ I \in \mc C $.
    We denote this structure $ (\SMC(\mc C, \mc C), \hox, \hat I) $.

  \item
    By composition:
    The tensor product is composition of functors,
    and the tensor unit is the identity functor.
    We denote this structure
    $ (\SMC(\mc C, \mc C), \circ, \id) $.
  \end{enumerate}
\end{definition}

\begin{definition}
  Let $ \mc C $ be a symmetric monoidal category,
  and let $ \mode m = (R_{\mode m},  \mathsf{Cont}( \mathsf{m} ) ,  \mathsf{Weak} (  \mathsf{m}  ) ) $ be a mode.
  An \emph{exponential action} (of $ \mode m $ on $ \mc C $) is a functor
  $ D \from R^{\op}_{\mode m} \to \SMC(\mc C, \mc C) $
  equipped with the following additional structure.
  \begin{enumerate}
  \item
    The structure of a colax monoidal functor with signature
    $ (R_{\mode m}^\op, \cdot, 1) \to (\SMC(\mc C, \mc C), \circ, \id) $
    on $ D $, i.e. natural transformations
    $\delta_{r, q} \from D(r \cdot q) \to D(r) \circ D(q)$
    and $\epsilon \from D(1) \to \id$.

  \item
    For $ r, q \in  \mathsf{Cont}( \mathsf{m} )  $, a natural transformation
    $ c_{r, q} \from D(r + q) \to D(r) \hox D(q) $.
    Making the following diagram commute.
    \[
      \begin{tikzcd}
        D((r_1 + r_2) + r_3)
        \arrow[r, "c_{r_1 + r_2 , r_3}"]
        \arrow[d, equals]
        &
        D(r_1 + r_2) \hox D(r_3)
        \arrow[r, "c_{r_1, r_2} \hox \id"]
        &
        (D(r_1) \hox D(r_2)) \hox D(r_3)
        \arrow[d, "\sim"]
        \\
        D(r_1 + (r_2 + r_3))
        \arrow[r, "c_{r_1 , r_2 + r_3}"]
        &
        D(r_1) \hox D(r_2 + r_3)
        \arrow[r, "\id \hox c_{r_2, r_3}"]
        &
        D(r_1) \hox (D(r_2) \hox D(r_3))
      \end{tikzcd}
    \]
    Furthermore, if $  \mathsf{Weak} (  \mathsf{m}  )  $ is true, a morphism $ w \from D(0) \to \hat I $.
    In this case, the transformations $ c_{r, q} $ and~$ w $ are required
    to assemble into a colax monoidal structure on the functor
    \[
      r \mapsto D(r) \from (  \mathsf{Cont}( \mathsf{m} ) ^{\op}, +, 0) \to \SMC(\mc C, \mc C).
    \]
    Note that the requirement that this functor is colax monoidal
    already includes the above diagram.
    We include that diagram explicitly, since we cannot require that this functor
    be colax monoidal in the absence of the morphism $ w $.
  \end{enumerate}
  These structures are related by the diagrams of Figure
  \ref{fig:graded-comonad-coherence},
  which are required to commute whenever all morphisms in them are defined.
  For the first two diagrams this means that $  \grassnt{r_{{\mathrm{1}}}}  \grasssym{,}  \grassnt{r_{{\mathrm{2}}}}  \in \mathsf{Cont}( \mathsf{m} )  $
  and for the last two diagrams this means that $  \mathsf{Weak} (  \mathsf{m}  )  $ is true.
  We write $\tau_{r, X, Y} \from D(r)(X) \ox D(r)(Y) \to D(r)(X \ox Y)$ and $\iota_{r} \from I \to D(r)(I)$
  for the strong monoidal structure on $ D(r) $.
  Furthermore, the diagrams in Figure \ref{fig:graded-comonad-coherence}
  use the infix notation $ D(r)(X) =: r \at X $,
  which we will use for the remainder of this text.
\end{definition}

\begin{example}
  \label{exmp:single-mode-models}
  We consider exponential actions of the modes given in
  Example \ref{exmp:non-graded-logics}.
  Let $ \mc C $ be an SMCC.

  \begin{enumerate}
  \item
    Any SMCC admits an exponential action of the initial mode $ \ms L $
    by setting $ n \at X = X^{\ox n} $.
    Similarly, any SMCC admits an exponential action of the mode
    $ (\{ 0, 1 \}, \{ 0 \}, \ms{false}) $, by setting $ 0 \at X = I $
    and $ 1 \at X = X $.

  \item
    An exponential action of the terminal mode $ \ms U $
    is equivalent to specifying the structure of a Linear Category
    \cite{bierman_intuitionistic_2000} on $ \mc C $.
    In other words, an SMCC with an exponential action of the terminal mode
    is precisely a model of Intuitionistic Linear Logic, with
    the exponential action modelling the of-course modality.
    If we choose $ 1 \at (-) = \id $,
    the action equips $ \mc C $ with ``projections and diagonals''
    and hence $ \mc C $ is a cartesian closed category.%
    \footnote{%
      This result is folklore;
      A proof is given in Melliès \cite{mellies_categorical_2009}.
    }

  \item
    An exponential action of the mode $ \ms R $ on $ \mc C $
    is a functor $ D $ together with diagonal maps $ DX \to DX \ox DX $.
    If we fix $ D = \id $,
    then $ \mc C $ becomes equipped with the structure of a
    \emph{relevant category}~\cite{petric_coherence_2002}.
    Conversely, any relevant category admits an exponential action of $ \ms R $
    by setting $ 1 \at (-) = \id $.

  \item
    If the tensor unit $ I $ is the terminal object of $ \mc C $,
    then $ \mc C $ admits an exponential action of the mode $ \ms A $
    by setting $ 1 \at (-) = \id $ and $ 0 \at (-) = \hat I $.
    Conversely, an action of $ \ms A $ on $ \mc C $ with $ 1 \at (-) = \id $
    and $ 0 \at (-) = \hat I $ exhibits the tensor unit as the terminal object
    of $ \mc C $.
  \end{enumerate}
  All in all, we see that known models of intuitionistic, linear, relevant and affine
  logic admit canonical exponential actions of the modes $ \ms U, \ms L, \ms R $
  and $ \ms A $ respectively.
\end{example}

\begin{figure*}
  \begin{mathpar}
    \small
    \begin{tikzcd}
      (q \cdot r_1 + q \cdot r_2) \at X
      \arrow[d, equals]
      \arrow[r, "c_{qr_1, qr_2, X}"]
      &
      ((q \cdot r_1) \at X) \ox ((q \cdot r_2) \at X)
      \arrow[d, "\delta_{q, r_1, X} \ox \delta_{q, r_2, X}"]
      \\
      (q \cdot (r_1 + r_2)) \at X
      \arrow[d, "\delta_{q, r_1 + r_2, X}"']
      &
      (q \at (r_1 \at X)) \ox (q \at (r_2 \at X))
      \arrow[d, "\tau_{q, r_1 \at X, r_2 \at X}"]
      \\
      q \at ((r_1 + r_2) \at X)
      \arrow[r, "q \at c_{r_1, r_2, X}"']
      &
      q \at ((r_1 \at X) \ox (r_2 \at X))
    \end{tikzcd}
    \quad
    \begin{tikzcd}[column sep=small]
      (r_1 \cdot q + r_2 \cdot q) \at X
      \arrow[r, equals]
      \arrow[d, "c_{r_1 q, r_2 q, X}"']
      &
      ((r_1 + r_2) \cdot q) \at X
      \arrow[d, "\delta_{(r_1 + r_2), q, X}"]
      \\
      ((r_1 \cdot q) \at X) \ox ((r_2 \cdot q) \at X)
      \arrow[d, "\delta_{r_1, q, X} \ox \delta_{r_2, q, X}"']
      &
      (r_1 + r_2) \at (q \at X)
      \arrow[dl, "c_{r_1, r_2, q\at X}"]
      \\
      (r_1 \at (q \at X)) \ox (r_2 \at (q \at X))
    \end{tikzcd}
    \\
    \begin{tikzcd}
      0 \at X
      \arrow[r, equals]
      \arrow[d, "w_X"']
      &
      (0 \cdot r) \at X
      \arrow[d, "\delta_{0, r, X}"]
      \\
      I
      &
      0 \at (r \at X)
      \arrow[l, "w_{r \at X}"]
    \end{tikzcd}
    \and
    \begin{tikzcd}
      0 \at X
      \arrow[r, equals]
      \arrow[d, "w_X"']
      &
      (r \cdot 0) \at X
      \arrow[r, "\delta_{r, 0, X}"]
      &
      r \at (0 \at X)
      \arrow[d, "r \at w_X"]
      \\
      I
      \arrow[rr, "\iota_r"']
      &
      &
      r \at I
    \end{tikzcd}
  \end{mathpar}
  \caption{Graded comonad coherence conditions}
  \label{fig:graded-comonad-coherence}
\end{figure*}

\begin{definition}
  By a \emph{symmetric lax monoidal adjunction}
  we mean a symmetric lax monoidal functor $ F \from \mc C \to \mc D $ between SMCs
  together with a lax monoidal right adjoint $ G $,
  such that the unit and counit of the adjunction are monoidal natural transformations.
  We concisely denote this by $ F \dashv G \from \mc C \to \mc D $.
  Recall that $ F $ is strong monoidal in this case \cite{kelly_doctrinal_1974}.
\end{definition}

\begin{definition}
  \label{defn:action-morphism}
  Let $ \phi \from \mode m \to \mode n $ be a morphism of modes,
  and let $ \at_{\mode m} $
  and $ \at_{\mode n} $
  be exponential actions of $ \mode m $, $ \mode n $
  on the SMCCs $ \mc C_{\mode m} $ and $ \mc C_{\mode n} $
  respectively.
  A morphism of actions (over $ \phi $)
  consists of a symmetric lax monoidal adjunction
  $ F \dashv G \from \mc C_{\mode n} \to \mc C_{\mode m} $
  together with a monoidal natural transformation
  \[
    \ell_{r, A} \from \phi(r) \at_{\mode n} G(A) \to G(r \at_{\mode m} A)
  \]
  satisfying the coherence conditions given in Figure \ref{fig:action-morphism-coherence}.
  We define another natural transformation
  $ \mu_{r, A} \from F(\phi (r) \at A) \to r \at F(A) $
  as the transpose of the composite
  \[
    \phi(r) \at A
    \xrightarrow{\phi(r) \at \eta} \phi(r) \at GFA
    \xrightarrow{\ell} G (r \at FA).
  \]
  Morphisms between exponential actions compose in a straightforward way,
  giving rise to a category $ \ExpAct $ of exponential actions,
  whose objects are triples $ (\mode m, \mc C, \at) $ consisting of a mode $ \mode m $,
  an SMCC $ \mc C $ and an exponential action of $ \mode m $ on $ \mc C $,
  and whose morphisms are quadruples $ (\phi, F, G, \ell) $, where
  $ \phi $ is a morphism of modes and $ (F, G, \ell) $
  is a morphism of actions over $ \phi $.
  We endow $ \ExpAct $ with the structure of a 2-category:
  A $ 2 $-cell between $ (\phi, F, G, \ell) $ and $ (\phi, F', G', \ell') $
  (note that the underlying morphisms of modes are required to be identical)
  is a natural isomorphism $ \alpha \from G \to G' $ which respects the lineators
  in the obvious way:
  \[
    \begin{tikzcd}
      \phi(r) \at GA
      \arrow[r, "\ell"]
      \arrow[d, "\phi(r) \at \alpha"']
      &
      G(r \at A)
      \arrow[d, "\alpha"]
      \\
      \phi(r) \at G'A
      \arrow[r, "\ell'"']
      &
      G'(r \at A)
    \end{tikzcd}
  \]
\end{definition}

\begin{figure}
  \begin{mathpar}
    \small
    \begin{tikzcd}
      \phi(1) \at G(A)
      \arrow[r, "\ell"]
      \arrow[d, equals]
      &
      G(1 \at A)
      \arrow[d, "G(\epsilon)"]
      \\
      1 \at G(A)
      \arrow[r, "\epsilon"']
      &
      G(A)
    \end{tikzcd}
    \and
    \begin{tikzcd}
      0 \at G(A)
      \arrow[r, "\ell"]
      \arrow[d, "w_{G(A)}"']
      &
      G(0 \at A)
      \arrow[d, "G(w_A)"]
      \\
      I
      \arrow[r]
      &
      G(I)
    \end{tikzcd}
    \\
    \begin{tikzcd}
      (\phi(q)\cdot\phi(r)) \at G(A)
      \arrow[r, equals]
      \arrow[d, "\delta"']
      &
      \phi(qr) \at G(A)
      \arrow[d, "\ell"]
      \\
      \phi(q) \at \phi(r) \at G(A)
      \arrow[d, "\phi(q) \at \ell"']
      &
      G(qr \at A)
      \arrow[d, "G(\delta)"]
      \\
      \phi(q) \at G(r \at A)
      \arrow[r, "\ell"]
      &
      G(q \at r \at A)
    \end{tikzcd}
    \and
    \begin{tikzcd}
      (\phi(q) + \phi(r)) \at G(A)
      \arrow[r, equals]
      \arrow[d, "c"']
      &
      \phi(q + r) \at G(A)
      \arrow[d, "\ell"]
      \\
      \phi(q) \at G(A) \ox \phi(r) \at G(A)
      \arrow[d, "\ell \ox \ell"']
      &
      G((q + r) \at A)
      \arrow[d, "G(c)"]
      \\
      G(q \at A) \ox G(r \at A)
      \arrow[r]
      &
      G(q \at A \ox r \at A)
    \end{tikzcd}
  \end{mathpar}
  \caption{Coherence laws on morphisms of exponential actions.
    Unlabelled arrows are the lax monoidal structure on $ G $.}
  \label{fig:action-morphism-coherence}
\end{figure}

\begin{remark}
  The terminology \emph{lineator} is borrowed from the theory of actegories [sic].
  In fact, our notion of morphisms between actions is adapted from the notion
  of morphisms between actegories.
  See e.g.\@ Cappucci and Gavranović \cite{capucci_actegories_2024}.
\end{remark}

\begin{remark}
  Let
  $ (\phi, F, G, \ell)
  \from
  (\mode m, \mc C_{\mode m}, \odot_{\mode m})
  \to
  (\mode n, \mc C_{\mode n}, \odot_{\mode n})
  $
  be a morphism in $ \ExpAct $.
  Transposing $ \ell $ along the adjunction $ F \dashv G $ yields a
  natural transformation
  $ F(\phi(-) \at_{\mode n} G(=)) \to (-) \at_{\mode m} (=) $.
  Katsumata proved that the domain of this
  transformation satisfies the coherence conditions for exponential actions
  \cite{katsumata_double_2018}.
  In Katsumata's setting, this functor is in fact a graded linear exponential comonad.
  Here, it may fail to be an exponential action
  since $ G $ is not required to be strong monoidal.
  Thus,~$ \ell $ is equivalent to a comparison morphism of the two actions
  on $ \mc C_{\mode m} $ that naturally arise from the adjunction $ F \dashv G $.
\end{remark}

\begin{lemma}
  \label{lem:mu-coherence}
  The coherence conditions on $ \ell $ given in
  Figure \ref{fig:action-morphism-coherence}
  are equivalent to the analogous coherence conditions on~$ \mu $ given in
  Figure \ref{fig:coherence-mu}.
\end{lemma}

\begin{figure}
  \begin{mathpar}
    \small
    \begin{tikzcd}
      F(\phi(1) \at A)
      \arrow[r, "\mu"]
      \arrow[d, equals]
      &
      1 \at FA
      \arrow[d, "\epsilon"]
      \\
      F(1 \at A)
      \arrow[r, "F\epsilon"']
      &
      F(A)
    \end{tikzcd}
    \and
    \begin{tikzcd}
      F(\phi(0) \at A)
      \arrow[r, "\mu"]
      \arrow[d, "Fw"']
      &
      0 \at FA
      \arrow[d, "w"]
      \\
      F(I)
      \arrow[r, "u^{-1}"']
      &
      I
    \end{tikzcd}
    \\
    \begin{tikzcd}
      F(\phi(q)\phi(r) \at A)
      \arrow[r, equals]
      \arrow[d, "F(\delta)"']
      &
      F(\phi(qr) \at A)
      \arrow[d]
      \\
      F(\phi(q) \at \phi(r) \at A)
      \arrow[d, "\mu"']
      &
      qr \at F(A)
      \arrow[d, "\delta"]
      \\
      q \at F(\phi(r) \at A)
      \arrow[r, "q \at \mu"']
      &
      q \at r \at FA
    \end{tikzcd}
    \and
    \begin{tikzcd}
      F((\phi(q) + \phi(r)) \at A)
      \arrow[r, equals]
      \arrow[d, "F(c)"']
      &
      F(\phi(q + r) \at A)
      \arrow[d, "\mu"]
      \\
      F(\phi(q) \at A \ox \phi(r) \at A)
      \arrow[d, "m^{-1}"']
      &
      (q + r) \at FA
      \arrow[d, "c"]
      \\
      F(\phi(q) \at A) \ox F(\phi(r) \at A)
      \arrow[r, "\mu \ox \mu"']
      &
      q \at FA \ox r \at FA
    \end{tikzcd}
  \end{mathpar}
  \caption{Coherence laws on $ \mu $.}
  \label{fig:coherence-mu}
\end{figure}

\begin{proof}
  Straightforward manipulation of diagrams,
  using the fact that $ \ell $ can be written as the composite
  $
  \phi(r) \at G(A)
  \xrightarrow{\eta}
  GF(\phi(r) \at G(A))
  \xrightarrow{G(\mu)}
  G(r \at FG(A))
  \xrightarrow{G(r \at \epsilon)}
  G(r \at A)
  $.
  Here $ \epsilon $ is the counit of the adjunction $ F \dashv G $.
\end{proof}

\begin{lemma}
  \label{lem:ell-mu-induced-coherence}
  Let $ (G, F, \ell) $ and $ (G', F', \ell') $ be morphisms of actions over $ \phi $,
  and $ \alpha \from G \to G' $ a 2-cell in $ \ExpAct $.
  Then $ \alpha $ induces a natural isomorphism $ \beta \from F' \to F $
  which respects $ \mu $ in the obvious way:
  \[
	\begin{tikzcd}
      F'(\phi(q) \at A)
      \arrow[r, "\mu"]
      \arrow[d, "\beta"']
      &
      q \at F'(A)
      \arrow[d, "q \at \beta"]
      \\
      F(\phi(q) \at A)
      \arrow[r, "\mu'"']
      &
      q \at F(A)
    \end{tikzcd}
  \]
\end{lemma}

\begin{definition}
  Recall that \grass{} is parametrized by a functor $ d \from S \to \Mode $
  where $ S $ is some preordered set.
  We have a projection functor
  $ p \from \ExpAct \to \Mode $ given by
  $ (\mode m, \mc C, \at) \mapsto \mode m $.
  A categorial model for \grass{} consists of a pseudofunctor
  $ D \from S \to \ExpAct $ such that $ p \circ D = d $.
  Spelling this out, this amounts to specifying the following data:
  For each mode $ \mode m $, an SMCC $ \mc C_{\mode m} $ together with an exponential
  action $ \at_{\mode m} $ of $ \mode m $ on $ \mc C_{\mode m} $.
  And whenever $ \mathsf{m}  \le  \mathsf{n} $, a morphism of exponential actions
  $ ( \phi_{\mode m, \mode n},
      F^{\mode n}_{\mode m},
      G^{\mode n}_{\mode m},
      \ell^{\mode n}_{\mode m}
      )
  \from (\mode m, \mc C_{\mode m}, \at_{\mode m}) \to
  (\mode n, \mc C_{\mode n}, \at_{\mode n}) $
  where $ \phi_{\mode m, \mode n} $ is the morphism of grade algebras specified
  in Definition \ref{defn:parametrizing-data}
  along with a coherent system of natural isomorphisms
  $ G^{\mode m_3}_{\mode m_2 }\circ G^{\mode m_2}_{\mode m_1}
  \iso G^{\mode m_3}_{\mode m_1} $
  whenever $ \mathsf{m}_{{\mathrm{1}}}  \le  \mathsf{m}_{{\mathrm{2}}}  \le  \mathsf{m}_{{\mathrm{3}}} $
  which are compatible with the lineators.

  We write
  $ u \from I \to F^{\mode m}_{\mode n}(I) $ and
  $ m
  \from F^{\mode m}_{\mode n}(A) \ox F^{\mode m}_{\mode n}(B)
  \to F^{\mode m}_{\mode n}(A \ox B)
  $
  for the strong monoidal structure on the functors $ F^{\mode m}_{\mode n} $.
  We denote internal hom-objects by $ X \lto Y $.
  We write $ \textit{ev} \from (X \lto Y) \ox X \to Y $ for the counit of the
  of the adjunction $ (-) \ox X \dashv X \lto (-) $.
  Given $ e \from X \ox Y \to Z $ we write $ \lambda e \from X \to (Y \lto Z) $
  for the transpose of $ e $ along this adjunction.
\end{definition}

\begin{example}
  Let $ F \dashv G \from \mc C \to \mc L $ be a symmetric lax monoidal adjunction.
  Suppose that $ \ms L $ acts on $ \mc L $ as described in
  Example \ref{exmp:single-mode-models},
  and that $ \mc C $ admits an exponential action by a mode $ \ms m = (R, R, \ms{true}) $
  for some grade algebra $ R $.
  We promote this adjunction to a morphism of exponential actions with lineator
  \[
	\phi(n) \at G(A) = (\phi(1) + \mathellipsis + \phi(1)) \at G(A)
    \xrightarrow{l} G(A)^{\ox n} \xrightarrow{z} G(A^{\ox n})
  \]
  where the first arrow, $l$, is $ c $ (resp.\@ $ w $ if $ n = 0 $)
  and the second arrow, $z$, is the monoidal structure on $ G $.
  Therefore, any model of
  \textsf{mGL}~\cite{vollmer_mixed_2025}
  is a model of \grass{} instantiated
  with the modes $ \ms L < \mode m $.
  When $ \mode m = \ms U $,
  $ \mc C $ is cartesian closed, and $ \ms U $ acts on $ \mc C $ by $ 1 \at (-) = \id $,
  as described in Example \ref{exmp:single-mode-models},
  we recover the categorical situation of LNL \cite{benton_mixed_1995}.
\end{example}

\begin{definition}
  Recall the rules for well-formed types from Figure \ref{fig:types-well-formed}.
  We assign each type $  \grassnt{A}  \in \mathsf{Type}( \mathsf{m} )  $ an object of $ \interp A \in \mc C_{\mode m} $
  by induction on the proof of $  \grassnt{A}  \in \mathsf{Type}( \mathsf{m} )  $.
  \begin{mathpar}
    \inferrule[unit]
    { }
    {
      \interp{ \mathbf{I}_ \mathsf{m} } = I \in \mc C_{\mode m}
    }
    \and
    \inferrule[pair]
    {
       \grassnt{A}  \in \mathsf{Type}( \mathsf{m} )  \and  \grassnt{B}  \in \mathsf{Type}( \mathsf{m} ) 
    }
    {
      \interp {A \ox B} = \interp A \ox \interp B \in \mc C_{\mode m}
    }
    \and
    \inferrule[fun]
    {
      \mathsf{n}  \le  \mathsf{m} \and  \grassnt{q}  \in R_{ \mathsf{m} } 
      \\\\
       \grassnt{A}  \in \mathsf{Type}( \mathsf{m} ) 
      \and
       \grassnt{B}  \in \mathsf{Type}( \mathsf{n} ) 
    }
    {
      \interp{  \grassnt{A} ^{ \grassnt{q} : \mathsf{m} } \multimap  \grassnt{B} } =
      F^{\mode m}_{\mode n}(q \at \interp A) \lto \interp B \in \mc C_{\mode n}
    }
    \and
    \inferrule[drop]
    {
      \mathsf{n}  \le  \mathsf{m} \and  \grassnt{q}  \in R_{ \mathsf{m} } 
      \and
       \grassnt{A}  \in \mathsf{Type}( \mathsf{m} ) 
    }
    {
      \interp{  \operatorname{\downarrow} ^{ \grassnt{q} }_{ \mathsf{n}  \le  \mathsf{m} }  \grassnt{A} } =
      F^{\mode m}_{\mode n}(q \at \interp A) \in \mc C_{\mode n}
    }
    \and
    \inferrule[drop]
    {
      \mathsf{m}  \le  \mathsf{n} \and  \grassnt{A}  \in \mathsf{Type}( \mathsf{m} ) 
    }
    {
      \interp{  \operatorname{\uparrow} _{ \mathsf{m}  \le  \mathsf{n} }  \grassnt{A} } = G^{\mode n}_{\mode m} \interp A \in \mc C_{\mode n}
    }
  \end{mathpar}
  We will often omit the brackets $ \interp - $ writing $ A $ instead of $ \interp A $.
\end{definition}

\begin{definition}
  \label{defn:ctx-interp}
  Let $ \rho = \grassnt{r_{{\mathrm{1}}}}  \grasssym{,}  \mathellipsis  \grasssym{,}  \grassnt{r_{\grassmv{k}}} $
  and $ \mathsf{M} = \mathsf{m}_{{\mathrm{1}}}  \grasssym{,}  \mathellipsis  \grasssym{,}  \mathsf{m}_{\grassmv{k}} $
  and $ \Gamma = \grassmv{x_{{\mathrm{1}}}}  \grasssym{:}  \grassnt{A_{{\mathrm{1}}}}  \grasssym{,}  \mathellipsis  \grasssym{,}  \grassmv{x_{\grassmv{k}}}  \grasssym{:}  \grassnt{A_{\grassmv{k}}} $,
  and $ \mathsf{M}  \ge  \mathsf{m} $.
  Define
  \[
    \interp{ \rho  \mid  \mathsf{M}  \odot  \Gamma }_\mathsf{m}
    :=
    \bigotimes_{i = 1}^{k} \interp { \operatorname{\downarrow} ^{ \grassnt{r_{\grassmv{i}}} }_{ \mathsf{m}  \le  \mathsf{m}_{\grassmv{i}} }  \grassnt{A_{\grassmv{i}}} }
    =
    \bigotimes_{i = 1}^{k} F^{\mathsf{m}_{\grassmv{i}}}_\mathsf{m} (r_i \at A_i)
  \]
  Note that there is a canonical isomorphism
  $ F^\mathsf{m}_\mathsf{n}(\interp{  \rho  \mid  \mathsf{M}  \odot  \Gamma }_\mathsf{m})
  \iso \interp{  \rho  \mid  \mathsf{M}  \odot  \Gamma }_\mathsf{n} $
  whenever $ \mathsf{n}  \le  \mathsf{m} $.
\end{definition}

\begin{definition}
  In the situation of Definition \ref{defn:ctx-interp}, let $  \grassnt{q}  \in R_{ \mathsf{m} }  $
  and let $ t \from \interp{ \rho  \mid  \mathsf{M}  \odot  \Gamma }_{\mathsf{m}} \to A $ be a morphism.
  Define $ q \cdot t \from \interp{   \grassnt{q}   \rho   \mid  \mathsf{M}  \odot  \Gamma }_\mathsf{m} \to q \at A $ as
  \[
    \bigotimes_{i = 1}^k
    F^{\mathsf{m}_{\grassmv{i}}}_\mathsf{m} (q \cdot r_i) \at A_i
    \xrightarrow{\tau \circ \bigotimes_i (\mu \circ F^{\mathsf{m}_{\grassmv{i}}}_{\mathsf{m}}\delta)}
    q \at \bigotimes_{i = 1}^k F^{\mathsf{m}_{\grassmv{i}}}_\mathsf{m} (r_i \at A_i)
    \xrightarrow {q \at t}
    q \at A
  \]
  for $ k > 0 $.
  When $ k = 0 $, define
  $ q \cdot t = (q \at t) \circ \iota \from
  I_{\mode m} \to q \at I_{\mode m} \to q \at A $.
\end{definition}

\begin{figure}
  \small
  \centering
  \renewcommand{\arraystretch}{1.1}%
  \begin{tabular}{|c | l | c |}
    \hline
    Name & Type & Isomorphism?
    \\ \hline
    $ \delta $ & $ (rq) \at X \to r \at (q \at X) $ & no
    \\
    $ \epsilon $ & $ 1 \at X \to X $ & no
    \\
    $ \tau $ & $ (r \at X) \tensor (r \at Y) \to r \at (X \tensor Y) $ & yes
    \\
    $ \iota $ & $ I \to r \at I $ & yes
    \\
    $ c $ & $ (r + q) \at X \to (r \at X) \tensor (q \at X) $ & no
    \\
    $ w $ & $ 0 \at X \to I $ & no
    \\
    $ m $ & $ F(X) \tensor F(Y) \to F (X \tensor Y) $ & yes
    \\
    $ u $ & $ I \to F(I) $ & yes
    \\
    $ \mu $ & $ F(\phi_{\mathsf{n}, \mathsf{m}}(q) \at X) \to q \at F^\mathsf{m}_\mathsf{n} (X) $ & no
    \\ \hline
  \end{tabular}
  \caption{Summary of named natural transformations in the categorical model for
    \grass{}}
  \label{fig:natural-trafo-summary}
\end{figure}

\begin{definition}
  We assign each judgment $  \rho  \mid  \mathsf{M}   \odot   \Gamma  \vdash_{ \mathsf{m} }  \grassnt{t}  \colon  \grassnt{A}  $ a morphism
  $ \interp t \from \interp{  \rho  \mid  \mathsf{M}  \odot  \Gamma } \to \interp A $ in $ \mc C_{\mode m} $.
  This assignment is specified in Figure \ref{fig:grass-interp-model}
  by induction on the derivation of $ t $.
\end{definition}

\begin{figure}
  \small
  \begin{mathpar}
    \inferrule
    [\rulename{grass}{term}{var}]
    {
    }
    {
      F^\mathsf{m}_\mathsf{m} (1 \at A) \iso 1 \at A \xrightarrow{\epsilon} A
    }
    \and
    \inferrule
    [\rulename{grass}{term}{weak}]
    {
       \mathsf{Weak} (  \mathsf{l}  ) 
      \and
      \mathsf{m}  \le  \mathsf{l}
      \and
       \grassnt{B}  \in \mathsf{Type}( \mathsf{l} ) 
      \and
      \interp{ \rho  \mid  \mathsf{M}  \odot  \Gamma }_\mathsf{m} \xrightarrow{t} A
    }
    {
      \interp{ \rho  \mid  \mathsf{M}  \odot  \Gamma }_\mathsf{m} \ox F^ \mathsf{l}_ \mathsf{m}(0 \at B)
      \xrightarrow{t \ox (u^{-1} \circ F(w))}
      A \ox I \iso A
    }
    \and
    \inferrule
    [\rulename{grass}{term}{sub}]
    {
      (r_i \le q_i)_i
      \and
      \bigotimes\nolimits_i F^{\mathsf{m}_{\grassmv{i}}}_{\mathsf{m}} (r_i \at X_i) \xrightarrow{t} A
    }
    {
      \bigotimes\nolimits_i F^{\mathsf{m}_{\grassmv{i}}}_{\mathsf{m}} (q_i \at X_i)
      \xrightarrow{F^{\mathsf{m}_{\grassmv{i}}}_{\mathsf{m}}((r_i \le q_i) \at X_i)}
      \bigotimes\nolimits_i F^ {\mathsf{m}_{\grassmv{i}}}_ {\mathsf{m}} (r_i \at X_i)
      \xrightarrow{t} A
    }
    \and
    \inferrule
    [\rulename{grass}{term}{cont}]
    {
       \grassnt{r_{{\mathrm{1}}}}  \grasssym{,}  \grassnt{r_{{\mathrm{2}}}}  \in \mathsf{Cont}( \mathsf{m} ) 
      \and
      \interp{ \rho  \mid  \mathsf{M}  \odot  \Gamma }_{ \mathsf{n}}
      \ox F^\mathsf{m}_\mathsf{n}(r_1 \at A) \ox F^\mathsf{m}_\mathsf{n}(r_2 \at A)
      \xrightarrow{t}
      B
    }
    {
      \interp{ \rho  \mid  \mathsf{M}  \odot  \Gamma }_{ \mathsf{n}}
      \ox F^\mathsf{m}_\mathsf{n}((r_1 + r_2) \at A)
      \xrightarrow{\id \ox (m^{-1} \circ F^{ \mathsf{m}}_{ \mathsf{n}}(c)) }
      \interp{ \rho  \mid  \mathsf{M}  \odot  \Gamma }_{ \mathsf{n}}
      \ox F^\mathsf{m}_\mathsf{n}(r_1 \at A) \ox F^\mathsf{m}_\mathsf{n}(r_2 \at A)
      \xrightarrow{t}
      B
    }
    \and
    \inferrule
    [\rulename{grass}{term}{unitI}]
    {
    }
    {
      I_{ \mathsf{m}} \xrightarrow{\id} I_{ \mathsf{m}}
    }
    \and
    \inferrule
    [\rulename{grass}{term}{unitE}]
    {
      \interp{  \sigma  \mid  \mathsf{N}  \odot  \Delta }_{ \mathsf{m}} \xrightarrow{e} I_{\mathsf{m}}
      \and
      \interp{  \rho  \mid  \mathsf{M}  \odot  \Gamma }_{ \mathsf{m}} \xrightarrow{t} A
    }
    {
      \interp{  \rho  \mid  \mathsf{M}  \odot  \Gamma }_{ \mathsf{m}} \ox \interp{   \grassnt{q}   \sigma   \mid  \mathsf{N}  \odot  \Delta }_{ \mathsf{m}}
      \xrightarrow{t \ox (\iota^{-1} \circ (q \cdot e))}
      A \ox I_{\mode m} \iso A
    }
    \and
    \inferrule
    [\rulename{grass}{term}{arrowI}]
    {
      \interp{  \rho  \mid  \mathsf{M}  \odot  \Gamma }_{\mathsf{n}} \ox F^{\mathsf{m}}_{\mathsf{n}}(r \at A)
      \xrightarrow{e} B
    }
    {
      \interp{  \rho  \mid  \mathsf{M}  \odot  \Gamma }_{\mathsf{m}}
      \xrightarrow{\lambda e} F^{\mathsf{m}}_{\mathsf{n}} (r \at A) \lto B
    }
    \and
    \inferrule
    [\rulename{grass}{term}{arrowE}]
    {
      \interp{  \rho  \mid  \mathsf{M}  \odot  \Gamma }_{\mathsf{n}} \xrightarrow{t} F^{\mathsf{m}}_{\mathsf{n}} (q \at A) \lto B
      \and
      \interp{  \sigma  \mid  \mathsf{N}  \odot  \Delta }_{\mathsf{m}} \xrightarrow{e} A
    }
    {
      \interp{  \rho  \mid  \mathsf{M}  \odot  \Gamma }_{\mathsf{n}}
      \ox
      F^{\mode m}_{\mode n} \interp{  \sigma  \mid  \mathsf{N}  \odot  \Delta }_{\mathsf{m}}
      \xrightarrow{t \ox F^{\mathsf{m}}_{\mathsf{n}}(q \cdot e)}
      (F^{\mathsf{m}}_{\mathsf{n}} (q \at A) \lto B) \ox F^{\mathsf{m}}_{\mathsf{n}}(q \at A)
      \xrightarrow{\textit{ev}} B
    }
    \and
    \inferrule[\rulename{grass}{term}{pairI}]
    {
      \interp{ \rho_{{\mathrm{1}}}  \mid  \mathsf{M}_{{\mathrm{1}}}  \odot  \Gamma_{{\mathrm{1}}} }_{\mathsf{m}} \xrightarrow{a_1} A_1
      \and
      \interp{ \rho_{{\mathrm{2}}}  \mid  \mathsf{M}_{{\mathrm{2}}}  \odot  \Gamma_{{\mathrm{2}}} }_{\mathsf{m}} \xrightarrow{a_2} A_2
    }
    {
      \interp{ \rho_{{\mathrm{1}}}  \mid  \mathsf{M}_{{\mathrm{1}}}  \odot  \Gamma_{{\mathrm{1}}} }_{\mathsf{m}} \ox \interp{ \rho_{{\mathrm{2}}}  \mid  \mathsf{M}_{{\mathrm{2}}}  \odot  \Gamma_{{\mathrm{2}}} }_{\mathsf{m}}
      \xrightarrow{a_1 \ox a_2} A_1 \ox A_2
    }
    \and
    \inferrule[\rulename{grass}{term}{pairE}]
    {
      \interp{  \sigma  \mid  \mathsf{N}  \odot  \Delta }_{\mathsf{m}} \xrightarrow{e} A_1 \ox A_2
      \and
      \interp{  \rho  \mid  \mathsf{M}  \odot  \Gamma }_{\mathsf{n}}
      \ox F^{\mathsf{m}}_{\mathsf{n}}(q \at A_1) \ox F^{\mathsf{m}}_{\mathsf{n}}(q \at A_2)
      \xrightarrow{t} B
    }
    {
      \interp{  \rho  \mid  \mathsf{M}  \odot  \Gamma }_{\mathsf{n}}
      \ox F^{\mode m}_{\mode n} \interp{   \grassnt{q}   \sigma   \mid  \mathsf{N}  \odot  \Delta }_{\mathsf{m}}
      \xrightarrow{\id \ox (m^{-1}\circ F^{\mathsf{m}}_{\mathsf{n}}(\tau^{-1} \circ(q \cdot e)))}
      \interp{  \rho  \mid  \mathsf{M}  \odot  \Gamma }_{\mathsf{n}}
      \ox F^{\mathsf{m}}_{\mathsf{n}}(q \at A_1) \ox F^{\mathsf{m}}_{\mathsf{n}}(q \at A_2)
      \xrightarrow{t} B
    }
    \and
    \inferrule[\rulename{grass}{term}{raiseI}]
    {
      \mathsf{m}  \le  \mathsf{n}  \le  \mathsf{M}
      \and
      \interp{ \rho  \mid  \mathsf{M}  \odot  \Gamma }_{\mathsf{m}} \xrightarrow{e} A
    }
    {
      \interp{ \rho  \mid  \mathsf{M}  \odot  \Gamma }_{\mathsf{n}}
      \xrightarrow{\text{unit}}
      G^{\mathsf{n}}_{\mathsf{m}}F^{\mathsf{n}}_{\mathsf{m}}\interp{ \rho  \mid  \mathsf{M}  \odot  \Gamma }_{\mathsf{n}}
      \iso G^{\mathsf{n}}_{\mathsf{m}}\interp{ \rho  \mid  \mathsf{M}  \odot  \Gamma }_{\mathsf{m}}
      \xrightarrow{G^{\mathsf{n}}_{\mathsf{m}}e}
      G^{\mathsf{n}}_{\mathsf{m}} A
    }
    \and
    \inferrule[\rulename{grass}{term}{raiseE}]
    {
      \mathsf{m}  \le  \mathsf{n}
      \and
      \interp{ \rho  \mid  \mathsf{M}  \odot  \Gamma }_{\mathsf{n}} \xrightarrow{e} G^\mathsf{n}_\mathsf{m} A
    }
    {
      F^\mathsf{m}_\mathsf{n} \interp{ \rho  \mid  \mathsf{M}  \odot  \Gamma }_{\mathsf{n}}
      \xrightarrow{F^\mathsf{m}_\mathsf{n} e}
      F^\mathsf{m}_\mathsf{n} G^\mathsf{n}_\mathsf{m} A
      \xrightarrow{\text{counit}}
      A
    }
    \and
    \inferrule[\rulename{grass}{term}{dropI}]
    {
      \mathsf{n}  \le  \mathsf{m}
      \and
       \grassnt{q}  \in R_{ \mathsf{m} } 
      \and
      \interp{  \rho  \mid  \mathsf{M}  \odot  \Gamma }_\mathsf{m} \xrightarrow{e} A
    }
    {
      F^{\mode m}_{\mode n}
      \interp{   \grassnt{q}   \rho   \mid  \mathsf{M}  \odot  \Gamma }_\mathsf{m}
      \xrightarrow{F^\mathsf{m}_\mathsf{n}(q \cdot e)}
      F^\mathsf{m}_\mathsf{n} (q \at A)
    }
    \and
    \inferrule[\rulename{grass}{term}{dropE}]
    {
      \mathsf{l}  \le  \mathsf{n}  \le  \mathsf{m}
      \and
      \interp{ \sigma  \mid  \mathsf{N}  \odot  \Delta }_\mathsf{n} \xrightarrow{e} F^\mathsf{m}_\mathsf{n} (q \at A)
      \and
      \interp{ \rho  \mid  \mathsf{M}  \odot  \Gamma }_\mathsf{l}
      \ox F^\mathsf{m}_\mathsf{l} (q \at A) \xrightarrow{t} B
    }
    {
      \interp{ \rho  \mid  \mathsf{M}  \odot  \Gamma }_\mathsf{l}
      \ox
      F^{\mode n}_{\mode l}\interp{ \sigma  \mid  \mathsf{N}  \odot  \Delta }_\mathsf{n}
      \xrightarrow{\id \ox F^\mathsf{n}_\mathsf{l} e}
      \interp{ \rho  \mid  \mathsf{M}  \odot  \Gamma }_\mathsf{l}
      \ox
      F^\mathsf{n}_\mathsf{l} F^\mathsf{m}_\mathsf{n} (q \at A)
      \iso
      \interp{ \rho  \mid  \mathsf{M}  \odot  \Gamma }_\mathsf{l}
      \ox
      F^\mathsf{m}_\mathsf{l} (q \at A)
      \xrightarrow{t}
      B
    }
  \end{mathpar}
  \caption{\grass{} interpretation into a categorical model}
  \label{fig:grass-interp-model}
\end{figure}

\begin{theorem}[Soundness]
  $ \beta $- and $ \eta $-conversions are sound with respect to the categorical
  semantics:
  \begin{enumerate}
  \item
    If $  \rho  \mid  \mathsf{M}  \odot  \Gamma  \vdash_{ \mathsf{m} }  \grassnt{t_{{\mathrm{1}}}}  \equiv_\beta  \grassnt{t_{{\mathrm{2}}}}  \colon  \grassnt{A}  $
    then $ \interp {t_1} = \interp {t_2} $ in $ \mc C_{\mode m} $.

  \item
    If $  \rho  \mid  \mathsf{M}  \odot  \Gamma  \vdash_{ \mathsf{m} }  \grassnt{t_{{\mathrm{1}}}}  \equiv_\eta  \grassnt{t_{{\mathrm{2}}}}  \colon  \grassnt{A}  $
    then $ \interp {t_1} = \interp {t_2} $ in $ \mc C_{\mode m} $.
  \end{enumerate}
\end{theorem}

\begin{proof}
  Given in
  \ifextended
    \S\S \ref{asec:beta-semantic-soundness}, \ref{asec:eta-semantic-soundness} of the appendix.
  \else
    the extended version \cite{extendedVersion}.
  \fi
\end{proof}

\section{Conclusion, related and future work}
\label{sec:conclusion}
  We defined a graded logic which allows grades
from multiple modes to exist in a single system.
Furthermore, we have generalized the graded weakening and contraction rules
to obtain more fine-grained control in the graded setting.
This allowed us to recover non-graded substructural logics
in an entirely graded setting.
We gave an account of the categorical semantics of our logic based on graded
comonads.
Our logic consists of multiple graded systems connected by morphisms of modes,
and this design is mirrored in the categorical semantics.
We introduced a novel notion of morphisms of graded comonads,
and showed how it connects to the existing literature.
We saw that categorical models of substructural logics
provide categorical models for the modes which recover such logics
at the syntactic level.

\subsection*{Combined systems}

Systems which combine logics permitting different sets of structural rules
originate in Benton's Linear-non-Linear Logic (LNL) \cite{benton_mixed_1995},
a system which combines linear and intuitionistic logics.
A similar result was achieved for graded and linear logics in \textsf{mGL}
of Vollmer et al.~\cite{vollmer_mixed_2025}.
In the non-graded world, LNL was further generalized by Pruiksma et al.\@
into Adjoint Logic \cite{pruiksma_adjointlogic_2018},
which allows an arbitrary collection of logics with different structural rules
(relevant, affine, linear, intuitionistic) to exist in one system.
The present work initially started out as a ``graded Adjoint Logic,''
and is a continuation of previous work \cite{hanukaev_combining_2023} with the same objective.
In our previous work we also attempted to incorporate dependent types,
which posed problems for control over the structural rules.
Furthermore we were able to describe the categorical semantic of the type system
proposed in the present work, while we did not do so in our previous work.

Being inspired by Adjoint Logic, \grass{} bears many similarities to Adjoint Logic.
For example, our mode system and modal operators
$  \operatorname{\uparrow} _{ \mathsf{n}  \le  \mathsf{m} }  \grassnt{A}  $ and $  \operatorname{\downarrow} ^{ \grassnt{q} }_{ \mathsf{n}  \le  \mathsf{m} }  \grassnt{A}  $
are direct adaptations of the mode system and modal operators presented
in Adjoint Logic.
All of the systems mentioned here have one observation in common:
Terms from logics with fewer structural rules may depend on variables permitting
more structural rules, but not vice-versa.
This invariant is also present in our work,
and we leverage it to ensure that reuse of more structured variables
is well-behaved in less structured modes.

\subsection*{Heterogenously graded systems}

Systems where grades can be drawn from more than one grade algebra
have been studied for subsets of Java by
Bianchini et al.~\cite{bianchini_java-like_2023,bianchini_multi-graded_2023}
as well as Giannini and Duso~\cite{giannini_coeffects_2024}.
We refer to these works collectively as \emph{Graded Java}.
Modes in \grass{} correspond to \emph{kinds} of coeffects in Graded Java.
Both modes and kinds are arranged in a preorder,
and both systems use morphisms of grade algebras to combine (i.e. add or multiply)
grades from different algebras.

We highlight some differences:
Graded Java reduces the problem of tracking multiple kinds of grades in one program
to the case of a single grade algebra by equipping the disjoint union
$ \bigsqcup_k R_k $ with an appropriate grade algebra structure,
where $ k $ ranges over the set of kinds of coeffects.
The set of kinds is required to have joins,
and if $ r_1 \in R_{k_1} $ and $ r_2 \in R_{k_2} $,
then $ r_1 + r_2 \in R_{k_1 \lor k_2} $ and similarly for multiplication.
This can lead to a loss of information, e.g. when $ k_1 \lor k_2 $ is the
terminal kind $ k_\top $ with $ R_{k_\top} = \top $ the terminal semiring.
This means that a variable can be marked unrestricted in a term,
even though it should have non-trivial usage restrictions in each of its subterms.
In contrast to this,
\grass{}'s typing rules ensure that additions only happen between grades from
the same mode,
and multiplications only happen when the modes are comparable in the preorder,
thus avoiding the need for joins.
Furthermore each variable in \grass{} belongs to a unique mode,
and its grades are always guaranteed to be drawn from that mode.
Hence there is no risk of the loss of information mentioned above.
We leave it to future work to investigate if/how our approach can be implemented
in a system like Graded Java.

\subsection*{Categorical semantics of graded modal types}

Many existing graded type systems have been given categorical semantics
\cite{brunel_core_2014,%
  gaboardi_combining_2016,%
  ghica_bounded_2014,%
  petricek_coeffects_2013,%
  petricek_coeffects_2014%
}.
These semantics were all somewhat similar,
employing a notion of \emph{graded linear exponential comonads}
(though terminology varied).
Later, Katsumata \cite{katsumata_double_2018} gave a comprehensive account
of graded linear exponential comonads.
Katsumata's work also serves as the basis for the exponential actions defined in
the present text.
Beyond the changes that were necessary to account for our modified weakening
and contraction rules,
there is one other difference between our presentation and that of Katsumata:
We require that $ r \at (-) $ be a strong monoidal functor,
while the categorical semantics in the work cited here all require this functor only to
be lax monoidal.
We attribute this difference to the fact that strong monoidality of $ r \at (-) $
is required for the interpretation of the tensor product type,
which was not included in the above works.
Categorical semantics were given to \textsf{mGL}
by Vollmer et al.~\cite{vollmer_mixed_2025}.
This type system does have a tensor product,
and they require the functors $ r \at (-) $
to be strict monoidal, i.e.\@ the morphism
$ (r \at A) \ox (r \at B) \to r \at (A \ox B) $ must be an identity.
We have shown here that the strictness is not necessary for the interpretation of
the graded tensor product.
Finally, Atkey \cite{atkey_syntax_2018} has given categorical semantics
to a dependently typed graded system using \emph{Quantitative Categories with Families},
an approach which is not based on graded comonads.

\subsection*{Future work}

One future direction is to use the lessons learned here about the categorical semantics
of graded types towards a categorical semantics of graded \emph{dependent} types
which is based on graded comonads.
One problem in this area is that uses of terms inside of types are considered
computationally irrelevant, and therefore exempt from grading.
In other words, types are treated as intuitionistic, but terms must still be graded.
The categorical semantics presented here are able to capture
graded and intuitionistic variables existing in the same system,
by instantiating \grass{} with any one mode for the graded system, together
with the terminal mode to recover intuitionistic logic.
We hope to use this as a starting point for the categorical semantics
of graded dependent types.
%

\bibliographystyle{entics}
\bibliography{refs.bib}

\begin{thebibliography}{10}
\providecommand{\url}[1]{\texttt{#1}}
\providecommand{\urlprefix}{ }
\providecommand{\eprint}[2][]{\url{#2}}

\bibitem{abel_unified_2020}
Abel, A. and J.-P. Bernardy, \emph{A unified view of modalities in type
  systems}, Proceedings of the ACM on Programming Languages \textbf{4}, pages
  1--28 (2020), ISSN 2475-1421.
\newline\urlprefix\url{https://doi.org/10.1145/3408972}

\bibitem{atkey_syntax_2018}
Atkey, R., \emph{Syntax and {Semantics} of {Quantitative} {Type} {Theory}}, in:
  \emph{Proceedings of the 33rd {Annual} {ACM}/{IEEE} {Symposium} on {Logic} in
  {Computer} {Science}}, {LICS} '18, pages 56--65, Association for Computing
  Machinery, New York, NY, USA (2018), ISBN 978-1-4503-5583-4.
\newline\urlprefix\url{https://doi.org/10.1145/3209108.3209189}

\bibitem{azevedo_de_amorim_semantic_2017}
Azevedo~de Amorim, A., M.~Gaboardi, J.~Hsu, S.-y. Katsumata and I.~Cherigui,
  \emph{A semantic account of metric preservation}, SIGPLAN Not. \textbf{52},
  pages 545--556 (2017), ISSN 0362-1340.
\newline\urlprefix\url{https://doi.org/10.1145/3093333.3009890}

\bibitem{benton_mixed_1995}
Benton, P.~N., \emph{A mixed linear and non-linear logic: {Proofs}, terms and
  models}, in: L.~Pacholski and J.~Tiuryn, editors, \emph{Computer {Science}
  {Logic}}, pages 121--135, Springer, Berlin, Heidelberg (1995), ISBN
  978-3-540-49404-1.
\newline\urlprefix\url{https://doi.org/10.1007/BFb0022251}

\bibitem{bianchini_java-like_2023}
Bianchini, R., F.~Dagnino, P.~Giannini and E.~Zucca, \emph{A {Java}-like
  calculus with heterogeneous coeffects}, Theoretical Computer Science
  \textbf{971}, page 114063 (2023), ISSN 0304-3975.
\newline\urlprefix\url{https://doi.org/10.1016/j.tcs.2023.114063}

\bibitem{bianchini_multi-graded_2023}
Bianchini, R., F.~Dagnino, P.~Giannini and E.~Zucca, \emph{Multi-{Graded}
  {Featherweight} {Java}}, in: \emph{37th {European} {Conference} on
  {Object}-{Oriented} {Programming} ({ECOOP} 2023)}, pages 3:1--3:27, Schloss
  Dagstuhl – Leibniz-Zentrum für Informatik (2023).
\newline\urlprefix\url{https://doi.org/10.4230/LIPIcs.ECOOP.2023.3}

\bibitem{bierman_intuitionistic_2000}
Bierman, G.~M. and V.~d. Paiva, \emph{On an {Intuitionistic} {Modal} {Logic}},
  Stud Logica \textbf{65}, pages 383--416 (2000).
\newline\urlprefix\url{https://doi.org/10.1023/A:1005291931660}

\bibitem{brunel_core_2014}
Brunel, A., M.~Gaboardi, D.~Mazza and S.~Zdancewic, \emph{A {Core}
  {Quantitative} {Coeffect} {Calculus}}, in: Z.~Shao, editor, \emph{Programming
  {Languages} and {Systems}}, pages 351--370, Springer, Berlin, Heidelberg
  (2014), ISBN 978-3-642-54833-8.
\newline\urlprefix\url{https://doi.org/10.1007/978-3-642-54833-8_19}

\bibitem{capucci_actegories_2024}
Capucci, M. and B.~Gavranović, \emph{Actegories for the {Working}
  {Amthematician}} (2024). ArXiv:2203.16351 [math].
\newline\urlprefix\url{https://doi.org/10.48550/arXiv.2203.16351}

\bibitem{choudhury_graded_2021}
Choudhury, P., H.~Eades~III, R.~A. Eisenberg and S.~Weirich, \emph{A graded
  dependent type system with a usage-aware semantics}, Proceedings of the ACM
  on Programming Languages \textbf{5}, pages 1--32 (2021), ISSN 2475-1421.
\newline\urlprefix\url{https://doi.org/10.1145/3434331}

\bibitem{de_amorim_really_2014}
de~Amorim, A.~A., M.~Gaboardi, E.~J. Gallego~Arias and J.~Hsu, \emph{Really
  {Natural} {Linear} {Indexed} {Type} {Checking}}, in: \emph{Proceedings of the
  26nd 2014 {International} {Symposium} on {Implementation} and {Application}
  of {Functional} {Languages}}, {IFL} '14, pages 1--12, Association for
  Computing Machinery, New York, NY, USA (2014), ISBN 978-1-4503-3284-2.
\newline\urlprefix\url{https://doi.org/10.1145/2746325.2746335}

\bibitem{gaboardi_combining_2016}
Gaboardi, M., S.-y. Katsumata, D.~Orchard, F.~Breuvart and T.~Uustalu,
  \emph{Combining effects and coeffects via grading}, in: \emph{Proceedings of
  the 21st {ACM} {SIGPLAN} {International} {Conference} on {Functional}
  {Programming}}, {ICFP} 2016, pages 476--489, Association for Computing
  Machinery, New York, NY, USA (2016), ISBN 978-1-4503-4219-3.
\newline\urlprefix\url{https://doi.org/10.1145/2951913.2951939}

\bibitem{ghica_bounded_2014}
Ghica, D.~R. and A.~I. Smith, \emph{Bounded {Linear} {Types} in a {Resource}
  {Semiring}}, in: Z.~Shao, editor, \emph{Programming {Languages} and
  {Systems}}, pages 331--350, Springer, Berlin, Heidelberg (2014), ISBN
  978-3-642-54833-8.
\newline\urlprefix\url{https://doi.org/10.1007/978-3-642-54833-8_18}

\bibitem{giannini_coeffects_2024}
Giannini, P. and G.~Duso, \emph{Coeffects for {MiniJava}: {Cf}-{Mj}}, in:
  \emph{Proceedings of the 26th {ACM} {International} {Workshop} on {Formal}
  {Techniques} for {Java}-like {Programs}}, {FTfJP} 2024, pages 30--36,
  Association for Computing Machinery, New York, NY, USA (2024), ISBN
  979-8-4007-1111-4.
\newline\urlprefix\url{https://doi.org/10.1145/3678721.3686232}

\bibitem{girard_linear_1987}
Girard, J.-Y., \emph{Linear logic}, Theoretical Computer Science \textbf{50},
  pages 1--101 (1987), ISSN 0304-3975.
\newline\urlprefix\url{https://doi.org/https://doi.org/10.1016/0304-3975(87)90045-4}

\bibitem{hanukaev_combining_2023}
Hanukaev, P. and H.~Eades~III, \emph{Combining {Dependency}, {Grades}, and
  {Adjoint} {Logic}}, in: \emph{Proceedings of the 8th {ACM} {SIGPLAN}
  {International} {Workshop} on {Type}-{Driven} {Development}}, {TyDe} 2023,
  pages 58--70, Association for Computing Machinery, New York, NY, USA (2023),
  ISBN 979-8-4007-0299-0.
\newline\urlprefix\url{https://doi.org/10.1145/3609027.3609408}

\bibitem{extendedVersion}
Hanukaev, P. and H.~Eades~III, \emph{A unification of graded and substructural
  logics} (2026). Version with appendices.
\newline\urlprefix\url{https://arxiv.org/abs/2605.17112v2}

\bibitem{katsumata_double_2018}
Katsumata, S.-y., \emph{A {Double} {Category} {Theoretic} {Analysis} of
  {Graded} {Linear} {Exponential} {Comonads}}, in: C.~Baier and U.~Dal~Lago,
  editors, \emph{Foundations of {Software} {Science} and {Computation}
  {Structures}}, pages 110--127, Springer International Publishing, Cham
  (2018), ISBN 978-3-319-89366-2.
\newline\urlprefix\url{https://doi.org/10.1007/978-3-319-89366-2_6}

\bibitem{kelly_doctrinal_1974}
Kelly, G.~M., \emph{Doctrinal adjunction}, in: G.~M. Kelly, editor,
  \emph{Category {Seminar}}, pages 257--280, Springer, Berlin, Heidelberg
  (1974), ISBN 978-3-540-37270-7.
\newline\urlprefix\url{https://doi.org/10.1007/BFb0063105}

\bibitem{Lambek:1958}
Lambek, J., \emph{The mathematics of sentence structure}, The American
  Mathematical Monthly \textbf{65}, pages 154--170 (1958).
\newline\urlprefix\url{https://doi.org/10.2307/2310058}

\bibitem{Lang_2002}
Lang, S., \emph{Algebra}, Springer New York (2002).
\newline\urlprefix\url{https://doi.org/10.1007/978-1-4613-0041-0}

\bibitem{liepelt_graded_2026}
Liepelt, V.-B., D.~Marshall, D.~Orchard, V.~Rajani and M.~Vollmer, \emph{On
  {Graded} {Coeffect} {Types} for {Information}-{Flow} {Control}}, in:
  D.~Orchard, T.~Petricek and J.~Singer, editors, \emph{Languages, {Compilers},
  {Analysis} - {From} {Beautiful} {Theory} to {Useful} {Practice}: {Essays}
  {Dedicated} to {Alan} {Mycroft} on the {Occasion} of {His} {Retirement}},
  pages 114--148, Springer Nature Switzerland, Cham (2026), ISBN
  978-3-032-08187-2.
\newline\urlprefix\url{https://doi.org/10.1007/978-3-032-08187-2_7}

\bibitem{mellies_categorical_2009}
Melliès, P.-A., \emph{Categorical semantics of linear logic}, in:
  \emph{Interactive models of computation and program behaviour}, number~27 in
  Panoramas et {Synthèses}, pages 1 -- 196, Société Mathématique de France
  (2009).
\newline\urlprefix\url{https://hal.science/hal-00154229v1}

\bibitem{moon_graded_2021}
Moon, B., H.~Eades~III and D.~Orchard, \emph{Graded {Modal} {Dependent} {Type}
  {Theory}}, in: N.~Yoshida, editor, \emph{Programming {Languages} and
  {Systems}}, pages 462--490, Springer International Publishing, Cham (2021),
  ISBN 978-3-030-72019-3.
\newline\urlprefix\url{https://doi.org/10.1007/978-3-030-72019-3_17}

\bibitem{orchard_quantitative_2019}
Orchard, D., V.-B. Liepelt and H.~Eades~III, \emph{Quantitative program
  reasoning with graded modal types}, Language implementation for Quantitative
  Program Reasoning with Graded Modal Types \textbf{3}, pages 110:1--110:30
  (2019).
\newline\urlprefix\url{https://doi.org/10.1145/3341714}

\bibitem{petricek_coeffects_2013}
Petricek, T., D.~Orchard and A.~Mycroft, \emph{Coeffects: {Unified} {Static}
  {Analysis} of {Context}-{Dependence}}, in: F.~V. Fomin, R.~Freivalds,
  M.~Kwiatkowska and D.~Peleg, editors, \emph{Automata, {Languages}, and
  {Programming}}, pages 385--397, Springer Berlin Heidelberg, Berlin,
  Heidelberg (2013), ISBN 978-3-642-39212-2.
\newline\urlprefix\url{https://doi.org/10.1007/978-3-642-39212-2_35}

\bibitem{petricek_coeffects_2014}
Petricek, T., D.~Orchard and A.~Mycroft, \emph{Coeffects: a calculus of
  context-dependent computation}, ACM SIGPLAN Notices \textbf{49}, pages
  123--135 (2014), ISSN 0362-1340, 1558-1160.
\newline\urlprefix\url{https://doi.org/10.1145/2692915.2628160}

\bibitem{petric_coherence_2002}
Petrić, Z., \emph{Coherence in {Substructural} {Categories}}, Studia Logica
  \textbf{70}, pages 271--296 (2002), ISSN 1572-8730.
\newline\urlprefix\url{https://doi.org/10.1023/A:1015186718090}

\bibitem{pruiksma_adjointlogic_2018}
Pruiksma, K., W.~Chargin, F.~Pfenning and J.~Reed, \emph{Adjoint logic} (2018).
\newline\urlprefix\url{https://www.cs.cmu.edu/~./fp/papers/adjoint18b.pdf}

\bibitem{rajani_modal_2024}
Rajani, V., G.~Barthe and D.~Garg, \emph{A {Modal} {Type} {Theory} of
  {Expected} {Cost} in {Higher}-{Order} {Probabilistic} {Programs}}, Proc. ACM
  Program. Lang. \textbf{8}, pages 285:389--285:414 (2024).
\newline\urlprefix\url{https://doi.org/10.1145/3689725}

\bibitem{uustalu_signals_2005}
Uustalu, T. and T.~Vene, \emph{Signals and {Comonads}}, JUCS - Journal of
  Universal Computer Science \textbf{11}, pages 1310--1326 (2005), ISSN
  0948-6968, 0948-695X. Number: 7.
\newline\urlprefix\url{https://doi.org/10.3217/jucs-011-07-1311}

\bibitem{uustalu_comonadic_2008}
Uustalu, T. and V.~Vene, \emph{Comonadic {Notions} of {Computation}},
  Electronic Notes in Theoretical Computer Science \textbf{203}, pages 263--284
  (2008), ISSN 1571-0661.
\newline\urlprefix\url{https://doi.org/10.1016/j.entcs.2008.05.029}

\bibitem{vollmer_mixed_2025}
Vollmer, V., D.~Marshall, H.~Eades~III and D.~Orchard, \emph{A {Mixed} {Linear}
  and {Graded} {Logic}: {Proofs}, {Terms}, and {Models}}, in: J.~Endrullis and
  S.~Schmitz, editors, \emph{33rd {EACSL} {Annual} {Conference} on {Computer}
  {Science} {Logic} ({CSL} 2025)}, volume 326 of \emph{Leibniz {International}
  {Proceedings} in {Informatics} ({LIPIcs})}, pages 32:1--32:21, Schloss
  Dagstuhl – Leibniz-Zentrum für Informatik, Dagstuhl, Germany (2025), ISBN
  978-3-95977-362-1, ISSN 1868-8969.
\newline\urlprefix\url{https://doi.org/10.4230/LIPIcs.CSL.2025.32}

\end{thebibliography}
\newpage

\ifextended
\appendix
\section{Specification of the single mode system}
\label{asec:single-mode-spec}
  This section contains the specification for the single mode specialization of \grass{}.
The syntax for types and terms, and the corresponding typing rules are
\[
  \setlength{\arraycolsep}{.19em}
  \begin{array}{rrl}
    \textit{(Types)}
    &
      A, B, C, T
      ::=
    &
       \mathbf{I} 
      \mid  \grassSnt{A}  \otimes  \grassSnt{B} 
      \mid  \grassSnt{A}  \oplus  \grassSnt{B} 
      \mid  \grassSnt{A} ^{ \grassSnt{q} } \multimap  \grassSnt{B} 
      \mid  \square _{ \grassSnt{q} }  \grassSnt{A} 
    \\
    \textit{(Terms)} &
                       e, f, t
                       ::=
    &
      \grassSmv{x}
      \mid  \star 
      \mid  \mathbin{\mathsf{let} } _{@  \grassSnt{q} }   \star   =  \grassSnt{t}   \mathbin{\mathsf{in} }   \grassSnt{e} 
      \mid \grassSsym{(}  \grassSnt{e_{{\mathrm{1}}}}  \grassSsym{,}  \grassSnt{e_{{\mathrm{2}}}}  \grassSsym{)}
      \mid  \mathbin{\mathsf{let} } _{@  \grassSnt{q} }  \grassSsym{(}  \grassSmv{x}  \grassSsym{,}  \grassSmv{y}  \grassSsym{)}  =  \grassSnt{t}   \mathbin{\mathsf{in} }   \grassSnt{e} 
      \mid \operatorname{\mathsf{inl} } \, \grassSnt{t}
      \mid \operatorname{\mathsf{inr} } \, \grassSnt{t}
    \\
    &
    &
      \mid  \mathsf{case}_ \grassSnt{q} ( \grassSnt{t} ; \grassSmv{x_{{\mathrm{1}}}} . \grassSnt{e_{{\mathrm{1}}}} ; \grassSmv{x_{{\mathrm{2}}}} . \grassSnt{e_{{\mathrm{2}}}} ) 
      \mid \lambda  \grassSmv{x}  \grassSsym{.}  \grassSnt{t}
      \mid \grassSnt{f} \, \grassSnt{t}
      \mid  \square _{ \grassSnt{q} }  \grassSnt{t} 
      \mid  \mathbin{\mathsf{let} } _{@  \grassSnt{q} }  \square \, \grassSmv{x}  =  \grassSnt{t}   \mathbin{\mathsf{in} }   \grassSnt{e} 
  \end{array}
\]
\hrule
\begin{mathpar}
  \namedrules{grassS}{term}{var,weak,cont,sub,unitI,unitE,pairI,pairE,sumIL,sumIR,sumE,funI,funE,boxI,boxE}
\end{mathpar}

\newpage

\section{Syntactic properties of \grass{}}
\label{asec:multi-mode-syntax-proofs}

\subsection{Substitution theorem}
\label{asec:multi-mode-subst}
  \begin{theorem}
  If $  \rho  \grasssym{,}  \grassnt{r}  \mid  \mathsf{M}  \grasssym{,}  \mathsf{m}   \odot   \Gamma  \grasssym{,}  \grassmv{x}  \grasssym{:}  \grassnt{A}  \vdash_{ \mathsf{l} }  \grassnt{t}  \colon  \grassnt{B}  $
  and $  \sigma  \mid  \mathsf{N}   \odot   \Delta  \vdash_{ \mathsf{m} }  \grassnt{e}  \colon  \grassnt{A}  $,
  then $  \rho  \grasssym{,}   \grassnt{r}   \sigma   \mid  \mathsf{M}  \grasssym{,}  \mathsf{N}   \odot   \Gamma  \grasssym{,}  \Delta  \vdash_{ \mathsf{l} }  \grasssym{[}  \grassnt{e}  \grasssym{/}  \grassmv{x}  \grasssym{]}  \grassnt{t}  \colon  \grassnt{B}  $.
\end{theorem}

Due to the presence of an exlicit contraction rule,
we cannot prove this theorem directly.
Instead we must strengthen the inductive hypothesis to make the structural
induction work out.

\begin{theorem}[Simultaneous Subsitution]
  \label{athm:multi-subst}
  Suppose that $  \rho  \mid  \mathsf{M}   \odot   \Gamma  \vdash_{ \mathsf{m} }  \grassnt{t}  \colon  \grassnt{T}  $
  with $ \rho = \grassnt{r_{{\mathrm{1}}}}  \grasssym{,}  \mathellipsis  \grasssym{,}  \grassnt{r_{\grassmv{k}}} $,
  $ \mathsf{M} = \mathsf{m}_{{\mathrm{1}}}  \grasssym{,}  \mathellipsis  \grasssym{,}  \mathsf{m}_{\grassmv{k}} $ and
  $ \Gamma = \grassmv{x_{{\mathrm{1}}}}  \grasssym{:}  \grassnt{A_{{\mathrm{1}}}}  \grasssym{,}  \mathellipsis  \grasssym{,}  \grassmv{x_{\grassmv{k}}}  \grasssym{:}  \grassnt{A_{\grassmv{k}}} $.
  Suppose that for each $ i \in \{1 , \mathellipsis, k \} $ we have a term
  $  \sigma_{\grassmv{i}}  \mid  \mathsf{N}_{\grassmv{i}}   \odot   \Delta_{\grassmv{i}}  \vdash_{ \mathsf{m}_{\grassmv{i}} }  \grassnt{e_{\grassmv{i}}}  \colon  \grassnt{A_{\grassmv{i}}}  $.
  Then
  \[
	  \grassnt{r_{{\mathrm{1}}}}   \sigma_{{\mathrm{1}}}   \grasssym{,}  \mathellipsis  \grasssym{,}   \grassnt{r_{\grassmv{k}}}   \sigma_{\grassmv{k}}   \mid  \mathsf{N}_{{\mathrm{1}}}  \grasssym{,}  \mathellipsis  \grasssym{,}  \mathsf{N}_{\grassmv{k}}   \odot   \Delta_{{\mathrm{1}}}  \grasssym{,}  \mathellipsis  \grasssym{,}  \Delta_{\grassmv{k}}  \vdash_{ \mathsf{m} }  \grasssym{[}  \grassnt{e_{{\mathrm{1}}}}  \grasssym{/}  \grassmv{x_{{\mathrm{1}}}}  \grasssym{,} \, .. \, \grasssym{,}  \grassnt{e_{\grassmv{k}}}  \grasssym{/}  \grassmv{x_{\grassmv{k}}}  \grasssym{]}  \grassnt{t}  \colon  \grassnt{T} 
  \]
  We may abbreviate this term as $ \grasssym{[}  \grassnt{e_{\grassmv{i}}}  \grasssym{/}  \grassmv{x_{\grassmv{i}}}  \grasssym{]}  \grassnt{t} $.
\end{theorem}

\begin{proof}
  By induction on the derivation of $  \rho  \mid  \mathsf{M}   \odot   \Gamma  \vdash_{ \mathsf{m} }  \grassnt{t}  \colon  \grassnt{T}  $.

  \begin{case}[\normalfont\rulename{grass}{term}{var}\bf]\rm
    In this case $ k = 1 $ and $ T = A_1 $ and the derivation is
    \[
      \inferrule*
      {
      }
      {
          1   \mid  \mathsf{m}_{{\mathrm{1}}}   \odot   \grassmv{x_{{\mathrm{1}}}}  \grasssym{:}  \grassnt{A_{{\mathrm{1}}}}  \vdash_{ \mathsf{m} }  \grassmv{x_{{\mathrm{1}}}}  \colon  \grassnt{A_{{\mathrm{1}}}} 
      }.
    \]
    We must show that $  \sigma_{{\mathrm{1}}}  \mid  \mathsf{N}_{{\mathrm{1}}}   \odot   \Delta_{{\mathrm{1}}}  \vdash_{ \mathsf{m}_{{\mathrm{1}}} }  \grassnt{e_{{\mathrm{1}}}}  \colon  \grassnt{A_{{\mathrm{1}}}}  $ which is an
    assumption.

  \end{case}
  \begin{case}[\normalfont\rulename{grass}{term}{weak}\bf]\rm
    In this case the derivation is
    \[
      \inferrule*
      {
        \Weak (\mode m_{k+1})
        \and
        \mode m \le \mode m_{k+1}
        \and
        A_{k+1} \in \Type (\mode m_{k+1})
        \and
         \rho  \mid  \mathsf{M}   \odot   \Gamma  \vdash_{ \mathsf{m} }  \grassnt{t}  \colon  \grassnt{T} 
      }
      {
        \rho, 0 \mid \mode M, \mode m_{k+1} \at \Gamma, x_{k+1} : A_{k+1}
        \proves_{\mode m} t : T
      }
    \]
    The inductive hypothesis  is
    \[
      r_1 \cdot \sigma_1, \mathellipsis, r_k \cdot \sigma_k
      \mid \mode N_1 , \mathellipsis, \mode N_k
      \at  \Delta_1, \mathellipsis, \Delta_k
      \proves \grasssym{[}  \grassnt{e_{{\mathrm{1}}}}  \grasssym{/}  \grassmv{x_{{\mathrm{1}}}}  \grasssym{,} \, .. \, \grasssym{,}  \grassnt{e_{\grassmv{k}}}  \grasssym{/}  \grassmv{x_{\grassmv{k}}}  \grasssym{]}  \grassnt{t} : T
    \]
    and we need to show that
    \[
      r_1 \cdot \sigma_1, \mathellipsis, r_k \cdot \sigma_k, 0 \cdot \sigma_{k+1}
      \mid \mode N_1 , \mathellipsis, \mode N_k, \mode N_{k+1}
      \at  \Delta_1, \mathellipsis, \Delta_k, \Delta_{k+1}
      \proves [e_1/x_1, .., e_k/x_k,e_{k+1}/x_{k+1}]t : T.
    \]
    From
    $ \rho_{k+1} \mid \mode N_{k+1} \at \Delta_{k+1}
    \proves_{\mode m_{k+1}} e_{k+1} : A_{k+1} $
    it follows that $ \mode N \ge \mode m_{k+1} $.
    Together with $ \Weak \mode m_{k+1} $,
    this implies that $ \Weak (\mode n) $ holds for each $ \mode n $ in $ \mode N_{k+1} $.
    Hence we may apply a sequence of weakenings to to weaken by each type
    occuring in $ \Delta_{k+1} $ to obtain
    \begin{mathpar}
      r_1 \cdot \sigma_1, \mathellipsis, r_k \cdot \sigma_k, 0 \cdot \sigma_{k+1}
      \mid \mode N_1 , \mathellipsis, \mode N_k, \mode N_{k+1}
      \at  \Delta_1, \mathellipsis, \Delta_k, \Delta_{k+1}
      \proves [e_1/x_1, .., e_k/x_k]t : T.
    \end{mathpar}
    which is the desired judgment since $ x_{k+1} $ does not occur freely in $ t $.

  \end{case}
  \begin{case}[\normalfont\rulename{grass}{term}{sub}\bf]\rm
    In this case we have $ \rho' = r_1', \mathellipsis, r_k' $
    and $ r_i' \le r_i $ for each $ i $ and the derivation is
    \begin{mathpar}
      \inferrule*
      {
        \rho' \mid \mode M \at \Gamma \proves_{\mode m} t : T
      }
      {
        \rho \mid \mode M \at \Gamma \proves_{\mode m} t : T
      }
    \end{mathpar}
    By monotonicity of multiplication, we have
    $   \grassnt{r'_{\grassmv{i}}}   \sigma_{\grassmv{i}}   \le   \grassnt{r_{\grassmv{i}}}   \sigma_{\grassmv{i}}   $ for each $ i $
    and hence $   \grassnt{r'_{{\mathrm{1}}}}   \sigma_{{\mathrm{1}}}   \grasssym{,}  \mathellipsis  \grasssym{,}   \grassnt{r'_{\grassmv{k}}}   \sigma_{\grassmv{k}}   \le   \grassnt{r_{{\mathrm{1}}}}   \sigma_{{\mathrm{1}}}   \grasssym{,}  \mathellipsis  \grasssym{,}   \grassnt{r_{\grassmv{k}}}   \sigma_{\grassmv{k}}   $.
    The following application of the \rulename{grass}{term}{sub} rule
    \begin{mathpar}
      \inferrule*
      {
          \grassnt{r'_{{\mathrm{1}}}}   \sigma_{{\mathrm{1}}}   \grasssym{,}  \mathellipsis  \grasssym{,}   \grassnt{r'_{\grassmv{k}}}   \sigma_{\grassmv{k}}   \mid  \mathsf{N}_{{\mathrm{1}}}  \grasssym{,}  \mathellipsis  \grasssym{,}  \mathsf{N}_{\grassmv{k}}   \odot   \Delta_{{\mathrm{1}}}  \grasssym{,}  \mathellipsis  \grasssym{,}  \Delta_{\grassmv{k}}  \vdash_{ \mathsf{m} }  \grasssym{[}  \grassnt{e_{{\mathrm{1}}}}  \grasssym{/}  \grassmv{x_{{\mathrm{1}}}}  \grasssym{,} \, .. \, \grasssym{,}  \grassnt{e_{\grassmv{k}}}  \grasssym{/}  \grassmv{x_{\grassmv{k}}}  \grasssym{]}  \grassnt{t}  \colon  \grassnt{T} 
        \and
          \grassnt{r'_{{\mathrm{1}}}}   \sigma_{{\mathrm{1}}}   \grasssym{,}  \mathellipsis  \grasssym{,}   \grassnt{r'_{\grassmv{k}}}   \sigma_{\grassmv{k}}   \le   \grassnt{r_{{\mathrm{1}}}}   \sigma_{{\mathrm{1}}}   \grasssym{,}  \mathellipsis  \grasssym{,}   \grassnt{r_{\grassmv{k}}}   \sigma_{\grassmv{k}}  
      }
      {
          \grassnt{r_{{\mathrm{1}}}}   \sigma_{{\mathrm{1}}}   \grasssym{,}  \mathellipsis  \grasssym{,}   \grassnt{r_{\grassmv{k}}}   \sigma_{\grassmv{k}}   \mid  \mathsf{N}_{{\mathrm{1}}}  \grasssym{,}  \mathellipsis  \grasssym{,}  \mathsf{N}_{\grassmv{k}}   \odot   \Delta_{{\mathrm{1}}}  \grasssym{,}  \mathellipsis  \grasssym{,}  \Delta_{\grassmv{k}}  \vdash_{ \mathsf{m} }  \grasssym{[}  \grassnt{e_{{\mathrm{1}}}}  \grasssym{/}  \grassmv{x_{{\mathrm{1}}}}  \grasssym{,} \, .. \, \grasssym{,}  \grassnt{e_{\grassmv{k}}}  \grasssym{/}  \grassmv{x_{\grassmv{k}}}  \grasssym{]}  \grassnt{t}  \colon  \grassnt{T} 
      }
    \end{mathpar}
    yields the desired result.

  \end{case}
  \begin{case}[\normalfont\rulename{grass}{term}{cont}\bf]\rm
    The derivation has the shape
    \begin{mathpar}
      \inferrule*
      {
        \rho, r_{k}, r_{k+1}
        \mid \mode M, \mode m_{k}, \mode m_{k}
        \at \Gamma, x_k : A_{k}, x_{k+1} : A_{k}
        \proves_{\mode m} t: T
        \and
        r_{k}, r_{k+1} \in \Cont \mode m
      }
      {
        \rho, r_{k} + r_{k+1}
        \mid \mode M, \mode m_{k}
        \at \Gamma, z : A_{k}
        \proves_{\mode m} [z/x_k, z/x_{k+1}] t : T
      }
    \end{mathpar}
    We apply the inductive hypothesis with $ e_{k+1} = e_k $ to obtain
    \[
      r_1 \cdot \sigma_1, \mathellipsis, r_{k} \cdot \sigma_k , r_{k+1} \cdot \sigma_{k}
      \mid \mode N_1, \mathellipsis, \mode N_k, \mode N_k
      \at \Delta_1, \mathellipsis, \Delta_k, \Delta_k
      \proves_{\mode m} [e_1/x_1, \mathellipsis, e_k/x_k, e_k/x_{k+1}] t : T
    \]
    Since $ \Cont \mode m $ is an ideal, each entry in $ r_k \cdot \sigma_k $
    and $ r_{k+1} \cdot \sigma_k $is in $ \Cont \mode m $.
    A sequence of exchanges followed by a sequence of contractions yields
    \[
      r_1 \cdot \sigma_1, \mathellipsis, (r_{k} + r_{k+1}) \cdot \sigma_k
      \mid \mode N_1, \mathellipsis, \mode N_k
      \at \Delta_1, \mathellipsis, \Delta_k
      \proves_{\mode m} [e_1/x_1, \mathellipsis, e_k/x_k, e_k/x_{k+1}] t : T.
    \]
    The term in the desired result
    $ [e_1 / x_1, \mathellipsis, e_{k-1}/x_{k-1}, e_k/z][z/x_k, z/x_{k+1}]t $
    which is equal to the one above.

  \end{case}
  \begin{case}[\normalfont\rulename{grass}{term}{unitI}\bf]\rm
    There is nothing to prove in this case.

  \end{case}
  \begin{case}[\normalfont\rulename{grass}{term}{unitE}\bf]\rm
    In this case the derivation is
    \begin{mathpar}
      \inferrule*{
         \rho  \mid  \mathsf{M}   \odot   \Gamma  \vdash_{ \mathsf{m} }  \grassnt{t}  \colon  \grassnt{T} 
        \and
         \rho'  \mid  \mathsf{M}'   \odot   \Gamma'  \vdash_{ \mathsf{m} }  \grassnt{t'}  \colon   \mathbf{I}_ \mathsf{m}  
      }
      {
         \rho  \grasssym{,}   \grassnt{q}   \rho'   \mid  \mathsf{M}  \grasssym{,}  \mathsf{M}'   \odot   \Gamma  \grasssym{,}  \Gamma'  \vdash_{ \mathsf{m} }   \mathbin{\mathsf{let} } _{@  \grassnt{q} }   \star_ \mathsf{m}   =  \grassnt{t'}   \mathbin{\mathsf{in} }   \grassnt{t}   \colon  \grassnt{T} 
      }
    \end{mathpar}
    Write $ \rho = r_1, \mathellipsis, r_k $ and
    $ \rho' = r_1', \mathellipsis, r_{\ell}' $;
    Write $ \Gamma = x_1 : A_1, \mathellipsis, x_k : A_k $
    and $ \Gamma' = x_1' : A_1', \mathellipsis, x_\ell' : A_\ell' $;
    Write $ \mode M = \mode m_1, \mathellipsis, \mode m_k $
    and $ \mode M' = \mode m_1', \mathellipsis, \mode m_\ell' $.
    Given terms
    $  \sigma_{\grassmv{i}}  \mid  \mathsf{N}_{\grassmv{i}}   \odot   \Delta_{\grassmv{i}}  \vdash_{ \mathsf{m}_{\grassmv{i}} }  \grassnt{e_{\grassmv{i}}}  \colon  \grassnt{A_{\grassmv{i}}}  $ ($ i \in \{1, \mathellipsis, k\} $)
    and
    $  \sigma'_{\grassmv{i}}  \mid  \mathsf{N}_{\grassmv{i}}   \odot   \Delta'_{\grassmv{i}}  \vdash_{ \mathsf{m}'_{\grassmv{i}} }  \grassnt{e'_{\grassmv{i}}}  \colon  \grassnt{A'_{\grassmv{i}}}  $
    ($ i \in \{1, \mathellipsis, \ell \} $),
    we must show that
    \begin{align*}
        r_1 \sigma_1, \mathellipsis, r_k \sigma_k ,
        q r_1' \cdot \sigma_1', \mathellipsis, q r_k \cdot \sigma_\ell'
        \mid
        \mode N_1, \mathellipsis, \mode N_k,
        \mode N_1', \mathellipsis, \mode N_\ell'
      &
        \at
        \Delta_1, \mathellipsis, \Delta_k,
        \Delta_1', \mathellipsis, \Delta_\ell'
      \\
      &
        \proves_{\mode m}
         \mathbin{\mathsf{let} } _{@  \grassnt{q} }   \star_ \mathsf{m}   =  \grasssym{[}  \grassnt{e'_{\grassmv{i}}}  \grasssym{/}  \grassmv{x'_{\grassmv{i}}}  \grasssym{]}  \grassnt{t'}   \mathbin{\mathsf{in} }   \grasssym{[}  \grassnt{e_{\grassmv{i}}}  \grasssym{/}  \grassmv{x_{\grassmv{i}}}  \grasssym{]}  \grassnt{t}  : T
    \end{align*}
    By induction we have
    \begin{mathpar}
        \grassnt{r_{{\mathrm{1}}}}   \sigma_{{\mathrm{1}}}   \grasssym{,}  \mathellipsis  \grasssym{,}   \grassnt{r_{\grassmv{k}}}   \sigma_{\grassmv{k}}   \mid  \mathsf{N}_{{\mathrm{1}}}  \grasssym{,}  \mathellipsis  \grasssym{,}  \mathsf{N}_{\grassmv{k}}   \odot   \Delta_{{\mathrm{1}}}  \grasssym{,}  \mathellipsis  \grasssym{,}  \Delta_{\grassmv{k}}  \vdash_{ \mathsf{m} }  \grasssym{[}  \grassnt{e_{\grassmv{i}}}  \grasssym{/}  \grassmv{x_{\grassmv{i}}}  \grasssym{]}  \grassnt{t}  \colon  \grassnt{T} 
      \and
        \grassnt{r'_{{\mathrm{1}}}}   \sigma'_{{\mathrm{1}}}   \grasssym{,}  \mathellipsis  \grasssym{,}   \grassnt{r'_{\grassmv{\ell}}}   \sigma'_{\grassmv{\ell}}   \mid  \mathsf{N}'_{{\mathrm{1}}}  \grasssym{,}  \mathellipsis  \grasssym{,}  \mathsf{N}'_{\grassmv{\ell}}   \odot   \Delta'_{{\mathrm{1}}}  \grasssym{,}  \mathellipsis  \grasssym{,}  \Delta'_{\grassmv{\ell}}  \vdash_{ \mathsf{m} }  \grasssym{[}  \grassnt{e'_{\grassmv{i}}}  \grasssym{/}  \grassmv{x'_{\grassmv{i}}}  \grasssym{]}  \grassnt{t'}  \colon   \mathbf{I}_ \mathsf{m}  
    \end{mathpar}
    applying the rule \rulename{grass}{term}{unitE} to these judgments,
    we obtain the desired result, since
    $  \grassnt{q}   \grasssym{(}   \grassnt{r'_{{\mathrm{1}}}}   \sigma'_{{\mathrm{1}}}   \grasssym{,}  \mathellipsis  \grasssym{,}   \grassnt{r'_{\grassmv{\ell}}}   \sigma'_{\grassmv{\ell}}   \grasssym{)}  =   \grassnt{q}   \grassnt{r'_{{\mathrm{1}}}}    \sigma'_{{\mathrm{1}}}   \grasssym{,}  \mathellipsis  \grasssym{,}    \grassnt{q}   \grassnt{r'_{\grassmv{\ell}}}    \sigma'_{\grassmv{\ell}}  $.

  \end{case}
  \begin{case}[\normalfont\rulename{grass}{term}{arrowI}\bf]\rm
    The derivation is
    \[
      \inferrule*
      {
         \rho  \grasssym{,}  \grassnt{q}  \mid  \mathsf{M}  \grasssym{,}  \mathsf{n}   \odot   \Gamma  \grasssym{,}  \grassmv{y}  \grasssym{:}  \grassnt{B}  \vdash_{ \mathsf{m} }  \grassnt{t}  \colon  \grassnt{T} 
      }
      {
         \rho  \mid  \mathsf{M}   \odot   \Gamma  \vdash_{ \mathsf{m} }  \lambda  \grassmv{y}  \grasssym{.}  \grassnt{t}  \colon   \grassnt{B} ^{ \grassnt{q} : \mathsf{n} } \multimap  \grassnt{T}  
      }
    \]
    With $ \rho = r_1, \mathellipsis, r_k  $
    $ \mode M = \mathsf{m}_{{\mathrm{1}}}  \grasssym{,}  \mathellipsis  \grasssym{,}  \mathsf{m}_{\grassmv{k}} $
    and $ \Gamma = \grassmv{x_{{\mathrm{1}}}}  \grasssym{:}  \grassnt{A_{{\mathrm{1}}}}  \grasssym{,}  \mathellipsis  \grasssym{,}  \grassmv{x_{\grassmv{k}}}  \grasssym{:}  \grassnt{A_{\grassmv{k}}} $.
    $  \sigma_{\grassmv{i}}  \mid  \mathsf{N}_{\grassmv{i}}   \odot   \Delta_{\grassmv{i}}  \vdash_{ \mathsf{m}_{\grassmv{i}} }  \grassnt{e_{\grassmv{i}}}  \colon  \grassnt{A_{\grassmv{i}}}  $ ($ i \in \{1, \mathellipsis, k\} $)
    and we must prove
    \[
        \grassnt{r_{{\mathrm{1}}}}   \sigma_{{\mathrm{1}}}   \grasssym{,}  \mathellipsis  \grasssym{,}   \grassnt{r_{\grassmv{k}}}   \sigma_{\grassmv{k}}   \mid  \mathsf{N}_{{\mathrm{1}}}  \grasssym{,}  \mathellipsis  \grasssym{,}  \mathsf{N}_{\grassmv{k}}   \odot   \Delta_{{\mathrm{1}}}  \grasssym{,}  \mathellipsis  \grasssym{,}  \Delta_{\grassmv{k}}  \vdash_{ \mathsf{m} }  \lambda  \grassmv{y}  \grasssym{.}  \grasssym{[}  \grassnt{e_{\grassmv{i}}}  \grasssym{/}  \grassmv{x_{\grassmv{i}}}  \grasssym{]}  \grassnt{t}  \colon   \grassnt{B} ^{ \grassnt{q} : \mathsf{n} } \multimap  \grassnt{T}  
    \]
    We apply the inductive hypothesis with $ A_{k+1} = B $
    and $ e_{k+1} = (  1   \mid  \mathsf{n}   \odot   \grassmv{y}  \grasssym{:}  \grassnt{B}  \vdash_{ \mathsf{n} }  \grassmv{y}  \colon  \grassnt{B} ) $
    to obtain
    \[
        \grassnt{r_{{\mathrm{1}}}}   \sigma_{{\mathrm{1}}}   \grasssym{,}  \mathellipsis  \grasssym{,}   \grassnt{r_{\grassmv{k}}}   \sigma_{\grassmv{k}}   \grasssym{,}  \grassnt{q}  \mid  \mathsf{N}_{{\mathrm{1}}}  \grasssym{,}  \mathellipsis  \grasssym{,}  \mathsf{N}_{\grassmv{k}}  \grasssym{,}  \mathsf{n}   \odot   \Delta_{{\mathrm{1}}}  \grasssym{,}  \mathellipsis  \grasssym{,}  \Delta_{\grassmv{k}}  \grasssym{,}  \grassmv{y}  \grasssym{:}  \grassnt{B}  \vdash_{ \mathsf{m} }  \grasssym{[}  \grassnt{e_{\grassmv{i}}}  \grasssym{/}  \grassmv{x_{\grassmv{i}}}  \grasssym{]}  \grassnt{t}  \colon  \grassnt{T} .
    \]
    Applying rule \rulename{grass}{term}{arrowI} yields the desired result.

  \end{case}
  \begin{case}[\normalfont\rulename{grass}{term}{arrowE}\bf]\rm
    In this case the derivation has the shape
    \begin{mathpar}
      \inferrule*
      {
         \rho  \mid  \mathsf{M}   \odot   \Gamma  \vdash_{ \mathsf{m} }  \grassnt{t}  \colon   \grassnt{B} ^{ \grassnt{q} : \mathsf{m}' } \multimap  \grassnt{T}  
        \and
         \rho'  \mid  \mathsf{M}'   \odot   \Gamma'  \vdash_{ \mathsf{m}' }  \grassnt{t'}  \colon  \grassnt{B} 
      }
      {
         \rho  \grasssym{,}   \grassnt{q}   \rho'   \mid  \mathsf{M}  \grasssym{,}  \mathsf{M}'   \odot   \Gamma  \grasssym{,}  \Gamma'  \vdash_{ \mathsf{m} }  \grassnt{t} \, \grassnt{t'}  \colon  \grassnt{T} 
      }
    \end{mathpar}
    Write $ \rho = r_1, \mathellipsis, r_k $ and
    $ \rho' = r_1', \mathellipsis, r_{\ell}' $;
    Write $ \Gamma = x_1 : A_1, \mathellipsis, x_k : A_k $
    and $ \Gamma' = x_1' : A_1', \mathellipsis, x_\ell' : A_\ell' $;
    Write $ \mode M = \mode m_1, \mathellipsis, \mode m_k $
    and $ \mode M' = \mode m_1', \mathellipsis, \mode m_\ell' $.
    Given terms
    $  \sigma_{\grassmv{i}}  \mid  \mathsf{N}_{\grassmv{i}}   \odot   \Delta_{\grassmv{i}}  \vdash_{ \mathsf{m}_{\grassmv{i}} }  \grassnt{e_{\grassmv{i}}}  \colon  \grassnt{A_{\grassmv{i}}}  $ ($ i \in \{1, \mathellipsis, k\} $)
    and
    $  \sigma'_{\grassmv{i}}  \mid  \mathsf{N}_{\grassmv{i}}   \odot   \Delta'_{\grassmv{i}}  \vdash_{ \mathsf{m}'_{\grassmv{i}} }  \grassnt{e'_{\grassmv{i}}}  \colon  \grassnt{A'_{\grassmv{i}}}  $
    ($ i \in \{1, \mathellipsis, \ell \} $),
    we must show that
    \[
      r_1 \sigma_1, \mathellipsis, r_k \sigma_k ,
      q r_1' \cdot \sigma_1', \mathellipsis, q r_k \cdot \sigma_\ell'
      \mid
      \mode N_1, \mathellipsis, \mode N_k,
      \mode N_1', \mathellipsis, \mode N_\ell'
      \at
      \Delta_1, \mathellipsis, \Delta_k,
      \Delta_1', \mathellipsis, \Delta_\ell'
      \proves_{\mode m}
      \grasssym{[}  \grassnt{e_{\grassmv{i}}}  \grasssym{/}  \grassmv{x_{\grassmv{i}}}  \grasssym{]}  \grassnt{t} \, \grasssym{[}  \grassnt{e'_{\grassmv{i}}}  \grasssym{/}  \grassmv{x'_{\grassmv{i}}}  \grasssym{]}  \grassnt{t'} : T
    \]
    The inductive hypotheses are
    \begin{mathpar}
      r_1 \sigma_1, \mathellipsis, r_k \sigma_k ,
      \mid
      \mode N_1, \mathellipsis, \mode N_k,
      \at
      \Delta_1, \mathellipsis, \Delta_k,
      \proves_{\mode m}
      [e_i/x_i]t :  \grassnt{B} ^{ \grassnt{q} : \mathsf{m}' } \multimap  \grassnt{T} 
      \and
      r_1' \cdot \sigma_1', \mathellipsis,  r_k \cdot \sigma_\ell'
      \mid
      \mode N_1', \mathellipsis, \mode N_\ell'
      \at
      \Delta_1', \mathellipsis, \Delta_\ell'
      \proves_{\mode m}
      [e_i'/x_i']t' : B
    \end{mathpar}
    Applying the rule \rulename{grass}{term}{arrowE} yields the desired result.

  \end{case}
  \begin{case}[\normalfont\rulename{grass}{term}{pairI}\bf]\rm
    In this case the derivation has the shape
    \begin{mathpar}
      \inferrule*
      {
         \rho  \mid  \mathsf{M}   \odot   \Gamma  \vdash_{ \mathsf{m} }  \grassnt{t}  \colon  \grassnt{T} 
        \and
         \rho'  \mid  \mathsf{M}'   \odot   \Gamma'  \vdash_{ \mathsf{m} }  \grassnt{t'}  \colon  \grassnt{T'} 
      }
      {
         \rho  \grasssym{,}  \rho'  \mid  \mathsf{M}  \grasssym{,}  \mathsf{M}'   \odot   \Gamma  \grasssym{,}  \Gamma'  \vdash_{ \mathsf{m} }  \grasssym{(}  \grassnt{t}  \grasssym{,}  \grassnt{t'}  \grasssym{)}  \colon   \grassnt{T}  \otimes  \grassnt{T'}  
      }
    \end{mathpar}
    Write $ \rho = r_1, \mathellipsis, r_k $ and
    $ \rho' = r_1', \mathellipsis, r_{\ell}' $;
    Write $ \Gamma = x_1 : A_1, \mathellipsis, x_k : A_k $
    and $ \Gamma' = x_1' : A_1', \mathellipsis, x_\ell' : A_\ell' $;
    Write $ \mode M = \mode m_1, \mathellipsis, \mode m_k $
    and $ \mode M' = \mode m_1', \mathellipsis, \mode m_\ell' $.
    Given terms
    $  \sigma_{\grassmv{i}}  \mid  \mathsf{N}_{\grassmv{i}}   \odot   \Delta_{\grassmv{i}}  \vdash_{ \mathsf{m}_{\grassmv{i}} }  \grassnt{e_{\grassmv{i}}}  \colon  \grassnt{A_{\grassmv{i}}}  $ ($ i \in \{1, \mathellipsis, k\} $)
    and
    $  \sigma'_{\grassmv{i}}  \mid  \mathsf{N}_{\grassmv{i}}   \odot   \Delta'_{\grassmv{i}}  \vdash_{ \mathsf{m}'_{\grassmv{i}} }  \grassnt{e'_{\grassmv{i}}}  \colon  \grassnt{A'_{\grassmv{i}}}  $
    ($ i \in \{1, \mathellipsis, \ell \} $),
    we must show that
    \[
      r_1 \sigma_1, \mathellipsis, r_k \sigma_k ,
      r_1' \cdot \sigma_1', \mathellipsis, r_k \cdot \sigma_\ell'
      \mid
      \mode N_1, \mathellipsis, \mode N_k,
      \mode N_1', \mathellipsis, \mode N_\ell'
      \at
      \Delta_1, \mathellipsis, \Delta_k,
      \Delta_1', \mathellipsis, \Delta_\ell'
      \proves_{\mode m}
      \grasssym{(}  \grasssym{[}  \grassnt{e_{\grassmv{i}}}  \grasssym{/}  \grassmv{x_{\grassmv{i}}}  \grasssym{]}  \grassnt{t}  \grasssym{,}  \grasssym{[}  \grassnt{e'_{\grassmv{i}}}  \grasssym{/}  \grassmv{x'_{\grassmv{i}}}  \grasssym{]}  \grassnt{t'}  \grasssym{)} :  \grassnt{T}  \otimes  \grassnt{T'} 
    \]
    By induction we have
    \begin{mathpar}
      r_1 \sigma_1, \mathellipsis, r_k \sigma_k ,
      \mid
      \mode N_1, \mathellipsis, \mode N_k,
      \at
      \Delta_1, \mathellipsis, \Delta_k,
      \proves_{\mode m}
      \grasssym{[}  \grassnt{e_{\grassmv{i}}}  \grasssym{/}  \grassmv{x_{\grassmv{i}}}  \grasssym{]}  \grassnt{t} : T
      \and
      r_1' \cdot \sigma_1', \mathellipsis, r_k \cdot \sigma_\ell'
      \mid
      \mode N_1', \mathellipsis, \mode N_\ell'
      \at
      \Delta_1', \mathellipsis, \Delta_\ell'
      \proves_{\mode m}
      \grasssym{[}  \grassnt{e'_{\grassmv{i}}}  \grasssym{/}  \grassmv{x'_{\grassmv{i}}}  \grasssym{]}  \grassnt{t'} : T'
    \end{mathpar}
    Applying the rule \rulename{grass}{term}{pairI} yields the desired result.

  \end{case}
  \begin{case}[\normalfont\rulename{grass}{term}{pairE}\bf]\rm
    In this case the derivation is
    \begin{mathpar}
      \inferrule*{
         \rho  \grasssym{,}  \grassnt{q}  \grasssym{,}  \grassnt{q}  \mid  \mathsf{M}  \grasssym{,}  \mathsf{n}  \grasssym{,}  \mathsf{n}   \odot   \Gamma  \grasssym{,}  \grassmv{y_{{\mathrm{1}}}}  \grasssym{:}  \grassnt{T_{{\mathrm{1}}}}  \grasssym{,}  \grassmv{y_{{\mathrm{2}}}}  \grasssym{:}  \grassnt{T_{{\mathrm{2}}}}  \vdash_{ \mathsf{m} }  \grassnt{t}  \colon  \grassnt{T} 
        \and
         \rho'  \mid  \mathsf{M}'   \odot   \Gamma'  \vdash_{ \mathsf{n} }  \grassnt{t'}  \colon   \grassnt{T_{{\mathrm{1}}}}  \otimes  \grassnt{T_{{\mathrm{2}}}}  
      }
      {
         \rho  \grasssym{,}   \grassnt{q}   \rho'   \mid  \mathsf{M}  \grasssym{,}  \mathsf{M}'   \odot   \Gamma  \grasssym{,}  \Gamma'  \vdash_{ \mathsf{m} }   \mathbin{\mathsf{let} } _{@  \grassnt{q} }  \grasssym{(}  \grassmv{y_{{\mathrm{1}}}}  \grasssym{,}  \grassmv{y_{{\mathrm{2}}}}  \grasssym{)}  =  \grassnt{t'}   \mathbin{\mathsf{in} }   \grassnt{t}   \colon  \grassnt{T} 
      }
    \end{mathpar}
    Write $ \rho = r_1, \mathellipsis, r_k $ and
    $ \rho' = r_1', \mathellipsis, r_{\ell}' $;
    Write $ \Gamma = x_1 : A_1, \mathellipsis, x_k : A_k $
    and $ \Gamma' = x_1' : A_1', \mathellipsis, x_\ell' : A_\ell' $;
    Write $ \mode M = \mode m_1, \mathellipsis, \mode m_k $
    and $ \mode M' = \mode m_1', \mathellipsis, \mode m_\ell' $.
    Given terms
    $  \sigma_{\grassmv{i}}  \mid  \mathsf{N}_{\grassmv{i}}   \odot   \Delta_{\grassmv{i}}  \vdash_{ \mathsf{m}_{\grassmv{i}} }  \grassnt{e_{\grassmv{i}}}  \colon  \grassnt{A_{\grassmv{i}}}  $ ($ i \in \{1, \mathellipsis, k\} $)
    and
    $  \sigma'_{\grassmv{i}}  \mid  \mathsf{N}_{\grassmv{i}}   \odot   \Delta'_{\grassmv{i}}  \vdash_{ \mathsf{m}'_{\grassmv{i}} }  \grassnt{e'_{\grassmv{i}}}  \colon  \grassnt{A'_{\grassmv{i}}}  $
    ($ i \in \{1, \mathellipsis, \ell \} $),
    we must show that
    \begin{align*}
      r_1 \sigma_1, \mathellipsis, r_k \sigma_k ,
      q r_1' \cdot \sigma_1', \mathellipsis, q r_k \cdot \sigma_\ell'
      \mid
      \mode N_1, \mathellipsis, \mode N_k,
      \mode N_1', \mathellipsis, \mode N_\ell'
      &
        \at
        \Delta_1, \mathellipsis, \Delta_k,
        \Delta_1', \mathellipsis, \Delta_\ell'
      \\
      &
        \proves_{\mode m}
         \mathbin{\mathsf{let} } _{@  \grassnt{q} }  \grasssym{(}  \grassmv{y_{{\mathrm{1}}}}  \grasssym{,}  \grassmv{y_{{\mathrm{2}}}}  \grasssym{)}  =  \grasssym{[}  \grassnt{e'_{\grassmv{i}}}  \grasssym{/}  \grassmv{x'_{\grassmv{i}}}  \grasssym{]}  \grassnt{t'}   \mathbin{\mathsf{in} }   \grasssym{[}  \grassnt{e_{\grassmv{i}}}  \grasssym{/}  \grassmv{x_{\grassmv{i}}}  \grasssym{]}  \grassnt{t}  : T
    \end{align*}
    Applying the inductive hypothesis to $  \rho'  \mid  \mathsf{M}'   \odot   \Gamma'  \vdash_{ \mathsf{n} }  \grassnt{t'}  \colon   \grassnt{T_{{\mathrm{1}}}}  \otimes  \grassnt{T_{{\mathrm{2}}}}   $,
    we obtain
    \begin{mathpar}
      r_1' \cdot \sigma_1', \mathellipsis, r_k \cdot \sigma_\ell'
      \mid
      \mode N_1', \mathellipsis, \mode N_\ell'
      \at
      \Delta_1', \mathellipsis, \Delta_\ell'
      \proves_{\mode m}
      \grasssym{[}  \grassnt{e'_{\grassmv{i}}}  \grasssym{/}  \grassmv{x'_{\grassmv{i}}}  \grasssym{]}  \grassnt{t'} :  \grassnt{T_{{\mathrm{1}}}}  \otimes  \grassnt{T_{{\mathrm{2}}}} 
    \end{mathpar}
    Furthermore, applying the inductive hypothesis to the judgment
    $  \rho  \grasssym{,}  \grassnt{q}  \grasssym{,}  \grassnt{q}  \mid  \mathsf{M}  \grasssym{,}  \mathsf{n}  \grasssym{,}  \mathsf{n}   \odot   \Gamma  \grasssym{,}  \grassmv{y_{{\mathrm{1}}}}  \grasssym{:}  \grassnt{T_{{\mathrm{1}}}}  \grasssym{,}  \grassmv{y_{{\mathrm{2}}}}  \grasssym{:}  \grassnt{T_{{\mathrm{2}}}}  \vdash_{ \mathsf{m} }  \grassnt{t}  \colon  \grassnt{T}  $
    with the terms $ e_1, \mathellipsis, e_k $ and
    $   1   \mid  \mathsf{n}   \odot   \grassmv{y_{\grassmv{i}}}  \grasssym{:}  \grassnt{T_{\grassmv{i}}}  \vdash_{ \mathsf{n} }  \grassmv{y_{\grassmv{i}}}  \colon  \grassnt{T_{\grassmv{i}}}  $ ($ i \in \{1, 2\} $)
    yields the judgment
    \[
      r_1 \sigma_1, \mathellipsis, r_k \sigma_k, q, q
      \mid
      \mode N_1, \mathellipsis, \mode N_k, \mode n, \mode n
      \at
      \Delta_1, \mathellipsis, \Delta_k, y_1 : T_1, y_2 : T_2
      \proves_{\mode m}
      \grasssym{[}  \grassnt{e_{\grassmv{i}}}  \grasssym{/}  \grassmv{x_{\grassmv{i}}}  \grasssym{]}  \grassnt{t} : T
    \]
    Now, applying the rule \rulename{grass}{term}{pairE} to the two inductive
    hypotheses yields the desired result.

  \end{case}
  \begin{case}[\normalfont\rulename{grass}{term}{sumIL}\bf]\rm
    The derivation has the form
    \[
      \inferrule*
      {
         \grassnt{T_{{\mathrm{2}}}}  \in \mathsf{Type}( \mathsf{m} ) 
        \and
         \rho  \mid  \mathsf{M}   \odot   \Gamma  \vdash_{ \mathsf{m} }  \grassnt{t}  \colon  \grassnt{T_{{\mathrm{1}}}} 
      }
      {
         \rho  \mid  \mathsf{M}   \odot   \Gamma  \vdash_{ \mathsf{m} }  \operatorname{\mathsf{inl} } \, \grassnt{t}  \colon   \grassnt{T_{{\mathrm{1}}}}  \oplus  \grassnt{T_{{\mathrm{2}}}}  
      }
    \]
    The inductive hypothesis is
    \[
        \grassnt{r_{{\mathrm{1}}}}   \sigma_{{\mathrm{1}}}   \grasssym{,}  \mathellipsis  \grasssym{,}   \grassnt{r_{\grassmv{k}}}   \sigma_{\grassmv{k}}   \mid  \mathsf{N}_{{\mathrm{1}}}  \grasssym{,}  \mathellipsis  \grasssym{,}  \mathsf{N}_{\grassmv{k}}   \odot   \Delta_{{\mathrm{1}}}  \grasssym{,}  \mathellipsis  \grasssym{,}  \Delta_{\grassmv{k}}  \vdash_{ \mathsf{m} }  \grasssym{[}  \grassnt{e_{\grassmv{i}}}  \grasssym{/}  \grassmv{x_{\grassmv{i}}}  \grasssym{]}  \grassnt{t}  \colon  \grassnt{T_{{\mathrm{1}}}} 
    \]
    applying the rule \rulename{grass}{term}{sumIL} yields the desired result.

  \end{case}
  \begin{case}[\normalfont\rulename{grass}{term}{sumIR}\bf]\rm
    Analogous to \rulename{grass}{term}{sumIL}.

  \end{case}
  \begin{case}[\normalfont\rulename{grass}{term}{sumE}\bf]\rm
    The derivation has the following shape
    \begin{mathpar}
      \inferrule*
      {
         \grassnt{q}  \ge   1  
        \and
         \rho  \grasssym{,}  \grassnt{q}  \mid  \mathsf{M}  \grasssym{,}  \mathsf{n}   \odot   \Gamma  \grasssym{,}  \grassmv{y_{{\mathrm{1}}}}  \grasssym{:}  \grassnt{T_{{\mathrm{1}}}}  \vdash_{ \mathsf{m} }  \grassnt{t_{{\mathrm{1}}}}  \colon  \grassnt{T} 
        \and
         \rho  \grasssym{,}  \grassnt{q}  \mid  \mathsf{M}  \grasssym{,}  \mathsf{n}   \odot   \Gamma  \grasssym{,}  \grassmv{y_{{\mathrm{2}}}}  \grasssym{:}  \grassnt{T_{{\mathrm{2}}}}  \vdash_{ \mathsf{m} }  \grassnt{t_{{\mathrm{2}}}}  \colon  \grassnt{T} 
        \and
         \rho'  \mid  \mathsf{M}'   \odot   \Gamma'  \vdash_{ \mathsf{n} }  \grassnt{t}  \colon   \grassnt{T_{{\mathrm{1}}}}  \oplus  \grassnt{T_{{\mathrm{2}}}}  
      }
      {
         \rho  \grasssym{,}   \grassnt{q}   \sigma   \mid  \mathsf{M}  \grasssym{,}  \mathsf{M}'   \odot   \Gamma  \grasssym{,}  \Gamma'  \vdash_{ \mathsf{m} }   \mathsf{case}_ \grassnt{q} ( \grassnt{t} ; \grassmv{y_{{\mathrm{1}}}} . \grassnt{t_{{\mathrm{1}}}} ; \grassmv{y_{{\mathrm{2}}}} . \grassnt{t_{{\mathrm{2}}}} )   \colon  \grassnt{T} 
      }
    \end{mathpar}
    Write $ \rho = r_1, \mathellipsis, r_k $ and
    $ \rho' = r_1', \mathellipsis, r_{\ell}' $;
    Write $ \Gamma = x_1 : A_1, \mathellipsis, x_k : A_k $
    and $ \Gamma' = x_1' : A_1', \mathellipsis, x_\ell' : A_\ell' $;
    Write $ \mode M = \mode m_1, \mathellipsis, \mode m_k $
    and $ \mode M' = \mode m_1', \mathellipsis, \mode m_\ell' $.
    Given terms
    $  \sigma_{\grassmv{i}}  \mid  \mathsf{N}_{\grassmv{i}}   \odot   \Delta_{\grassmv{i}}  \vdash_{ \mathsf{m}_{\grassmv{i}} }  \grassnt{e_{\grassmv{i}}}  \colon  \grassnt{A_{\grassmv{i}}}  $ ($ i \in \{1, \mathellipsis, k\} $)
    and
    $  \sigma'_{\grassmv{i}}  \mid  \mathsf{N}_{\grassmv{i}}   \odot   \Delta'_{\grassmv{i}}  \vdash_{ \mathsf{m}'_{\grassmv{i}} }  \grassnt{e'_{\grassmv{i}}}  \colon  \grassnt{A'_{\grassmv{i}}}  $
    ($ i \in \{1, \mathellipsis, \ell \} $),
    we must show that
    \begin{align*}
      r_1 \sigma_1, \mathellipsis, r_k \sigma_k ,
      q r_1' \cdot \sigma_1', \mathellipsis, q r_k \cdot \sigma_\ell'
      \mid
      \mode N_1, \mathellipsis, \mode N_k,
      \mode N_1', \mathellipsis, \mode N_\ell'
      &
        \at
        \Delta_1, \mathellipsis, \Delta_k,
        \Delta_1', \mathellipsis, \Delta_\ell'
      \\
      &
        \proves_{\mode m}
         \mathsf{case}_ \grassnt{q} ( \grasssym{[}  \grassnt{e'_{\grassmv{i}}}  \grasssym{/}  \grassmv{x'_{\grassmv{i}}}  \grasssym{]}  \grassnt{t} ; \grassmv{y_{{\mathrm{1}}}} . \grasssym{[}  \grassnt{e_{\grassmv{i}}}  \grasssym{/}  \grassmv{x_{\grassmv{i}}}  \grasssym{]}  \grassnt{t_{{\mathrm{1}}}} ; \grassmv{y_{{\mathrm{2}}}} . \grasssym{[}  \grassnt{e_{\grassmv{i}}}  \grasssym{/}  \grassmv{x_{\grassmv{i}}}  \grasssym{]}  \grassnt{t_{{\mathrm{2}}}} )  : T
    \end{align*}
    Applying the inductive hypothesis to $ t' $ yields the judgment
    \[
      r_1' \cdot \sigma_1', \mathellipsis, r_k \cdot \sigma_\ell'
      \mid
      \mode N_1', \mathellipsis, \mode N_\ell'
      \at
      \Delta_1', \mathellipsis, \Delta_\ell'
      \proves_{\mode m}
      \grasssym{[}  \grassnt{e'_{\grassmv{i}}}  \grasssym{/}  \grassmv{x'_{\grassmv{i}}}  \grasssym{]}  \grassnt{t} :  \grassnt{T_{{\mathrm{1}}}}  \oplus  \grassnt{T_{{\mathrm{2}}}} 
    \]
    Applying the inductive hypothesis to $ t_i $ ($ i \in \{1, 2\} $)
    with $ e_1, \mathellipsis, e_k $ and
    $ e_{k+1} =   1   \mid  \mathsf{n}   \odot   \grassmv{y_{\grassmv{i}}}  \grasssym{:}  \grassnt{T_{\grassmv{i}}}  \vdash_{ \mathsf{n} }  \grassmv{y_{\grassmv{i}}}  \colon  \grassnt{T_{\grassmv{i}}}  $
    yields the judgments
    \begin{mathpar}
      r_1 \sigma_1, \mathellipsis, r_k \sigma_k, q
      \mid
      \mode N_1, \mathellipsis, \mode N_k, \mode n
      \at
      \Delta_1, \mathellipsis, \Delta_k, y_1 : T_1
      \proves_{\mode n}
      \grasssym{[}  \grassnt{e_{\grassmv{i}}}  \grasssym{/}  \grassmv{x_{\grassmv{i}}}  \grasssym{]}  \grassnt{t_{{\mathrm{1}}}} : T
      \and
      r_1 \sigma_1, \mathellipsis, r_k \sigma_k, q
      \mid
      \mode N_1, \mathellipsis, \mode N_k, \mode n
      \at
      \Delta_1, \mathellipsis, \Delta_k, y_2 : T_2
      \proves_{\mode n}
      \grasssym{[}  \grassnt{e_{\grassmv{i}}}  \grasssym{/}  \grassmv{x_{\grassmv{i}}}  \grasssym{]}  \grassnt{t_{{\mathrm{2}}}} : T.
    \end{mathpar}
    Now applying the rule \rulename{grass}{term}{sumE} yields the desired result.

  \end{case}
  \begin{case}[\normalfont\rulename{grass}{term}{raiseI}\bf]\rm
    The derivation has the shape
    \[
      \inferrule*
      {
        \mathsf{m}  \le  \mathsf{n}
        \and
        \mathsf{n}  \le  \mathsf{M}
        \and
         \rho  \mid  \mathsf{M}   \odot   \Gamma  \vdash_{ \mathsf{m} }  \grassnt{t}  \colon  \grassnt{T} 
      }
      {
         \rho  \mid  \mathsf{M}   \odot   \Gamma  \vdash_{ \mathsf{n} }   \operatorname{\uparrow} _{ \mathsf{m}  \le  \mathsf{n} }  \grassnt{t}   \colon   \operatorname{\uparrow} _{ \mathsf{m}  \le  \mathsf{n} }  \grassnt{T}  
      }
    \]
    The inductive hypothesis is
    \[
        \grassnt{r_{{\mathrm{1}}}}   \sigma_{{\mathrm{1}}}   \grasssym{,}  \mathellipsis  \grasssym{,}   \grassnt{r_{\grassmv{k}}}   \sigma_{\grassmv{k}}   \mid  \mathsf{N}_{{\mathrm{1}}}  \grasssym{,}  \mathellipsis  \grasssym{,}  \mathsf{N}_{\grassmv{k}}   \odot   \Delta_{{\mathrm{1}}}  \grasssym{,}  \mathellipsis  \grasssym{,}  \Delta_{\grassmv{k}}  \vdash_{ \mathsf{m} }  \grasssym{[}  \grassnt{e_{\grassmv{i}}}  \grasssym{/}  \grassmv{x_{\grassmv{i}}}  \grasssym{]}  \grassnt{t}  \colon  \grassnt{T} .
    \]
    If we can apply the rule \rulename{grass}{term}{raiseI},
    we will obtain the desired result.
    But first we must verify that $ \mathsf{n}  \le  \mathsf{N}_{{\mathrm{1}}}  \grasssym{,}  \mathellipsis  \grasssym{,}  \mathsf{N}_{\grassmv{k}} $.
    To do so, it is sufficient to verify $ \mathsf{n}  \le  \mathsf{N}_{\grassmv{i}} $ for each $ i $.
    But this follows from transitiviy of $ \le $ and the fact that
    $ \mode n < \mode m_i < \mode N_i $.
    Where the second inequality is a consequence of the assumption
    $  \sigma_{\grassmv{i}}  \mid  \mathsf{N}_{\grassmv{i}}   \odot   \Delta_{\grassmv{i}}  \vdash_{ \mathsf{m}_{\grassmv{i}} }  \grassnt{e_{\grassmv{i}}}  \colon  \grassnt{A_{\grassmv{i}}}  $.

  \end{case}
  \begin{case}[\normalfont\rulename{grass}{term}{raiseE}\bf]\rm
    The derivation has the form
    \[
      \inferrule*
      {
         \rho  \mid  \mathsf{M}   \odot   \Gamma  \vdash_{ \mathsf{n} }  \grassnt{t}  \colon   \operatorname{\uparrow} _{ \mathsf{m}  \le  \mathsf{n} }  \grassnt{T}  
      }
      {
         \rho  \mid  \mathsf{M}   \odot   \Gamma  \vdash_{ \mathsf{m} }   \operatorname{\uparrow}^{-1} _{ \mathsf{m}  \le  \mathsf{n} }  \grassnt{t}   \colon  \grassnt{T} 
      }
    \]
    The inductive hypothesis is
    \[
        \grassnt{r_{{\mathrm{1}}}}   \sigma_{{\mathrm{1}}}   \grasssym{,}  \mathellipsis  \grasssym{,}   \grassnt{r_{\grassmv{k}}}   \sigma_{\grassmv{k}}   \mid  \mathsf{N}_{{\mathrm{1}}}  \grasssym{,}  \mathellipsis  \grasssym{,}  \mathsf{N}_{\grassmv{k}}   \odot   \Delta_{{\mathrm{1}}}  \grasssym{,}  \mathellipsis  \grasssym{,}  \Delta_{\grassmv{k}}  \vdash_{ \mathsf{n} }  \grasssym{[}  \grassnt{e_{\grassmv{i}}}  \grasssym{/}  \grassmv{x_{\grassmv{i}}}  \grasssym{]}  \grassnt{t}  \colon   \operatorname{\uparrow} _{ \mathsf{m}  \le  \mathsf{n} }  \grassnt{T}  
    \]
    and applying the rule \rulename{grass}{term}{raiseE} yields  the desired result.

  \end{case}
  \begin{case}[\normalfont\rulename{grass}{term}{dropI}\bf]\rm
    The derivation has the shape
    \[
      \inferrule*
      {
        \mathsf{n}  \le  \mathsf{m}
        \and
         \grassnt{q}  \in R_{ \mathsf{m} } 
        \and
         \rho  \mid  \mathsf{M}   \odot   \Gamma  \vdash_{ \mathsf{m} }  \grassnt{t}  \colon  \grassnt{T} 
      }
      {
          \grassnt{q}   \rho   \mid  \mathsf{M}   \odot   \Gamma  \vdash_{ \mathsf{n} }   \operatorname{\downarrow} ^{ \grassnt{q} }_{ \mathsf{n}  \le  \mathsf{m} }  \grassnt{t}   \colon   \operatorname{\downarrow} ^{ \grassnt{q} }_{ \mathsf{n}  \le  \mathsf{m} }  \grassnt{T}  
      }
    \]
    The inductive hypothesis is
    \[
        \grassnt{r_{{\mathrm{1}}}}   \sigma_{{\mathrm{1}}}   \grasssym{,}  \mathellipsis  \grasssym{,}   \grassnt{r_{\grassmv{k}}}   \sigma_{\grassmv{k}}   \mid  \mathsf{N}_{{\mathrm{1}}}  \grasssym{,}  \mathellipsis  \grasssym{,}  \mathsf{N}_{\grassmv{k}}   \odot   \Delta_{{\mathrm{1}}}  \grasssym{,}  \mathellipsis  \grasssym{,}  \Delta_{\grassmv{k}}  \vdash_{ \mathsf{m} }  \grasssym{[}  \grassnt{e_{\grassmv{i}}}  \grasssym{/}  \grassmv{x_{\grassmv{i}}}  \grasssym{]}  \grassnt{t}  \colon  \grassnt{T} 
    \]
    and applying the rule \rulename{grass}{term}{dropI} yields the desired result.

  \end{case}
  \begin{case}[\normalfont\rulename{grass}{term}{dropE}\bf]\rm
    The derivation has the shape
    \[
      \inferrule*
      {
        \mathsf{l}  \le  \mathsf{n}
        \and
        \mathsf{n}  \le  \mathsf{m}
        \and
         \rho  \grasssym{,}  \grassnt{q}  \mid  \mathsf{M}  \grasssym{,}  \mathsf{m}   \odot   \Gamma  \grasssym{,}  \grassmv{y}  \grasssym{:}  \grassnt{T'}  \vdash_{ \mathsf{l} }  \grassnt{t}  \colon  \grassnt{T} 
        \and
         \rho'  \mid  \mathsf{M}'   \odot   \Gamma'  \vdash_{ \mathsf{n} }  \grassnt{t'}  \colon   \operatorname{\downarrow} ^{ \grassnt{q} }_{ \mathsf{n}  \le  \mathsf{m} }  \grassnt{T'}  
      }
      {
         \rho  \grasssym{,}  \rho'  \mid  \mathsf{M}  \grasssym{,}  \mathsf{M}'   \odot   \Gamma  \grasssym{,}  \Gamma'  \vdash_{ \mathsf{l} }   \mathbin{\mathsf{let} } _{@  \grassnt{q} }   \operatorname{\downarrow} _{ \mathsf{n}  \le  \mathsf{m} }  \grassmv{y}   =  \grassnt{t'}   \mathbin{\mathsf{in} }   \grassnt{t}   \colon  \grassnt{T} 
      }
    \]
    Write $ \rho = r_1, \mathellipsis, r_k $ and
    $ \rho' = r_1', \mathellipsis, r_{\ell}' $;
    Write $ \Gamma = x_1 : A_1, \mathellipsis, x_k : A_k $
    and $ \Gamma' = x_1' : A_1', \mathellipsis, x_\ell' : A_\ell' $;
    Write $ \mode M = \mode m_1, \mathellipsis, \mode m_k $
    and $ \mode M' = \mode m_1', \mathellipsis, \mode m_\ell' $.
    Given terms
    $  \sigma_{\grassmv{i}}  \mid  \mathsf{N}_{\grassmv{i}}   \odot   \Delta_{\grassmv{i}}  \vdash_{ \mathsf{m}_{\grassmv{i}} }  \grassnt{e_{\grassmv{i}}}  \colon  \grassnt{A_{\grassmv{i}}}  $ ($ i \in \{1, \mathellipsis, k\} $)
    and
    $  \sigma'_{\grassmv{i}}  \mid  \mathsf{N}_{\grassmv{i}}   \odot   \Delta'_{\grassmv{i}}  \vdash_{ \mathsf{m}'_{\grassmv{i}} }  \grassnt{e'_{\grassmv{i}}}  \colon  \grassnt{A'_{\grassmv{i}}}  $
    ($ i \in \{1, \mathellipsis, \ell \} $),
    we must show that
    \begin{align*}
       \grassnt{r_{{\mathrm{1}}}}   \sigma_{{\mathrm{1}}}   \grasssym{,}  \mathellipsis  \grasssym{,}   \grassnt{r_{\grassmv{k}}}   \sigma_{\grassmv{k}}   \grasssym{,}   \grassnt{r'_{{\mathrm{1}}}}   \sigma'_{{\mathrm{1}}}   \grasssym{,}  \mathellipsis  \grasssym{,}   \grassnt{r_{\grassmv{k}}}   \sigma'_{\grassmv{\ell}} 
      \mid
      \mathsf{N}_{{\mathrm{1}}}  \grasssym{,}  \mathellipsis  \grasssym{,}  \mathsf{N}_{\grassmv{k}}  \grasssym{,}  \mathsf{N}'_{{\mathrm{1}}}  \grasssym{,}  \mathellipsis  \grasssym{,}  \mathsf{N}'_{\grassmv{\ell}}
      &
      \at
      \Delta_{{\mathrm{1}}}  \grasssym{,}  \mathellipsis  \grasssym{,}  \Delta_{\grassmv{k}}  \grasssym{,}  \Delta'_{{\mathrm{1}}}  \grasssym{,}  \mathellipsis  \grasssym{,}  \Delta'_{\grassmv{\ell}}
      \\
      &
      \proves_{\mode m}
       \mathbin{\mathsf{let} } _{@  \grassnt{q} }   \operatorname{\downarrow} _{ \mathsf{n}  \le  \mathsf{m} }  \grassmv{y}   =  \grasssym{[}  \grassnt{e'_{\grassmv{i}}}  \grasssym{/}  \grassmv{x'_{\grassmv{i}}}  \grasssym{]}  \grassnt{t'}   \mathbin{\mathsf{in} }   \grasssym{[}  \grassnt{e_{\grassmv{i}}}  \grasssym{/}  \grassmv{x_{\grassmv{i}}}  \grasssym{]}  \grassnt{t}  : T
    \end{align*}
    Applying the inductive hypothesis to $ t' $, we obtain
    \[
        \grassnt{r'_{{\mathrm{1}}}}   \sigma'_{{\mathrm{1}}}   \grasssym{,}  \mathellipsis  \grasssym{,}   \grassnt{r_{\grassmv{k}}}   \sigma'_{\grassmv{\ell}}   \mid  \mathsf{N}'_{{\mathrm{1}}}  \grasssym{,}  \mathellipsis  \grasssym{,}  \mathsf{N}'_{\grassmv{\ell}}   \odot   \Delta'_{{\mathrm{1}}}  \grasssym{,}  \mathellipsis  \grasssym{,}  \Delta'_{\grassmv{\ell}}  \vdash_{ \mathsf{m} }  \grasssym{[}  \grassnt{e'_{\grassmv{i}}}  \grasssym{/}  \grassmv{x'_{\grassmv{i}}}  \grasssym{]}  \grassnt{t'}  \colon  \grassnt{T'} .
    \]
    Furthermore, applying the inductive hypothesis to $ t $
    with $ e_1, \mathellipsis, e_k $ and
    $ e_{k+1} =   1   \mid  \mathsf{m}   \odot   \grassmv{y}  \grasssym{:}  \grassnt{T'}  \vdash_{ \mathsf{m} }  \grassmv{y}  \colon  \grassnt{T'}  $
    yields
    \[
        \grassnt{r_{{\mathrm{1}}}}   \sigma_{{\mathrm{1}}}   \grasssym{,}  \mathellipsis  \grasssym{,}   \grassnt{r_{\grassmv{k}}}   \sigma_{\grassmv{k}}   \grasssym{,}  \grassnt{q}  \mid  \mathsf{N}_{{\mathrm{1}}}  \grasssym{,}  \mathellipsis  \grasssym{,}  \mathsf{N}_{\grassmv{k}}  \grasssym{,}  \mathsf{m}   \odot   \Delta_{{\mathrm{1}}}  \grasssym{,}  \mathellipsis  \grasssym{,}  \Delta_{\grassmv{k}}  \grasssym{,}  \grassmv{y}  \grasssym{:}  \grassnt{T'}  \vdash_{ \mathsf{m} }  \grasssym{[}  \grassnt{e_{\grassmv{i}}}  \grasssym{/}  \grassmv{x_{\grassmv{i}}}  \grasssym{]}  \grassnt{t}  \colon  \grassnt{T} .
    \]
    Applying the rule \rulename{grass}{term}{dropE} to the two inductive hypotheses
    yields the desired result.
  \end{case}
\end{proof}

\newpage

\subsection{Syntactic soundness of beta-conversions}
\label{asec:multi-mode-beta}
  \begin{definition}
  $ \beta $-conversions are defined by
  \begin{mathpar}
    \small
    \namedrules{grass}{beta}{unit,pair,arrow,sumL,sumR,raise,drop}
  \end{mathpar}
\end{definition}

\begin{theorem}
  \label{athm:beta-syntactic-soundness}
  $ \beta $-conversion preserves types and grades, i.e.\@
  whenever
  $  \rho  \mid  \mathsf{M}  \odot  \Gamma  \vdash_{ \mathsf{m} }  \grassnt{t_{{\mathrm{1}}}}  \equiv_\beta  \grassnt{t_{{\mathrm{2}}}}  \colon  \grassnt{A}  $
  then
  $  \rho  \mid  \mathsf{M}   \odot   \Gamma  \vdash_{ \mathsf{m} }  \grassnt{t_{{\mathrm{1}}}}  \colon  \grassnt{A}  $
  and
  $  \rho  \mid  \mathsf{M}   \odot   \Gamma  \vdash_{ \mathsf{m} }  \grassnt{t_{{\mathrm{2}}}}  \colon  \grassnt{A}  $
\end{theorem}

\begin{proof}
  In each rule, $ t_1 $ is obtained from the assumptions of that rule
  by applying the corresponding elimination rule.
  It remains to show that $  \rho  \mid  \mathsf{M}   \odot   \Gamma  \vdash_{ \mathsf{m} }  \grassnt{t_{{\mathrm{2}}}}  \colon  \grassnt{A}  $.
  We proceed by case analysis.

  \setcounter{case}{0}
  \newcommand*{\subcase}[1]
  {\par\vspace{1ex}\noindent{\textsc{Subcase} (\rulename{grass}{term}{#1}).}}

  \begin{case}[\normalfont\rulename{grass}{beta}{unit}\bf]\rm
    We proceed by induction on the derivation of
    $  \sigma  \mid  \mathsf{N}   \odot   \Delta  \vdash_{ \mathsf{m} }   \star_ \mathsf{m}   \colon   \mathbf{I}_ \mathsf{m}   $.
    We must consider the cases \rulename{grass}{term}{unitI} and
    \rulenames{grass}{term}{weak,sub,cont}.
    All other subcases do not apply, since their derivations do not produce a
    term of the form $  \star_ \mathsf{m}  $.
      \subcase{unitI}
      In this case the assumptions are
      $  \emptyset  \mid  \emptyset   \odot   \emptyset  \vdash_{ \mathsf{m} }   \star_ \mathsf{m}   \colon   \mathbf{I}_ \mathsf{m}   $ and
      $  \rho  \mid  \mathsf{M}   \odot   \Gamma  \vdash_{ \mathsf{m} }  \grassnt{e}  \colon  \grassnt{A}  $
      and we must prove $  \rho  \mid  \mathsf{M}   \odot   \Gamma  \vdash_{ \mathsf{m} }  \grassnt{e}  \colon  \grassnt{A}  $ which is an assumption.

      \subcase{weak}
      The derivation has the form
      \[
        \inferrule*{
           \mathsf{Weak} (  \mathsf{n}  ) 
          \and
          \mathsf{n}  \ge  \mathsf{m}
          \and
           \grassnt{B}  \in \mathsf{Type}( \mathsf{n} ) 
          \and
           \sigma  \mid  \mathsf{N}   \odot   \Delta  \vdash_{ \mathsf{m} }   \star_ \mathsf{m}   \colon   \mathbf{I}_ \mathsf{m}  
        }
        {
           \sigma  \grasssym{,}   0   \mid  \mathsf{N}  \grasssym{,}  \mathsf{n}   \odot   \Delta  \grasssym{,}  \grassmv{x}  \grasssym{:}  \grassnt{B}  \vdash_{ \mathsf{m} }   \star_ \mathsf{m}   \colon   \mathbf{I}_ \mathsf{m}  
        }.
      \]
      The inductive hypothesis is
      \[
         \rho  \grasssym{,}   \grassnt{q}   \sigma   \mid  \mathsf{M}  \grasssym{,}  \mathsf{N}   \odot   \Gamma  \grasssym{,}  \Delta  \vdash_{ \mathsf{m} }  \grassnt{e}  \colon  \grassnt{A} 
      \]
      The assumptions $  \mathsf{Weak} (  \mathsf{n}  )  $, $ \mathsf{n}  \ge  \mathsf{m} $ and $  \grassnt{B}  \in \mathsf{Type}( \mathsf{n} )  $
      guarantee that we can apply weakening to the IH to obtain the desired result,
      noticing that $  \grassnt{q}   \grasssym{(}  \sigma  \grasssym{,}   0   \grasssym{)}  = \grasssym{(}   \grassnt{q}   \sigma   \grasssym{)}  \grasssym{,}   0  $.

      \subcase{sub}
      The derivation is
      \[
        \inferrule*
        {
           \sigma  \le  \sigma' 
          \and
           \sigma  \mid  \mathsf{N}   \odot   \Delta  \vdash_{ \mathsf{m} }   \star_ \mathsf{m}   \colon   \mathbf{I}_ \mathsf{m}  
        }
        {
           \sigma'  \mid  \mathsf{N}   \odot   \Delta  \vdash_{ \mathsf{m} }   \star_ \mathsf{m}   \colon   \mathbf{I}_ \mathsf{m}  
        }
      \]
      with IH $  \rho  \grasssym{,}   \grassnt{q}   \sigma   \mid  \mathsf{M}  \grasssym{,}  \mathsf{N}   \odot   \Gamma  \grasssym{,}  \Delta  \vdash_{ \mathsf{m} }   \star_ \mathsf{m}   \colon   \mathbf{I}_ \mathsf{m}   $.
      From $  \sigma  \le  \sigma'  $ we obtain $  \rho  \grasssym{,}   \grassnt{q}   \sigma   \le  \rho  \grasssym{,}   \grassnt{q}   \sigma'   $,
      so applying the rule \rulename{grass}{term}{sub} to the IH yields the
      desired result.

      \subcase{cont}
      In this case the derivation has the form
      \[
        \inferrule*
        {
           \grassnt{r_{{\mathrm{1}}}}  \grasssym{,}  \grassnt{r_{{\mathrm{2}}}}  \in \mathsf{Cont}( \mathsf{n} ) 
          \and
           \sigma  \grasssym{,}  \grassnt{r_{{\mathrm{1}}}}  \grasssym{,}  \grassnt{r_{{\mathrm{2}}}}  \mid  \mathsf{N}  \grasssym{,}  \mathsf{n}  \grasssym{,}  \mathsf{n}   \odot   \Delta  \grasssym{,}  \grassmv{y_{{\mathrm{1}}}}  \grasssym{:}  \grassnt{B}  \grasssym{,}  \grassmv{y_{{\mathrm{2}}}}  \grasssym{:}  \grassnt{B}  \vdash_{ \mathsf{m} }   \star_ \mathsf{m}   \colon   \mathbf{I}_ \mathsf{m}  
        }
        {
           \sigma  \grasssym{,}  \grassnt{r_{{\mathrm{1}}}}  \grasssym{+}  \grassnt{r_{{\mathrm{2}}}}  \mid  \mathsf{N}  \grasssym{,}  \mathsf{n}   \odot   \Delta  \grasssym{,}  \grassmv{z}  \grasssym{:}  \grassnt{B}  \vdash_{ \mathsf{m} }   \star_ \mathsf{m}   \colon   \mathbf{I}_ \mathsf{m}  
        }
      \]
      The IH is
      \[
         \rho  \grasssym{,}   \grassnt{q}   \sigma   \grasssym{,}   \grassnt{q}   \grassnt{r_{{\mathrm{1}}}}   \grasssym{,}   \grassnt{q}   \grassnt{r_{{\mathrm{2}}}}   \mid  \mathsf{M}  \grasssym{,}  \mathsf{N}  \grasssym{,}  \mathsf{n}  \grasssym{,}  \mathsf{n}   \odot   \Gamma  \grasssym{,}  \Delta  \grasssym{,}  \grassmv{y_{{\mathrm{1}}}}  \grasssym{:}  \grassnt{B}  \grasssym{,}  \grassmv{y_{{\mathrm{2}}}}  \grasssym{:}  \grassnt{B}  \vdash_{ \mathsf{m} }  \grassnt{e}  \colon  \grassnt{A} 
      \]
      The variables $ y_1, y_2 $ do not occur freely in $ e $,
      because we have the derivation $  \rho  \mid  \mathsf{M}   \odot   \Gamma  \vdash_{ \mathsf{m} }  \grassnt{e}  \colon  \grassnt{A}  $.
      Therefore, it suffices to show that
      \[
         \rho  \grasssym{,}   \grassnt{q}   \sigma   \grasssym{,}   \grassnt{q}   \grasssym{(}  \grassnt{r_{{\mathrm{1}}}}  \grasssym{+}  \grassnt{r_{{\mathrm{2}}}}  \grasssym{)}   \mid  \mathsf{M}  \grasssym{,}  \mathsf{N}  \grasssym{,}  \mathsf{n}   \odot   \Gamma  \grasssym{,}  \Delta  \grasssym{,}  \grassmv{z}  \grasssym{:}  \grassnt{B}  \vdash_{ \mathsf{m} }  \grassnt{e}  \colon  \grassnt{A} 
      \]
      Since $  \mathsf{Cont}( \mathsf{n} )  $ is an ideal, we have $   \grassnt{q}   \grassnt{r_{{\mathrm{1}}}}   \grasssym{,}   \grassnt{q}   \grassnt{r_{{\mathrm{2}}}}   \in \mathsf{Cont}( \mathsf{n} )  $,
      and hence may apply the rule \rulename{grass}{term}{cont}
      to obtain the desired result.

  \end{case}
  \begin{case}[\normalfont\rulename{grass}{beta}{pair}\bf]\rm
    Given
    \begin{mathpar}
       \rho  \grasssym{,}  \grassnt{q}  \grasssym{,}  \grassnt{q}  \mid  \mathsf{M}  \grasssym{,}  \mathsf{m}  \grasssym{,}  \mathsf{m}   \odot   \Gamma  \grasssym{,}  \grassmv{x_{{\mathrm{1}}}}  \grasssym{:}  \grassnt{A_{{\mathrm{1}}}}  \grasssym{,}  \grassmv{x_{{\mathrm{2}}}}  \grasssym{:}  \grassnt{A_{{\mathrm{2}}}}  \vdash_{ \mathsf{n} }  \grassnt{e}  \colon  \grassnt{A} 
      \and
       \sigma  \mid  \mathsf{N}   \odot   \Delta  \vdash_{ \mathsf{m} }  \grasssym{(}  \grassnt{t_{{\mathrm{1}}}}  \grasssym{,}  \grassnt{t_{{\mathrm{2}}}}  \grasssym{)}  \colon   \grassnt{A_{{\mathrm{1}}}}  \otimes  \grassnt{A_{{\mathrm{2}}}}   ,
    \end{mathpar}
    we must show that
    \[
       \rho  \grasssym{,}   \grassnt{q}   \sigma   \mid  \mathsf{M}  \grasssym{,}  \mathsf{N}   \odot   \Gamma  \grasssym{,}  \Delta  \vdash_{ \mathsf{n} }  \grasssym{[}  \grassnt{t_{{\mathrm{1}}}}  \grasssym{/}  \grassmv{x_{{\mathrm{1}}}}  \grasssym{,}  \grassnt{t_{{\mathrm{2}}}}  \grasssym{/}  \grassmv{x_{{\mathrm{2}}}}  \grasssym{]}  \grassnt{e}  \colon  \grassnt{A}  .
    \]
    Proceed by induction on the derivation of
    $  \sigma  \mid  \mathsf{N}   \odot   \Delta  \vdash_{ \mathsf{m} }  \grasssym{(}  \grassnt{t_{{\mathrm{1}}}}  \grasssym{,}  \grassnt{t_{{\mathrm{2}}}}  \grasssym{)}  \colon   \grassnt{A_{{\mathrm{1}}}}  \otimes  \grassnt{A_{{\mathrm{2}}}}   $.
    We must consider the cases \rulename{grass}{term}{pairI} and
    \rulenames{grass}{term}{weak,sub,cont}.
    All other subcases do not apply since they don't produce a term of the form
    $ \grasssym{(}  \grassnt{t_{{\mathrm{1}}}}  \grasssym{,}  \grassnt{t_{{\mathrm{2}}}}  \grasssym{)} $.
      \subcase{pairI}
      In this case the derivation is
      \[
        \inferrule*{
           \sigma_{{\mathrm{1}}}  \mid  \mathsf{N}_{{\mathrm{1}}}   \odot   \Delta_{{\mathrm{1}}}  \vdash_{ \mathsf{m} }  \grassnt{t_{{\mathrm{1}}}}  \colon  \grassnt{A_{{\mathrm{1}}}} 
          \and
           \sigma_{{\mathrm{2}}}  \mid  \mathsf{N}_{{\mathrm{2}}}   \odot   \Delta_{{\mathrm{2}}}  \vdash_{ \mathsf{m} }  \grassnt{t_{{\mathrm{2}}}}  \colon  \grassnt{A_{{\mathrm{2}}}} 
        }
        {
           \sigma_{{\mathrm{1}}}  \grasssym{,}  \sigma_{{\mathrm{2}}}  \mid  \mathsf{N}_{{\mathrm{1}}}  \grasssym{,}  \mathsf{N}_{{\mathrm{2}}}   \odot   \Delta_{{\mathrm{1}}}  \grasssym{,}  \Delta_{{\mathrm{2}}}  \vdash_{ \mathsf{m} }  \grasssym{(}  \grassnt{t_{{\mathrm{1}}}}  \grasssym{,}  \grassnt{t_{{\mathrm{2}}}}  \grasssym{)}  \colon   \grassnt{A_{{\mathrm{1}}}}  \otimes  \grassnt{A_{{\mathrm{2}}}}  
        }
      \]
      Applying the substitution theorem with $ t_1 $ and $ t_2 $,
      we obtain
      \[
         \rho  \grasssym{,}   \grassnt{q}   \sigma_{{\mathrm{1}}}   \grasssym{,}   \grassnt{q}   \sigma_{{\mathrm{2}}}   \mid  \mathsf{M}  \grasssym{,}  \mathsf{N}_{{\mathrm{1}}}  \grasssym{,}  \mathsf{N}_{{\mathrm{2}}}   \odot   \Gamma  \grasssym{,}  \Delta_{{\mathrm{1}}}  \grasssym{,}  \Delta_{{\mathrm{2}}}  \vdash_{ \mathsf{n} }  \grasssym{[}  \grassnt{t_{{\mathrm{1}}}}  \grasssym{/}  \grassmv{x_{{\mathrm{1}}}}  \grasssym{,}  \grassnt{t_{{\mathrm{2}}}}  \grasssym{/}  \grassmv{x_{{\mathrm{2}}}}  \grasssym{]}  \grassnt{e}  \colon  \grassnt{A} 
      \]
      as desired.

      \subcase{weak}
      In this case the derivation is
      \[
        \inferrule*
        {
          \mathsf{n}'  \ge  \mathsf{m}
          \and
           \mathsf{Weak} (  \mathsf{n}'  ) 
          \and
           \grassnt{B}  \in \mathsf{Type}( \mathsf{n}' ) 
          \and
           \sigma  \mid  \mathsf{N}   \odot   \Delta  \vdash_{ \mathsf{m} }  \grasssym{(}  \grassnt{t_{{\mathrm{1}}}}  \grasssym{,}  \grassnt{t_{{\mathrm{2}}}}  \grasssym{)}  \colon   \grassnt{A_{{\mathrm{1}}}}  \otimes  \grassnt{A_{{\mathrm{2}}}}  
        }
        {
           \sigma  \grasssym{,}   0   \mid  \mathsf{N}  \grasssym{,}  \mathsf{n}'   \odot   \Delta  \grasssym{,}  \grassmv{z}  \grasssym{:}  \grassnt{B}  \vdash_{ \mathsf{m} }  \grasssym{(}  \grassnt{t_{{\mathrm{1}}}}  \grasssym{,}  \grassnt{t_{{\mathrm{2}}}}  \grasssym{)}  \colon   \grassnt{A_{{\mathrm{1}}}}  \otimes  \grassnt{A_{{\mathrm{2}}}}  
        }
      \]
      By induction we have
      \[
         \rho  \grasssym{,}   \grassnt{q}   \sigma   \mid  \mathsf{M}  \grasssym{,}  \mathsf{N}   \odot   \Gamma  \grasssym{,}  \Delta  \vdash_{ \mathsf{n} }  \grasssym{[}  \grassnt{t_{{\mathrm{1}}}}  \grasssym{/}  \grassmv{x_{{\mathrm{1}}}}  \grasssym{,}  \grassnt{t_{{\mathrm{2}}}}  \grasssym{/}  \grassmv{x_{{\mathrm{2}}}}  \grasssym{]}  \grassnt{e}  \colon  \grassnt{A}  .
      \]
      We have $ \mathsf{n}  \le  \mathsf{m}  \le  \mathsf{n}' $ and hence may apply the rule
      \rulename{grass}{term}{weak} to the IH to get
      \[
         \rho  \grasssym{,}   \grassnt{q}   \sigma   \grasssym{,}   0   \mid  \mathsf{M}  \grasssym{,}  \mathsf{N}  \grasssym{,}  \mathsf{n}'   \odot   \Gamma  \grasssym{,}  \Delta  \grasssym{,}  \grassmv{z}  \grasssym{:}  \grassnt{B}  \vdash_{ \mathsf{n} }  \grasssym{[}  \grassnt{t_{{\mathrm{1}}}}  \grasssym{/}  \grassmv{x_{{\mathrm{1}}}}  \grasssym{,}  \grassnt{t_{{\mathrm{2}}}}  \grasssym{/}  \grassmv{x_{{\mathrm{2}}}}  \grasssym{]}  \grassnt{e}  \colon  \grassnt{A} 
      \]
      as desired.

      \subcase{sub}
      The derivation is
      \[
        \inferrule*
        {
           \sigma'  \le  \sigma 
          \and
           \sigma'  \mid  \mathsf{N}   \odot   \Delta  \vdash_{ \mathsf{m} }  \grasssym{(}  \grassnt{t_{{\mathrm{1}}}}  \grasssym{,}  \grassnt{t_{{\mathrm{2}}}}  \grasssym{)}  \colon   \grassnt{A_{{\mathrm{1}}}}  \otimes  \grassnt{A_{{\mathrm{2}}}}  
        }
        {
           \sigma  \mid  \mathsf{N}   \odot   \Delta  \vdash_{ \mathsf{m} }  \grasssym{(}  \grassnt{t_{{\mathrm{1}}}}  \grasssym{,}  \grassnt{t_{{\mathrm{2}}}}  \grasssym{)}  \colon   \grassnt{A_{{\mathrm{1}}}}  \otimes  \grassnt{A_{{\mathrm{2}}}}  
        }
      \]
      and the inductive hypothesis is
      \[
         \rho  \grasssym{,}   \grassnt{q}   \sigma'   \mid  \mathsf{M}  \grasssym{,}  \mathsf{N}   \odot   \Gamma  \grasssym{,}  \Delta  \vdash_{ \mathsf{n} }  \grasssym{[}  \grassnt{t_{{\mathrm{1}}}}  \grasssym{/}  \grassmv{x_{{\mathrm{1}}}}  \grasssym{,}  \grassnt{t_{{\mathrm{2}}}}  \grasssym{/}  \grassmv{x_{{\mathrm{2}}}}  \grasssym{]}  \grassnt{e}  \colon  \grassnt{A}  .
      \]
      The desired result follows by applying the rule \rulename{grass}{term}{sub}
      using the fact that multiplcation is monotone.

      \subcase{cont}
      The derivation is
      \[
        \inferrule*
        {
           \grassnt{r_{{\mathrm{1}}}}  \grasssym{,}  \grassnt{r_{{\mathrm{2}}}}  \in \mathsf{Cont}( \mathsf{n}' ) 
          \and
           \sigma  \grasssym{,}  \grassnt{r_{{\mathrm{1}}}}  \grasssym{,}  \grassnt{r_{{\mathrm{2}}}}  \mid  \mathsf{N}  \grasssym{,}  \mathsf{n}'  \grasssym{,}  \mathsf{n}'   \odot   \Delta  \grasssym{,}  \grassmv{y_{{\mathrm{1}}}}  \grasssym{:}  \grassnt{B}  \grasssym{,}  \grassmv{y_{{\mathrm{2}}}}  \grasssym{:}  \grassnt{B}  \vdash_{ \mathsf{m} }  \grasssym{(}  \grassnt{t_{{\mathrm{1}}}}  \grasssym{,}  \grassnt{t_{{\mathrm{2}}}}  \grasssym{)}  \colon   \grassnt{A_{{\mathrm{1}}}}  \otimes  \grassnt{A_{{\mathrm{2}}}}  
        }
        {
           \sigma  \grasssym{,}  \grassnt{r_{{\mathrm{1}}}}  \grasssym{+}  \grassnt{r_{{\mathrm{2}}}}  \mid  \mathsf{N}  \grasssym{,}  \mathsf{n}'  \grasssym{,}  \mathsf{n}'   \odot   \Delta  \grasssym{,}  \grassmv{z}  \grasssym{:}  \grassnt{B}  \vdash_{ \mathsf{m} }  \grasssym{(}  \grasssym{[}  \grassmv{z}  \grasssym{/}  \grassmv{y_{{\mathrm{1}}}}  \grasssym{,}  \grassmv{z}  \grasssym{/}  \grassmv{y_{{\mathrm{2}}}}  \grasssym{]}  \grassnt{t_{{\mathrm{1}}}}  \grasssym{,}  \grasssym{[}  \grassmv{z}  \grasssym{/}  \grassmv{y_{{\mathrm{1}}}}  \grasssym{,}  \grassmv{z}  \grasssym{/}  \grassmv{y_{{\mathrm{2}}}}  \grasssym{]}  \grassnt{t_{{\mathrm{2}}}}  \grasssym{)}  \colon   \grassnt{A_{{\mathrm{1}}}}  \otimes  \grassnt{A_{{\mathrm{2}}}}  
        }
      \]
      The IH is
      \[
         \rho  \grasssym{,}   \grassnt{q}   \sigma   \grasssym{,}   \grassnt{q}   \grassnt{r_{{\mathrm{1}}}}   \grasssym{,}   \grassnt{q}   \grassnt{r_{{\mathrm{2}}}}   \mid  \mathsf{M}  \grasssym{,}  \mathsf{N}  \grasssym{,}  \mathsf{n}'  \grasssym{,}  \mathsf{n}'   \odot   \Gamma  \grasssym{,}  \Delta  \grasssym{,}  \grassmv{y_{{\mathrm{1}}}}  \grasssym{:}  \grassnt{B}  \grasssym{,}  \grassmv{y_{{\mathrm{2}}}}  \grasssym{:}  \grassnt{B}  \vdash_{ \mathsf{n} }  \grasssym{[}  \grassnt{t_{{\mathrm{1}}}}  \grasssym{/}  \grassmv{x_{{\mathrm{1}}}}  \grasssym{,}  \grassnt{t_{{\mathrm{2}}}}  \grasssym{/}  \grassmv{x_{{\mathrm{2}}}}  \grasssym{]}  \grassnt{e}  \colon  \grassnt{A} 
      \]
      Since $  \mathsf{Cont}( \mathsf{n}' )  $ is an ideal, we have $   \grassnt{q}   \grassnt{r_{{\mathrm{1}}}}   \grasssym{,}   \grassnt{q}   \grassnt{r_{{\mathrm{2}}}}   \in \mathsf{Cont}( \mathsf{n}' )  $.
      We thus apply the rule \rulename{grass}{term}{cont}
      to obtain
      \[
         \rho  \grasssym{,}   \grassnt{q}   \sigma   \grasssym{,}    \grassnt{q}   \grassnt{r_{{\mathrm{1}}}}   \grasssym{+}  \grassnt{q}   \grassnt{r_{{\mathrm{2}}}}   \mid  \mathsf{M}  \grasssym{,}  \mathsf{N}  \grasssym{,}  \mathsf{n}'   \odot   \Gamma  \grasssym{,}  \Delta  \grasssym{,}  \grassmv{z}  \grasssym{:}  \grassnt{B}  \vdash_{ \mathsf{n} }  \grasssym{[}  \grassmv{z}  \grasssym{/}  \grassmv{y_{{\mathrm{1}}}}  \grasssym{,}  \grassmv{z}  \grasssym{/}  \grassmv{y_{{\mathrm{2}}}}  \grasssym{]}  \grasssym{[}  \grassnt{t_{{\mathrm{1}}}}  \grasssym{/}  \grassmv{x_{{\mathrm{1}}}}  \grasssym{,}  \grassnt{t_{{\mathrm{2}}}}  \grasssym{/}  \grassmv{x_{{\mathrm{2}}}}  \grasssym{]}  \grassnt{e}  \colon  \grassnt{A} 
      \]
      Since the variables $ y_1, y_2 $ only occur freely in $ t_1 $ and $ t_2 $,
      the resulting term is equal to
      \[
        \grasssym{[}  \grasssym{[}  \grassmv{z}  \grasssym{/}  \grassmv{y_{{\mathrm{1}}}}  \grasssym{,}  \grassmv{z}  \grasssym{/}  \grassmv{y_{{\mathrm{2}}}}  \grasssym{]}  \grassnt{t_{{\mathrm{1}}}}  \grasssym{/}  \grassmv{x_{{\mathrm{1}}}}  \grasssym{,}  \grasssym{[}  \grassmv{z}  \grasssym{/}  \grassmv{y_{{\mathrm{1}}}}  \grasssym{,}  \grassmv{z}  \grasssym{/}  \grassmv{y_{{\mathrm{2}}}}  \grasssym{]}  \grassnt{t_{{\mathrm{2}}}}  \grasssym{/}  \grassmv{x_{{\mathrm{2}}}}  \grasssym{]}  \grassnt{e}
      \]
      as desired.

  \end{case}
  \begin{case}[\normalfont\rulename{grass}{beta}{arrow}\bf]\rm
    Given
    \begin{mathpar}
       \rho  \mid  \mathsf{M}   \odot   \Gamma  \vdash_{ \mathsf{n} }  \lambda  \grassmv{x}  \grasssym{.}  \grassnt{e}  \colon   \grassnt{B} ^{ \grassnt{q} : \mathsf{m} } \multimap  \grassnt{A}  
      \and
       \sigma  \mid  \mathsf{N}   \odot   \Delta  \vdash_{ \mathsf{m} }  \grassnt{t}  \colon  \grassnt{B} 
    \end{mathpar}
    we must prove that
    \[
       \rho  \grasssym{,}   \grassnt{q}   \sigma   \mid  \mathsf{M}  \grasssym{,}  \mathsf{N}   \odot   \Gamma  \grasssym{,}  \Delta  \vdash_{ \mathsf{n} }  \grasssym{[}  \grassnt{t}  \grasssym{/}  \grassmv{x}  \grasssym{]}  \grassnt{e}  \colon  \grassnt{A} 
    \]
    We proceed by induction on the derivation of
    \[
       \rho  \mid  \mathsf{M}   \odot   \Gamma  \vdash_{ \mathsf{n} }  \lambda  \grassmv{x}  \grasssym{.}  \grassnt{e}  \colon   \grassnt{B} ^{ \grassnt{q} : \mathsf{m} } \multimap  \grassnt{A}  .
    \]
    We must analyse the cases \rulename{grass}{term}{arrowI} and
    \rulenames{grass}{term}{weak,sub,cont}.
    The other cases do not apply, since they do not produce a term of the form
    $ \lambda  \grassmv{x}  \grasssym{.}  \grassnt{e} $.
      \subcase{arrowI}
      In this case we have
      \[
        \inferrule*
        {
           \rho  \grasssym{,}  \grassnt{q}  \mid  \mathsf{M}  \grasssym{,}  \mathsf{n}   \odot   \Gamma  \grasssym{,}  \grassmv{x}  \grasssym{:}  \grassnt{B}  \vdash_{ \mathsf{n} }  \grassnt{e}  \colon  \grassnt{A} 
        }
        {
           \rho  \mid  \mathsf{M}   \odot   \Gamma  \vdash_{ \mathsf{n} }  \lambda  \grassmv{x}  \grasssym{.}  \grassnt{e}  \colon   \grassnt{B} ^{ \grassnt{q} : \mathsf{m} } \multimap  \grassnt{A}  
        }
      \]
      The claim follows by applying substitution to the assumption
      $  \rho  \grasssym{,}  \grassnt{q}  \mid  \mathsf{M}  \grasssym{,}  \mathsf{n}   \odot   \Gamma  \grasssym{,}  \grassmv{x}  \grasssym{:}  \grassnt{B}  \vdash_{ \mathsf{n} }  \grassnt{e}  \colon  \grassnt{A}  $

      \subcase{weak}
      In this case the derivation is
      \[
        \inferrule*
        {
           \mathsf{Weak} (  \mathsf{m}'  ) 
          \and
          \mathsf{m}'  \ge  \mathsf{n}
          \and
           \grassnt{T}  \in \mathsf{Type}( \mathsf{m}' ) 
          \and
           \rho  \mid  \mathsf{M}   \odot   \Gamma  \vdash_{ \mathsf{n} }  \lambda  \grassmv{x}  \grasssym{.}  \grassnt{e}  \colon   \grassnt{B} ^{ \grassnt{q} : \mathsf{m} } \multimap  \grassnt{A}  
        }
        {
           \rho  \grasssym{,}   0   \mid  \mathsf{M}  \grasssym{,}  \mathsf{m}'   \odot   \Gamma  \grasssym{,}  \grassmv{y}  \grasssym{:}  \grassnt{T}  \vdash_{ \mathsf{n} }  \lambda  \grassmv{x}  \grasssym{.}  \grassnt{e}  \colon   \grassnt{B} ^{ \grassnt{q} : \mathsf{m} } \multimap  \grassnt{A}  
        }
      \]
      By induction we have
      \[
         \rho  \grasssym{,}   \grassnt{q}   \sigma   \mid  \mathsf{M}  \grasssym{,}  \mathsf{N}   \odot   \Gamma  \grasssym{,}  \Delta  \vdash_{ \mathsf{n} }  \grasssym{[}  \grassnt{t}  \grasssym{/}  \grassmv{x}  \grasssym{]}  \grassnt{e}  \colon  \grassnt{A} 
      \]
      and we may apply the rule \rulename{grass}{term}{weak}
      to obtain
      \[
         \rho  \grasssym{,}   \grassnt{q}   \sigma   \grasssym{,}   0   \mid  \mathsf{M}  \grasssym{,}  \mathsf{N}  \grasssym{,}  \mathsf{m}'   \odot   \Gamma  \grasssym{,}  \Delta  \grasssym{,}  \grassmv{y}  \grasssym{:}  \grassnt{T}  \vdash_{ \mathsf{n} }  \grasssym{[}  \grassnt{t}  \grasssym{/}  \grassmv{x}  \grasssym{]}  \grassnt{e}  \colon  \grassnt{A} 
      \]
      Since we treat contexts as unoredered, this is equivalent to the desired
      judgment
      \[
         \rho  \grasssym{,}   0   \grasssym{,}   \grassnt{q}   \sigma   \mid  \mathsf{M}  \grasssym{,}  \mathsf{m}'  \grasssym{,}  \mathsf{N}   \odot   \Gamma  \grasssym{,}  \grassmv{y}  \grasssym{:}  \grassnt{T}  \grasssym{,}  \Delta  \vdash_{ \mathsf{n} }  \grasssym{[}  \grassnt{t}  \grasssym{/}  \grassmv{x}  \grasssym{]}  \grassnt{e}  \colon  \grassnt{A} 
      \]

      \subcase{cont}
      In this case the derivation is
      \[
        \inferrule*
        {
           \grassnt{r_{{\mathrm{1}}}}  \grasssym{,}  \grassnt{r_{{\mathrm{2}}}}  \in \mathsf{Cont}( \mathsf{m}' ) 
          \and
           \rho  \grasssym{,}  \grassnt{r_{{\mathrm{1}}}}  \grasssym{,}  \grassnt{r_{{\mathrm{2}}}}  \mid  \mathsf{M}  \grasssym{,}  \mathsf{m}'  \grasssym{,}  \mathsf{m}'   \odot   \Gamma  \grasssym{,}  \grassmv{y_{{\mathrm{1}}}}  \grasssym{:}  \grassnt{T}  \grasssym{,}  \grassmv{y_{{\mathrm{2}}}}  \grasssym{:}  \grassnt{T}  \vdash_{ \mathsf{n} }  \lambda  \grassmv{x}  \grasssym{.}  \grassnt{e}  \colon   \grassnt{B} ^{ \grassnt{q} : \mathsf{m} } \multimap  \grassnt{A}  
        }
        {
           \rho  \grasssym{,}  \grassnt{r_{{\mathrm{1}}}}  \grasssym{+}  \grassnt{r_{{\mathrm{2}}}}  \mid  \mathsf{M}  \grasssym{,}  \mathsf{m}'  \grasssym{,}  \mathsf{m}'   \odot   \Gamma  \grasssym{,}  \grassmv{z}  \grasssym{:}  \grassnt{T}  \vdash_{ \mathsf{n} }  \grasssym{(}  \lambda  \grassmv{x}  \grasssym{.}  \grasssym{[}  \grassmv{z}  \grasssym{/}  \grassmv{y_{{\mathrm{1}}}}  \grasssym{,}  \grassmv{z}  \grasssym{/}  \grassmv{y_{{\mathrm{2}}}}  \grasssym{]}  \grassnt{e}  \grasssym{)}  \colon   \grassnt{B} ^{ \grassnt{q} : \mathsf{m} } \multimap  \grassnt{A}  
        }
      \]
      By induction we have
      \[
         \rho  \grasssym{,}  \grassnt{r_{{\mathrm{1}}}}  \grasssym{,}  \grassnt{r_{{\mathrm{2}}}}  \grasssym{,}   \grassnt{q}   \sigma   \mid  \mathsf{M}  \grasssym{,}  \mathsf{m}'  \grasssym{,}  \mathsf{m}'  \grasssym{,}  \mathsf{N}   \odot   \Gamma  \grasssym{,}  \grassmv{y_{{\mathrm{1}}}}  \grasssym{:}  \grassnt{T}  \grasssym{,}  \grassmv{y_{{\mathrm{2}}}}  \grasssym{:}  \grassnt{T}  \grasssym{,}  \Delta  \vdash_{ \mathsf{n} }  \grasssym{[}  \grassnt{t}  \grasssym{/}  \grassmv{x}  \grasssym{]}  \grassnt{e}  \colon  \grassnt{A} 
      \]
      Since contexts are unordered, we may apply the rule
      \rulename{grass}{term}{cont} to obtain
      \[
         \rho  \grasssym{,}  \grassnt{r_{{\mathrm{1}}}}  \grasssym{+}  \grassnt{r_{{\mathrm{2}}}}  \grasssym{,}   \grassnt{q}   \sigma   \mid  \mathsf{M}  \grasssym{,}  \mathsf{m}'  \grasssym{,}  \mathsf{N}   \odot   \Gamma  \grasssym{,}  \grassmv{z}  \grasssym{:}  \grassnt{T}  \grasssym{,}  \Delta  \vdash_{ \mathsf{n} }  \grasssym{[}  \grassmv{z}  \grasssym{/}  \grassmv{y_{{\mathrm{1}}}}  \grasssym{,}  \grassmv{z}  \grasssym{/}  \grassmv{y_{{\mathrm{2}}}}  \grasssym{]}  \grasssym{[}  \grassnt{t}  \grasssym{/}  \grassmv{x}  \grasssym{]}  \grassnt{e}  \colon  \grassnt{A} 
      \]
      Since the variables $ y_1, y_2 $ are not free in $ t $ this term is equal to
      $ \grasssym{[}  \grassnt{t}  \grasssym{/}  \grassmv{x}  \grasssym{]}  \grasssym{[}  \grassmv{z}  \grasssym{/}  \grassmv{y_{{\mathrm{1}}}}  \grasssym{,}  \grassmv{z}  \grasssym{/}  \grassmv{y_{{\mathrm{2}}}}  \grasssym{]}  \grassnt{e} $ which is the desired term.

      \subcase{sub}
      In this case we have the derivation
      \[
        \inferrule*
        {
           \rho'  \le  \rho 
          \and
           \rho'  \mid  \mathsf{M}   \odot   \Gamma  \vdash_{ \mathsf{n} }  \lambda  \grassmv{x}  \grasssym{.}  \grassnt{e}  \colon   \grassnt{B} ^{ \grassnt{q} : \mathsf{m} } \multimap  \grassnt{A}  
        }
        {
           \rho  \mid  \mathsf{M}   \odot   \Gamma  \vdash_{ \mathsf{n} }  \lambda  \grassmv{x}  \grasssym{.}  \grassnt{e}  \colon   \grassnt{B} ^{ \grassnt{q} : \mathsf{m} } \multimap  \grassnt{A}  
        }
      \]
      The inductive hypothesis is
      \[
         \rho'  \grasssym{,}   \grassnt{q}   \sigma   \mid  \mathsf{M}  \grasssym{,}  \mathsf{N}   \odot   \Gamma  \grasssym{,}  \Delta  \vdash_{ \mathsf{n} }  \grasssym{[}  \grassnt{t}  \grasssym{/}  \grassmv{x}  \grasssym{]}  \grassnt{e}  \colon  \grassnt{A} 
      \]
      and since $  \rho'  \grasssym{,}   \grassnt{q}   \sigma   \le  \rho  \grasssym{,}   \grassnt{q}   \sigma   $ we may apply
      the rule \rulename{grass}{term}{sub} to obtain the desired result.

  \end{case}
  \begin{case}[\normalfont\rulename{grass}{beta}{raise}\bf]\rm
    For modes $ \mathsf{m}  \le  \mathsf{n} $,
    given $  \rho  \mid  \mathsf{M}   \odot   \Gamma  \vdash_{ \mathsf{n} }   \operatorname{\uparrow} _{ \mathsf{m}  \le  \mathsf{n} }  \grassnt{t}   \colon   \operatorname{\uparrow} _{ \mathsf{m}  \le  \mathsf{n} }  \grassnt{A}   $,
    we must prove that $  \rho  \mid  \mathsf{M}   \odot   \Gamma  \vdash_{ \mathsf{m} }  \grassnt{t}  \colon  \grassnt{A}  $.
    We proceed by induction on the derivation of
    $  \rho  \mid  \mathsf{M}   \odot   \Gamma  \vdash_{ \mathsf{n} }   \operatorname{\uparrow} _{ \mathsf{m}  \le  \mathsf{n} }  \grassnt{t}   \colon   \operatorname{\uparrow} _{ \mathsf{m}  \le  \mathsf{n} }  \grassnt{A}   $.
    We must analyse the cases \rulename{grass}{term}{raiseI} and
    \rulenames{grass}{term}{weak,sub,cont}.
      \subcase{raiseI}
      In this case the derivation is
      \[
        \inferrule*
        {
          \mathsf{m}  \le  \mathsf{n}
          \and
           \rho  \mid  \mathsf{M}   \odot   \Gamma  \vdash_{ \mathsf{m} }  \grassnt{t}  \colon  \grassnt{A} 
        }
        {
           \rho  \mid  \mathsf{M}   \odot   \Gamma  \vdash_{ \mathsf{n} }   \operatorname{\uparrow} _{ \mathsf{m}  \le  \mathsf{n} }  \grassnt{t}   \colon   \operatorname{\uparrow} _{ \mathsf{m}  \le  \mathsf{n} }  \grassnt{A}  
        }.
      \]
      We are done, since our goal is an assumption.

      \subcase{weak}
      In this case the derivation is
      \[
        \inferrule*
        {
           \mathsf{Weak} (  \mathsf{m}'  ) 
          \and
          \mathsf{n}  \le  \mathsf{m}'
          \and
           \grassnt{B}  \in \mathsf{Type}( \mathsf{m}' ) 
          \and
           \rho  \mid  \mathsf{M}   \odot   \Gamma  \vdash_{ \mathsf{n} }   \operatorname{\uparrow} _{ \mathsf{m}  \le  \mathsf{n} }  \grassnt{t}   \colon   \operatorname{\uparrow} _{ \mathsf{m}  \le  \mathsf{n} }  \grassnt{A}  
        }
        {
           \rho  \grasssym{,}   0   \mid  \mathsf{M}  \grasssym{,}  \mathsf{m}'   \odot   \Gamma  \grasssym{,}  \grassmv{x}  \grasssym{:}  \grassnt{B}  \vdash_{ \mathsf{n} }   \operatorname{\uparrow} _{ \mathsf{m}  \le  \mathsf{n} }  \grassnt{t}   \colon   \operatorname{\uparrow} _{ \mathsf{m}  \le  \mathsf{n} }  \grassnt{A}  
        }
      \]
      The IH is $  \rho  \mid  \mathsf{M}   \odot   \Gamma  \vdash_{ \mathsf{m} }  \grassnt{t}  \colon  \grassnt{A}  $.
      Because of $ \mathsf{m}  \le  \mathsf{n}  \le  \mathsf{m}' $ we may apply \rulename{grass}{term}{weak}
      to the IH to obtain
      $  \rho  \grasssym{,}   0   \mid  \mathsf{M}  \grasssym{,}  \mathsf{m}'   \odot   \Gamma  \grasssym{,}  \grassmv{x}  \grasssym{:}  \grassnt{B}  \vdash_{ \mathsf{m} }  \grassnt{t}  \colon  \grassnt{A}  $ as desired.

      \subcase{cont}
      In this case the derivation is
      \[
        \inferrule*
        {
           \grassnt{r_{{\mathrm{1}}}}  \grasssym{,}  \grassnt{r_{{\mathrm{2}}}}  \in \mathsf{Cont}( \mathsf{m}' ) 
          \and
           \rho  \grasssym{,}  \grassnt{r_{{\mathrm{1}}}}  \grasssym{,}  \grassnt{r_{{\mathrm{2}}}}  \mid  \mathsf{M}  \grasssym{,}  \mathsf{m}'  \grasssym{,}  \mathsf{m}'   \odot   \Gamma  \grasssym{,}  \grassmv{x_{{\mathrm{1}}}}  \grasssym{:}  \grassnt{B}  \grasssym{,}  \grassmv{x_{{\mathrm{2}}}}  \grasssym{:}  \grassnt{B}  \vdash_{ \mathsf{n} }   \operatorname{\uparrow} _{ \mathsf{m}  \le  \mathsf{n} }  \grassnt{t}   \colon   \operatorname{\uparrow} _{ \mathsf{m}  \le  \mathsf{n} }  \grassnt{A}  
        }
        {
           \rho  \grasssym{,}  \grassnt{r_{{\mathrm{1}}}}  \grasssym{+}  \grassnt{r_{{\mathrm{2}}}}  \mid  \mathsf{M}  \grasssym{,}  \mathsf{m}'   \odot   \Gamma  \grasssym{,}  \grassmv{z}  \grasssym{:}  \grassnt{B}  \vdash_{ \mathsf{n} }   \operatorname{\uparrow} _{ \mathsf{m}  \le  \mathsf{n} }  \grasssym{[}  \grassmv{z}  \grasssym{/}  \grassmv{x_{{\mathrm{1}}}}  \grasssym{,}  \grassmv{z}  \grasssym{/}  \grassmv{x_{{\mathrm{2}}}}  \grasssym{]}  \grassnt{t}   \colon   \operatorname{\uparrow} _{ \mathsf{m}  \le  \mathsf{n} }  \grassnt{A}  
        }.
      \]
      The IH is
      $  \rho  \grasssym{,}  \grassnt{r_{{\mathrm{1}}}}  \grasssym{,}  \grassnt{r_{{\mathrm{2}}}}  \mid  \mathsf{M}  \grasssym{,}  \mathsf{m}'  \grasssym{,}  \mathsf{m}'   \odot   \Gamma  \grasssym{,}  \grassmv{x_{{\mathrm{1}}}}  \grasssym{:}  \grassnt{B}  \grasssym{,}  \grassmv{x_{{\mathrm{2}}}}  \grasssym{:}  \grassnt{B}  \vdash_{ \mathsf{m} }  \grassnt{t}  \colon  \grassnt{A}  $
      and we still have $  \grassnt{r_{{\mathrm{1}}}}  \grasssym{,}  \grassnt{r_{{\mathrm{2}}}}  \in \mathsf{Cont}( \mathsf{m}' )  $
      and thus may apply the rule \rulename{grass}{term}{cont} to obtain
      $  \rho  \grasssym{,}  \grassnt{r_{{\mathrm{1}}}}  \grasssym{+}  \grassnt{r_{{\mathrm{2}}}}  \mid  \mathsf{M}  \grasssym{,}  \mathsf{m}'   \odot   \Gamma  \grasssym{,}  \grassmv{z}  \grasssym{:}  \grassnt{B}  \vdash_{ \mathsf{n} }  \grasssym{[}  \grassmv{z}  \grasssym{/}  \grassmv{x_{{\mathrm{1}}}}  \grasssym{,}  \grassmv{z}  \grasssym{/}  \grassmv{x_{{\mathrm{2}}}}  \grasssym{]}  \grassnt{t}  \colon  \grassnt{A}  $
      which is the desired term.

      \subcase{sub}
      The derivation has the form
      \[
        \inferrule*
        {
           \rho'  \le  \rho 
          \and
           \rho'  \mid  \mathsf{M}   \odot   \Gamma  \vdash_{ \mathsf{n} }   \operatorname{\uparrow} _{ \mathsf{m}  \le  \mathsf{n} }  \grassnt{t}   \colon   \operatorname{\uparrow} _{ \mathsf{m}  \le  \mathsf{n} }  \grassnt{A}  
        }
        {
           \rho  \mid  \mathsf{M}   \odot   \Gamma  \vdash_{ \mathsf{n} }   \operatorname{\uparrow} _{ \mathsf{m}  \le  \mathsf{n} }  \grassnt{t}   \colon   \operatorname{\uparrow} _{ \mathsf{m}  \le  \mathsf{n} }  \grassnt{A}  
        }
      \]
      The IH is
      $  \rho'  \mid  \mathsf{M}   \odot   \Gamma  \vdash_{ \mathsf{n} }  \grassnt{t}  \colon  \grassnt{A}  $ and applying rule \rulename{grass}{term}{sub}
      with $  \rho'  \le  \rho  $ produces the desired term.

  \end{case}
  \begin{case}[\normalfont\rulename{grass}{beta}{drop}\bf]\rm
    Given $ \mathsf{l}  \le  \mathsf{n}  \le  \mathsf{m} $
    and
    \begin{mathpar}
       \sigma  \mid  \mathsf{N}   \odot   \Delta  \vdash_{ \mathsf{n} }   \operatorname{\downarrow} ^{ \grassnt{q} }_{ \mathsf{n}  \le  \mathsf{m} }  \grassnt{t}   \colon   \operatorname{\downarrow} ^{ \grassnt{q} }_{ \mathsf{n}  \le  \mathsf{m} }  \grassnt{A}  
      \and
       \rho  \grasssym{,}  \grassnt{q}  \mid  \mathsf{M}  \grasssym{,}  \mathsf{m}   \odot   \Gamma  \grasssym{,}  \grassmv{x}  \grasssym{:}  \grassnt{A}  \vdash_{ \mathsf{l} }  \grassnt{e}  \colon  \grassnt{B} ,
    \end{mathpar}
    we must show that
    \[
       \rho  \grasssym{,}  \sigma  \mid  \mathsf{M}  \grasssym{,}  \mathsf{N}   \odot   \Gamma  \grasssym{,}  \Delta  \vdash_{ \mathsf{l} }  \grasssym{[}  \grassnt{t}  \grasssym{/}  \grassmv{x}  \grasssym{]}  \grassnt{e}  \colon  \grassnt{B} .
    \]
    We proceed by induction on the derivation of
    $  \sigma  \mid  \mathsf{N}   \odot   \Delta  \vdash_{ \mathsf{n} }   \operatorname{\downarrow} ^{ \grassnt{q} }_{ \mathsf{n}  \le  \mathsf{m} }  \grassnt{t}   \colon   \operatorname{\downarrow} ^{ \grassnt{q} }_{ \mathsf{n}  \le  \mathsf{m} }  \grassnt{A}   $.
    We must analyse the cases \rulename{grass}{term}{dropI} and
    \rulenames{grass}{term}{weak,sub,cont}.
      \subcase{dropI}
      In this case there exists a grade vector $ \sigma' $ with
      $  \grassnt{q}   \sigma'  = \sigma $ such that the dervation is
      \[
        \inferrule*
        {
           \sigma'  \mid  \mathsf{N}   \odot   \Delta  \vdash_{ \mathsf{m} }  \grassnt{t}  \colon  \grassnt{A} 
        }
        {
            \grassnt{q}   \sigma'   \mid  \mathsf{N}   \odot   \Delta  \vdash_{ \mathsf{n} }   \operatorname{\downarrow} ^{ \grassnt{q} }_{ \mathsf{n}  \le  \mathsf{m} }  \grassnt{t}   \colon   \operatorname{\downarrow} ^{ \grassnt{q} }_{ \mathsf{n}  \le  \mathsf{m} }  \grassnt{A}  
        }
      \]
      Apply substitution with the judgment $  \sigma'  \mid  \mathsf{N}   \odot   \Delta  \vdash_{ \mathsf{m} }  \grassnt{t}  \colon  \grassnt{A}  $
      to obtain the desired term.

      \subcase{weak}
      In this case the derivation is
      \[
        \inferrule*
        {
           \mathsf{Weak} (  \mathsf{n}'  ) 
          \and
          \mathsf{n}'  \ge  \mathsf{n}
          \and
           \grassnt{T}  \in \mathsf{Type}( \mathsf{n}' ) 
          \and
           \sigma  \mid  \mathsf{N}   \odot   \Delta  \vdash_{ \mathsf{n} }   \operatorname{\downarrow} ^{ \grassnt{q} }_{ \mathsf{n}  \le  \mathsf{m} }  \grassnt{t}   \colon   \operatorname{\downarrow} ^{ \grassnt{q} }_{ \mathsf{n}  \le  \mathsf{m} }  \grassnt{A}  
        }
        {
           \sigma  \grasssym{,}   0   \mid  \mathsf{N}  \grasssym{,}  \mathsf{n}'   \odot   \Delta  \grasssym{,}  \grassmv{y}  \grasssym{:}  \grassnt{T}  \vdash_{ \mathsf{n} }   \operatorname{\downarrow} ^{ \grassnt{q} }_{ \mathsf{n}  \le  \mathsf{m} }  \grassnt{t}   \colon   \operatorname{\downarrow} ^{ \grassnt{q} }_{ \mathsf{n}  \le  \mathsf{m} }  \grassnt{A}  
        }
      \]
      The inductive hypothesis is
      \[
         \rho  \grasssym{,}  \sigma  \mid  \mathsf{M}  \grasssym{,}  \mathsf{N}   \odot   \Gamma  \grasssym{,}  \Delta  \vdash_{ \mathsf{l} }  \grasssym{[}  \grassnt{t}  \grasssym{/}  \grassmv{x}  \grasssym{]}  \grassnt{e}  \colon  \grassnt{B} .
      \]
      and since we have $ \mathsf{n}'  \ge  \mathsf{n}  \ge  \mathsf{l}  $ we may apply weakening to obtain
      \[
         \rho  \grasssym{,}  \sigma  \grasssym{,}   0   \mid  \mathsf{M}  \grasssym{,}  \mathsf{N}  \grasssym{,}  \mathsf{n}'   \odot   \Gamma  \grasssym{,}  \Delta  \grasssym{,}  \grassmv{y}  \grasssym{:}  \grassnt{T}  \vdash_{ \mathsf{l} }  \grasssym{[}  \grassnt{t}  \grasssym{/}  \grassmv{x}  \grasssym{]}  \grassnt{e}  \colon  \grassnt{B} .
      \]

      \subcase{sub}
      In this case the derivation is
      \[
        \inferrule*
        {
           \sigma'  \le  \sigma 
          \and
           \sigma'  \mid  \mathsf{N}   \odot   \Delta  \vdash_{ \mathsf{n} }   \operatorname{\downarrow} ^{ \grassnt{q} }_{ \mathsf{n}  \le  \mathsf{m} }  \grassnt{t}   \colon   \operatorname{\downarrow} ^{ \grassnt{q} }_{ \mathsf{n}  \le  \mathsf{m} }  \grassnt{A}  
        }
        {
           \sigma  \mid  \mathsf{N}   \odot   \Delta  \vdash_{ \mathsf{n} }   \operatorname{\downarrow} ^{ \grassnt{q} }_{ \mathsf{n}  \le  \mathsf{m} }  \grassnt{t}   \colon   \operatorname{\downarrow} ^{ \grassnt{q} }_{ \mathsf{n}  \le  \mathsf{m} }  \grassnt{A}  
        }
      \]
      From $  \sigma'  \le  \sigma  $ we obtain $  \rho  \grasssym{,}  \sigma'  \le  \rho  \grasssym{,}  \sigma  $
      and hence we may apply the rule \rulename{grass}{term}{sub}
      to the inductive hypothesis
      \[
         \rho  \grasssym{,}  \sigma'  \mid  \mathsf{M}  \grasssym{,}  \mathsf{N}   \odot   \Gamma  \grasssym{,}  \Delta  \vdash_{ \mathsf{l} }  \grasssym{[}  \grassnt{t}  \grasssym{/}  \grassmv{x}  \grasssym{]}  \grassnt{e}  \colon  \grassnt{B} .
      \]
      to obtain the desired term.

      \subcase{cont}
      In this case the derivation is
      \[
        \inferrule*
        {
           \grassnt{r_{{\mathrm{1}}}}  \grasssym{,}  \grassnt{r_{{\mathrm{2}}}}  \in \mathsf{Cont}( \mathsf{n}' ) 
          \and
           \sigma  \grasssym{,}  \grassnt{r_{{\mathrm{1}}}}  \grasssym{,}  \grassnt{r_{{\mathrm{2}}}}  \mid  \mathsf{N}  \grasssym{,}  \mathsf{n}'  \grasssym{,}  \mathsf{n}'   \odot   \Delta  \grasssym{,}  \grassmv{y_{{\mathrm{1}}}}  \grasssym{:}  \grassnt{T}  \grasssym{,}  \grassmv{y_{{\mathrm{2}}}}  \grasssym{:}  \grassnt{T}  \vdash_{ \mathsf{n} }   \operatorname{\downarrow} ^{ \grassnt{q} }_{ \mathsf{n}  \le  \mathsf{m} }  \grassnt{t}   \colon   \operatorname{\downarrow} ^{ \grassnt{q} }_{ \mathsf{n}  \le  \mathsf{m} }  \grassnt{A}  
        }
        {
           \sigma  \grasssym{,}  \grassnt{r_{{\mathrm{1}}}}  \grasssym{+}  \grassnt{r_{{\mathrm{2}}}}  \mid  \mathsf{N}  \grasssym{,}  \mathsf{n}'   \odot   \Delta  \grasssym{,}  \grassmv{z}  \grasssym{:}  \grassnt{T}  \vdash_{ \mathsf{n} }   \operatorname{\downarrow} ^{ \grassnt{q} }_{ \mathsf{n}  \le  \mathsf{m} }  \grasssym{[}  \grassmv{z}  \grasssym{/}  \grassmv{y_{{\mathrm{1}}}}  \grasssym{,}  \grassmv{z}  \grasssym{/}  \grassmv{y_{{\mathrm{2}}}}  \grasssym{]}  \grassnt{t}   \colon   \operatorname{\downarrow} ^{ \grassnt{q} }_{ \mathsf{n}  \le  \mathsf{m} }  \grassnt{A}  
        }
      \]
      The inductive hypothesis is
      \[
         \rho  \grasssym{,}  \sigma  \grasssym{,}  \grassnt{r_{{\mathrm{1}}}}  \grasssym{,}  \grassnt{r_{{\mathrm{2}}}}  \mid  \mathsf{M}  \grasssym{,}  \mathsf{N}  \grasssym{,}  \mathsf{n}'  \grasssym{,}  \mathsf{n}'   \odot   \Gamma  \grasssym{,}  \Delta  \grasssym{,}  \grassmv{y_{{\mathrm{1}}}}  \grasssym{:}  \grassnt{T}  \grasssym{,}  \grassmv{y_{{\mathrm{2}}}}  \grasssym{:}  \grassnt{T}  \vdash_{ \mathsf{l} }  \grasssym{[}  \grassnt{t}  \grasssym{/}  \grassmv{x}  \grasssym{]}  \grassnt{e}  \colon  \grassnt{B} .
      \]
      Applying \rulename{grass}{term}{cont} we obtain
      \[
         \rho  \grasssym{,}  \sigma  \grasssym{,}  \grassnt{r_{{\mathrm{1}}}}  \grasssym{+}  \grassnt{r_{{\mathrm{2}}}}  \mid  \mathsf{M}  \grasssym{,}  \mathsf{N}  \grasssym{,}  \mathsf{n}'   \odot   \Gamma  \grasssym{,}  \Delta  \grasssym{,}  \grassmv{z}  \grasssym{:}  \grassnt{T}  \vdash_{ \mathsf{l} }  \grasssym{[}  \grassmv{z}  \grasssym{/}  \grassmv{y_{{\mathrm{1}}}}  \grasssym{,}  \grassmv{z}  \grasssym{/}  \grassmv{y_{{\mathrm{2}}}}  \grasssym{]}  \grasssym{[}  \grassnt{t}  \grasssym{/}  \grassmv{x}  \grasssym{]}  \grassnt{e}  \colon  \grassnt{B} .
      \]
      Since the variables $ y_1, y_2 $ are not free in $ e $,
      this is equal to the term $ \grasssym{[}  \grasssym{[}  \grassmv{z}  \grasssym{/}  \grassmv{y_{{\mathrm{1}}}}  \grasssym{,}  \grassmv{z}  \grasssym{/}  \grassmv{y_{{\mathrm{2}}}}  \grasssym{]}  \grassnt{t}  \grasssym{/}  \grassmv{x}  \grasssym{]}  \grassnt{e} $,
      which is the desired term.

  \end{case}
  \begin{case}[\normalfont\rulename{grass}{beta}{sumL}\bf]\rm
    Given
    \begin{mathpar}
       \grassnt{q}  \ge   1  
      \and
       \rho  \grasssym{,}  \grassnt{q}  \mid  \mathsf{M}  \grasssym{,}  \mathsf{n}   \odot   \Gamma  \grasssym{,}  \grassmv{x_{{\mathrm{1}}}}  \grasssym{:}  \grassnt{A_{{\mathrm{1}}}}  \vdash_{ \mathsf{m} }  \grassnt{e_{{\mathrm{1}}}}  \colon  \grassnt{A} 
      \and
       \sigma  \mid  \mathsf{N}   \odot   \Delta  \vdash_{ \mathsf{n} }  \operatorname{\mathsf{inl} } \, \grassnt{t}  \colon   \grassnt{A_{{\mathrm{1}}}}  \oplus  \grassnt{A_{{\mathrm{2}}}}  
    \end{mathpar}
    we must show that
    \[
       \rho  \grasssym{,}   \grassnt{q}   \sigma   \mid  \mathsf{M}  \grasssym{,}  \mathsf{N}   \odot   \Gamma  \grasssym{,}  \Delta  \vdash_{ \mathsf{m} }  \grasssym{[}  \grassnt{t}  \grasssym{/}  \grassmv{x_{{\mathrm{1}}}}  \grasssym{]}  \grassnt{e_{{\mathrm{1}}}}  \colon  \grassnt{A} 
    \]
    We proceed by induction on the derivation of
    $  \sigma  \mid  \mathsf{N}   \odot   \Delta  \vdash_{ \mathsf{n} }  \operatorname{\mathsf{inl} } \, \grassnt{t}  \colon   \grassnt{A_{{\mathrm{1}}}}  \oplus  \grassnt{A_{{\mathrm{2}}}}   $.
    We must analyse the cases \rulename{grass}{term}{sumIL} and
    \rulenames{grass}{term}{weak,sub,cont}.

      \subcase{sumIL}
      In this case the derivation is
      \[
        \inferrule*
        {
           \sigma  \mid  \mathsf{N}   \odot   \Delta  \vdash_{ \mathsf{n} }  \grassnt{t}  \colon  \grassnt{A_{{\mathrm{1}}}} 
        }
        {
           \sigma  \mid  \mathsf{N}   \odot   \Delta  \vdash_{ \mathsf{n} }  \operatorname{\mathsf{inl} } \, \grassnt{t}  \colon   \grassnt{A_{{\mathrm{1}}}}  \oplus  \grassnt{A_{{\mathrm{2}}}}  
        }
      \]
      Apply substitution to obtain the desired result.

      \subcase{weak}
      In this cae the derivation is
      \[
        \inferrule*
        {
           \mathsf{Weak} (  \mathsf{n}'  ) 
          \and
          \mathsf{n}'  \ge  \mathsf{n}
          \and
           \grassnt{T}  \in \mathsf{Type}( \mathsf{n}' ) 
          \and
           \sigma  \mid  \mathsf{N}   \odot   \Delta  \vdash_{ \mathsf{n} }  \operatorname{\mathsf{inl} } \, \grassnt{t}  \colon   \grassnt{A_{{\mathrm{1}}}}  \oplus  \grassnt{A_{{\mathrm{2}}}}  
        }
        {
           \sigma  \grasssym{,}   0   \mid  \mathsf{N}  \grasssym{,}  \mathsf{n}'   \odot   \Delta  \grasssym{,}  \grassmv{y}  \grasssym{:}  \grassnt{T}  \vdash_{ \mathsf{n} }  \operatorname{\mathsf{inl} } \, \grassnt{t}  \colon   \grassnt{A_{{\mathrm{1}}}}  \oplus  \grassnt{A_{{\mathrm{2}}}}  
        }
      \]
      We have $ \mathsf{n}'  \ge  \mathsf{n}  \ge  \mathsf{m} $ and hence may apply the rule
      \rulename{grass}{term}{weak} to the inductive hypothesis,
      obtaining
      \[
        \inferrule*
        {
           \rho  \grasssym{,}   \grassnt{q}   \sigma   \mid  \mathsf{M}  \grasssym{,}  \mathsf{N}   \odot   \Gamma  \grasssym{,}  \Delta  \vdash_{ \mathsf{m} }  \grasssym{[}  \grassnt{t}  \grasssym{/}  \grassmv{x_{{\mathrm{1}}}}  \grasssym{]}  \grassnt{e_{{\mathrm{1}}}}  \colon  \grassnt{A} 
        }
        {
           \rho  \grasssym{,}   \grassnt{q}   \sigma   \grasssym{,}   0   \mid  \mathsf{M}  \grasssym{,}  \mathsf{N}  \grasssym{,}  \mathsf{n}'   \odot   \Gamma  \grasssym{,}  \Delta  \grasssym{,}  \grassmv{y}  \grasssym{:}  \grassnt{T}  \vdash_{ \mathsf{m} }  \grasssym{[}  \grassnt{t}  \grasssym{/}  \grassmv{x_{{\mathrm{1}}}}  \grasssym{]}  \grassnt{e_{{\mathrm{1}}}}  \colon  \grassnt{A} 
        }
      \]
      as desired.

      \subcase{sub}
      In this case the derivation is
      \[
        \inferrule*
        {
           \sigma'  \le  \sigma 
          \and
           \sigma'  \mid  \mathsf{N}   \odot   \Delta  \vdash_{ \mathsf{n} }  \operatorname{\mathsf{inl} } \, \grassnt{t}  \colon   \grassnt{A_{{\mathrm{1}}}}  \oplus  \grassnt{A_{{\mathrm{2}}}}  
        }
        {
           \sigma  \mid  \mathsf{N}   \odot   \Delta  \vdash_{ \mathsf{n} }  \operatorname{\mathsf{inl} } \, \grassnt{t}  \colon   \grassnt{A_{{\mathrm{1}}}}  \oplus  \grassnt{A_{{\mathrm{2}}}}  
        }
      \]
      By monotonicity of multiplication we have
      $  \rho  \grasssym{,}   \grassnt{q}   \sigma'   \le  \rho  \grasssym{,}   \grassnt{q}   \sigma   $ and hence
      we may apply the rule \rulename{grass}{term}{sub}
      to the inductive hypothesis
      \[
         \rho  \grasssym{,}   \grassnt{q}   \sigma'   \mid  \mathsf{M}  \grasssym{,}  \mathsf{N}   \odot   \Gamma  \grasssym{,}  \Delta  \vdash_{ \mathsf{m} }  \grasssym{[}  \grassnt{t}  \grasssym{/}  \grassmv{x_{{\mathrm{1}}}}  \grasssym{]}  \grassnt{e_{{\mathrm{1}}}}  \colon  \grassnt{A} 
      \]
      obtaining the desired result.

      \subcase{cont}
      In this case the derivation is
      \[
        \inferrule*
        {
           \grassnt{r_{{\mathrm{1}}}}  \grasssym{,}  \grassnt{r_{{\mathrm{2}}}}  \in \mathsf{Cont}( \mathsf{n}' ) 
          \and
           \sigma  \grasssym{,}  \grassnt{r_{{\mathrm{1}}}}  \grasssym{,}  \grassnt{r_{{\mathrm{2}}}}  \mid  \mathsf{N}  \grasssym{,}  \mathsf{n}'  \grasssym{,}  \mathsf{n}'   \odot   \Delta  \grasssym{,}  \grassmv{y_{{\mathrm{1}}}}  \grasssym{:}  \grassnt{T}  \grasssym{,}  \grassmv{y_{{\mathrm{2}}}}  \grasssym{:}  \grassnt{T}  \vdash_{ \mathsf{n} }  \operatorname{\mathsf{inl} } \, \grassnt{t}  \colon   \grassnt{A_{{\mathrm{1}}}}  \oplus  \grassnt{A_{{\mathrm{2}}}}  
        }
        {
           \sigma  \grasssym{,}  \grassnt{r_{{\mathrm{1}}}}  \grasssym{+}  \grassnt{r_{{\mathrm{2}}}}  \mid  \mathsf{N}  \grasssym{,}  \mathsf{n}'   \odot   \Delta  \grasssym{,}  \grassmv{z}  \grasssym{:}  \grassnt{T}  \vdash_{ \mathsf{n} }  \operatorname{\mathsf{inl} } \, \grasssym{[}  \grassmv{z}  \grasssym{/}  \grassmv{y_{{\mathrm{1}}}}  \grasssym{,}  \grassmv{z}  \grasssym{/}  \grassmv{y_{{\mathrm{2}}}}  \grasssym{]}  \grassnt{t}  \colon   \grassnt{A_{{\mathrm{1}}}}  \oplus  \grassnt{A_{{\mathrm{2}}}}  
        }
      \]
      The inductive hypothesis is
      \[
         \rho  \grasssym{,}   \grassnt{q}   \grasssym{(}  \sigma  \grasssym{,}  \grassnt{r_{{\mathrm{1}}}}  \grasssym{,}  \grassnt{r_{{\mathrm{2}}}}  \grasssym{)}   \mid  \mathsf{M}  \grasssym{,}  \mathsf{N}  \grasssym{,}  \mathsf{n}'  \grasssym{,}  \mathsf{n}'   \odot   \Gamma  \grasssym{,}  \Delta  \grasssym{,}  \grassmv{y_{{\mathrm{1}}}}  \grasssym{:}  \grassnt{T}  \grasssym{,}  \grassmv{y_{{\mathrm{2}}}}  \grasssym{:}  \grassnt{T}  \vdash_{ \mathsf{m} }  \grasssym{[}  \grassnt{t}  \grasssym{/}  \grassmv{x_{{\mathrm{1}}}}  \grasssym{]}  \grassnt{e_{{\mathrm{1}}}}  \colon  \grassnt{A} .
      \]
      Since $  \mathsf{Cont}( \mathsf{n}' )  $ is an ideal,
      we have $   \grassnt{q}   \grassnt{r_{{\mathrm{1}}}}   \grasssym{,}   \grassnt{q}   \grassnt{r_{{\mathrm{2}}}}   \in \mathsf{Cont}( \mathsf{n}' )  $
      and thus may apply the rule \rulename{grass}{term}{cont}
      to obtain
      \[
         \rho  \grasssym{,}   \grassnt{q}   \grasssym{(}  \sigma  \grasssym{,}  \grassnt{r_{{\mathrm{1}}}}  \grasssym{+}  \grassnt{r_{{\mathrm{2}}}}  \grasssym{)}   \mid  \mathsf{M}  \grasssym{,}  \mathsf{N}  \grasssym{,}  \mathsf{n}'   \odot   \Gamma  \grasssym{,}  \Delta  \grasssym{,}  \grassmv{z}  \grasssym{:}  \grassnt{T}  \vdash_{ \mathsf{m} }  \grasssym{[}  \grassmv{z}  \grasssym{/}  \grassmv{y_{{\mathrm{1}}}}  \grasssym{,}  \grassmv{z}  \grasssym{/}  \grassmv{y_{{\mathrm{2}}}}  \grasssym{]}  \grasssym{[}  \grassnt{t}  \grasssym{/}  \grassmv{x_{{\mathrm{1}}}}  \grasssym{]}  \grassnt{e_{{\mathrm{1}}}}  \colon  \grassnt{A} .
      \]
      Since the variables $ y_1, y_2 $ are not free in $ e_1 $,
      this term is equal to $ \grasssym{[}  \grasssym{[}  \grassmv{z}  \grasssym{/}  \grassmv{y_{{\mathrm{1}}}}  \grasssym{,}  \grassmv{z}  \grasssym{/}  \grassmv{y_{{\mathrm{2}}}}  \grasssym{]}  \grassnt{t}  \grasssym{/}  \grassmv{x_{{\mathrm{1}}}}  \grasssym{]}  \grassnt{e_{{\mathrm{1}}}} $,
      concluding the proof.

  \end{case}

  \begin{case}[\normalfont\rulename{grass}{beta}{sumR}\bf]\rm
    Analogous to the case \rulename{grass}{beta}{sumL}.
  \end{case}
\end{proof}

\newpage

\subsection{Syntactic soundness eta-conversions}
\label{asec:multi-mode-eta}

We define $ \eta $-conversion as a set of typed equalities as follows:

\begin{mathpar}
  \namedrules{grass}{eta}{unit,pair,arrow,sum,raise,drop}
\end{mathpar}

\begin{theorem}
  \label{athm:eta-syntactic-soundness}
  $ \eta $-conversion preserves types and grades:
  Whenever $  \rho  \mid  \mathsf{M}  \odot  \Gamma  \vdash_{ \mathsf{m} }  \grassnt{t_{{\mathrm{1}}}}  \equiv_\eta  \grassnt{t_{{\mathrm{2}}}}  \colon  \grassnt{A}  $
  then $  \rho  \mid  \mathsf{M}   \odot   \Gamma  \vdash_{ \mathsf{m} }  \grassnt{t_{{\mathrm{1}}}}  \colon  \grassnt{A}  $ and
  $  \rho  \mid  \mathsf{M}   \odot   \Gamma  \vdash_{ \mathsf{m} }  \grassnt{t_{{\mathrm{2}}}}  \colon  \grassnt{A}  $.
\end{theorem}

\begin{proof}
  The assertion that $  \rho  \mid  \mathsf{M}   \odot   \Gamma  \vdash_{ \mathsf{m} }  \grassnt{t_{{\mathrm{1}}}}  \colon  \grassnt{A}  $
  is an assumption in each rule for $ \eta $-expansion.

  For the assertion that $  \rho  \mid  \mathsf{M}   \odot   \Gamma  \vdash_{ \mathsf{m} }  \grassnt{t_{{\mathrm{2}}}}  \colon  \grassnt{A}  $ we have the following
  derivations:
  \begin{mathpar}
    \inferrule*
    [Left = \rulename{grass}{term}{unitE}]
    {
      \inferrule*
      [Left=\rulename{grass}{term}{unitI}]
      {
      }
      {
         \emptyset  \mid  \emptyset   \odot   \emptyset  \vdash_{ \mathsf{m} }   \star_ \mathsf{m}   \colon   \mathbf{I}_ \mathsf{m}  
      }
      \and
       \rho  \mid  \mathsf{M}   \odot   \Gamma  \vdash_{ \mathsf{m} }  \grassnt{t}  \colon   \mathbf{I}_ \mathsf{m}  
    }
    {
       \rho  \mid  \mathsf{M}   \odot   \Gamma  \vdash_{ \mathsf{m} }   \mathbin{\mathsf{let} } _{@   1  }   \star_ \mathsf{m}   =  \grassnt{t}   \mathbin{\mathsf{in} }    \star_ \mathsf{m}    \colon   \mathbf{I}_ \mathsf{m}  
    }
    \\
    \inferrule*
    [Left = \rulename{grass}{term}{pairE}]
    {
      \inferrule*
      [Left = \rulename{grass}{term}{pairI}]
      {
        \inferrule*
        [Left = \rulename{grass}{term}{var}]
        {
        }
        {
            1   \mid  \mathsf{m}   \odot   \grassmv{x_{{\mathrm{1}}}}  \grasssym{:}  \grassnt{A_{{\mathrm{1}}}}  \vdash_{ \mathsf{m} }  \grassmv{x_{{\mathrm{1}}}}  \colon  \grassnt{A_{{\mathrm{1}}}} 
        }
        \and
        \inferrule*
        [Left = \rulename{grass}{term}{var}]
        {
        }
        {
            1   \mid  \mathsf{m}   \odot   \grassmv{x_{{\mathrm{2}}}}  \grasssym{:}  \grassnt{A_{{\mathrm{2}}}}  \vdash_{ \mathsf{m} }  \grassmv{x_{{\mathrm{2}}}}  \colon  \grassnt{A_{{\mathrm{2}}}} 
        }
      }
      {
          1   \grasssym{,}   1   \mid  \mathsf{m}  \grasssym{,}  \mathsf{m}   \odot   \grassmv{x_{{\mathrm{1}}}}  \grasssym{:}  \grassnt{A_{{\mathrm{1}}}}  \grasssym{,}  \grassmv{x_{{\mathrm{2}}}}  \grasssym{:}  \grassnt{A_{{\mathrm{2}}}}  \vdash_{ \mathsf{m} }  \grasssym{(}  \grassmv{x_{{\mathrm{1}}}}  \grasssym{,}  \grassmv{x_{{\mathrm{2}}}}  \grasssym{)}  \colon   \grassnt{A_{{\mathrm{1}}}}  \otimes  \grassnt{A_{{\mathrm{2}}}}  
      }
      \and
       \rho  \mid  \mathsf{M}   \odot   \Gamma  \vdash_{ \mathsf{m} }  \grassnt{t}  \colon   \grassnt{A_{{\mathrm{1}}}}  \otimes  \grassnt{A_{{\mathrm{2}}}}  
    }
    {
       \rho  \mid  \mathsf{M}   \odot   \Gamma  \vdash_{ \mathsf{m} }   \mathbin{\mathsf{let} } _{@   1  }  \grasssym{(}  \grassmv{x_{{\mathrm{1}}}}  \grasssym{,}  \grassmv{x_{{\mathrm{2}}}}  \grasssym{)}  =  \grassnt{t}   \mathbin{\mathsf{in} }   \grasssym{(}  \grassmv{x_{{\mathrm{1}}}}  \grasssym{,}  \grassmv{x_{{\mathrm{2}}}}  \grasssym{)}   \colon   \grassnt{A_{{\mathrm{1}}}}  \otimes  \grassnt{A_{{\mathrm{2}}}}  
    }
    \\
    \inferrule*
    [Left = \rulename{grass}{term}{arrowI}]
    {
      \inferrule*
      [Left = \rulename{grass}{term}{arrowE}]
      {
         \rho  \mid  \mathsf{M}   \odot   \Gamma  \vdash_{ \mathsf{n} }  \grassnt{t}  \colon   \grassnt{A} ^{ \grassnt{q} : \mathsf{m} } \multimap  \grassnt{B}  
        \and
        \inferrule*
        [Left = \rulename{grass}{term}{var}]
        {
        }
        {
            1   \mid  \mathsf{m}   \odot   \grassmv{x}  \grasssym{:}  \grassnt{A}  \vdash_{ \mathsf{m} }  \grassmv{x}  \colon  \grassnt{A} 
        }
      }
      {
         \rho  \grasssym{,}  \grassnt{q}  \mid  \mathsf{M}  \grasssym{,}  \mathsf{m}   \odot   \Gamma  \grasssym{,}  \grassmv{x}  \grasssym{:}  \grassnt{A}  \vdash_{ \mathsf{n} }  \grassnt{t} \, \grassmv{x}  \colon  \grassnt{B} 
      }
    }
    {
       \rho  \mid  \mathsf{M}   \odot   \Gamma  \vdash_{ \mathsf{n} }  \lambda  \grassmv{x}  \grasssym{.}  \grasssym{(}  \grassnt{t} \, \grassmv{x}  \grasssym{)}  \colon   \grassnt{A} ^{ \grassnt{q} : \mathsf{m} } \multimap  \grassnt{B}  
    }
    \\
    \inferrule*
    [Left = \rulename{grass}{term}{raiseI}]
    {
      \inferrule*
      [Left = \rulename{grass}{term}{raiseE}]
      {
         \rho  \mid  \mathsf{M}   \odot   \Gamma  \vdash_{ \mathsf{m} }  \grassnt{t}  \colon   \operatorname{\uparrow} _{ \mathsf{n}  \le  \mathsf{m} }  \grassnt{A}  
      }
      {
         \rho  \mid  \mathsf{M}   \odot   \Gamma  \vdash_{ \mathsf{n} }   \operatorname{\uparrow}^{-1} _{ \mathsf{n}  \le  \mathsf{m} }  \grassnt{t}   \colon  \grassnt{A} 
      }
    }
    {
       \rho  \mid  \mathsf{M}   \odot   \Gamma  \vdash_{ \mathsf{m} }   \operatorname{\uparrow} _{ \mathsf{n}  \le  \mathsf{m} }  \grasssym{(}   \operatorname{\uparrow}^{-1} _{ \mathsf{n}  \le  \mathsf{m} }  \grassnt{t}   \grasssym{)}   \colon   \operatorname{\uparrow} _{ \mathsf{n}  \le  \mathsf{m} }  \grassnt{A}  
    }
    \\
    \inferrule*
    [Left = \rulename{grass}{term}{dropE}]
    {
      \inferrule*
      [Left = \rulename{grass}{term}{dropI}]
      {
        \inferrule*
        [Left = \rulename{grass}{term}{var}]
        {
        }
        {
            1   \mid  \mathsf{m}   \odot   \grassmv{z}  \grasssym{:}  \grassnt{A}  \vdash_{ \mathsf{m} }  \grassmv{z}  \colon  \grassnt{A} 
        }
      }
      {
         \grassnt{q}  \mid  \mathsf{m}   \odot   \grassmv{z}  \grasssym{:}  \grassnt{A}  \vdash_{ \mathsf{n} }   \operatorname{\downarrow} ^{ \grassnt{q} }_{ \mathsf{n}  \le  \mathsf{m} }  \grassmv{z}   \colon   \operatorname{\downarrow} ^{ \grassnt{q} }_{ \mathsf{n}  \le  \mathsf{m} }  \grassnt{A}  
      }
      \and
       \rho  \mid  \mathsf{M}   \odot   \Gamma  \vdash_{ \mathsf{n} }  \grassnt{t}  \colon   \operatorname{\downarrow} ^{ \grassnt{q} }_{ \mathsf{n}  \le  \mathsf{m} }  \grassnt{A}  
    }
    {
       \rho  \mid  \mathsf{M}   \odot   \Gamma  \vdash_{ \mathsf{n} }   \mathbin{\mathsf{let} } _{@  \grassnt{q} }   \operatorname{\downarrow} _{ \mathsf{n}  \le  \mathsf{m} }  \grassmv{z}   =  \grassnt{t}   \mathbin{\mathsf{in} }    \operatorname{\downarrow} ^{ \grassnt{q} }_{ \mathsf{n}  \le  \mathsf{m} }  \grassmv{z}    \colon   \operatorname{\downarrow} ^{ \grassnt{q} }_{ \mathsf{n}  \le  \mathsf{m} }  \grassnt{A}  
    }
    \\
    \inferrule*
    [Left = \rulename{grass}{term}{sumE}]
    {
      \inferrule*[Left=\rulename{grass}{term}{sumIL}]
      {
        \inferrule*[Left=\rulename{grass}{term}{var}]
        {
        }
        {  1   \mid  \mathsf{m}   \odot   \grassmv{x_{{\mathrm{1}}}}  \grasssym{:}  \grassnt{A_{{\mathrm{1}}}}  \vdash_{ \mathsf{m} }  \grassmv{x_{{\mathrm{1}}}}  \colon  \grassnt{A_{{\mathrm{1}}}} }
      }
      {
          1   \mid  \mathsf{m}   \odot   \grassmv{x_{{\mathrm{1}}}}  \grasssym{:}  \grassnt{A_{{\mathrm{1}}}}  \vdash_{ \mathsf{m} }  \operatorname{\mathsf{inl} } \, \grassmv{x_{{\mathrm{1}}}}  \colon   \grassnt{A_{{\mathrm{1}}}}  \oplus  \grassnt{A_{{\mathrm{2}}}}  
      }
      \
      \inferrule*[left=\rulename{grass}{term}{sumIR}]
      {
        \inferrule*[Left=\rulename{grass}{term}{var}]
        {
        }
        {  1   \mid  \mathsf{m}   \odot   \grassmv{x_{{\mathrm{2}}}}  \grasssym{:}  \grassnt{A_{{\mathrm{2}}}}  \vdash_{ \mathsf{m} }  \grassmv{x_{{\mathrm{2}}}}  \colon  \grassnt{A_{{\mathrm{2}}}} }
      }
      {
          1   \mid  \mathsf{m}   \odot   \grassmv{x_{{\mathrm{2}}}}  \grasssym{:}  \grassnt{A_{{\mathrm{2}}}}  \vdash_{ \mathsf{m} }  \operatorname{\mathsf{inr} } \, \grassmv{x_{{\mathrm{1}}}}  \colon   \grassnt{A_{{\mathrm{1}}}}  \oplus  \grassnt{A_{{\mathrm{2}}}}  
      }
      \and
       \rho  \mid  \mathsf{M}   \odot   \Gamma  \vdash_{ \mathsf{m} }  \grassnt{t}  \colon   \grassnt{A_{{\mathrm{1}}}}  \oplus  \grassnt{A_{{\mathrm{2}}}}  
    }
    {
       \rho  \mid  \mathsf{M}  \odot  \Gamma  \vdash_{ \mathsf{m} }  \grassnt{t}  \equiv_\eta   \mathsf{case}_  1  ( \grassnt{t} ; \grassmv{x_{{\mathrm{1}}}} . \operatorname{\mathsf{inl} } \, \grassmv{x_{{\mathrm{1}}}} ; \grassmv{x_{{\mathrm{2}}}} . \operatorname{\mathsf{inr} } \, \grassmv{x_{{\mathrm{2}}}} )   \colon   \grassnt{A_{{\mathrm{1}}}}  \oplus  \grassnt{A_{{\mathrm{2}}}}  
    }
  \end{mathpar}
\end{proof}

\newpage
\section{Semantic soundness of the equational theory}
\label{asec:semantic-soundness}

\subsection{Substitution is composition}
\label{asec:subst-comp}
  \begin{lemma}
  \label{alem:subst-comp}
  In the situation of Theorem \ref{athm:multi-subst},
  suppose the terms $ t $ and $ e_1, \mathellipsis, e_k $ have categorical
  semantics
  \begin{mathpar}
	\interpCtx{\rho}{\mode M}{\Gamma}{\mode m}
    = \bigotimes_i F^{\mode m_i}_{\mode m}(r_i \at A_i)
    \xrightarrow{t} B
    \and
    \interpCtx{\sigma_i}{\mode N_i}{\Delta_i}{\mode m_i}
    = \bigotimes_j F^{\mode n_{ij}}_{\mode m_i} (q_{ij} \at A_{ij})
    \xrightarrow{e_i} A_i\
    (i \in \{1, \mathellipsis, k\})
  \end{mathpar}
  The interpretation of the substition
  $   \grassnt{r_{{\mathrm{1}}}}   \sigma_{{\mathrm{1}}}   \grasssym{,}  \mathellipsis  \grasssym{,}   \grassnt{r_{\grassmv{k}}}   \sigma_{\grassmv{k}}   \mid  \mathsf{N}_{{\mathrm{1}}}  \grasssym{,}  \mathellipsis  \grasssym{,}  \mathsf{N}_{\grassmv{k}}   \odot   \Delta_{{\mathrm{1}}}  \grasssym{,}  \mathellipsis  \grasssym{,}  \Delta_{\grassmv{k}}  \vdash_{ \mathsf{l} }  \grasssym{[}  \grassnt{e_{\grassmv{i}}}  \grasssym{/}  \grassmv{x_{\grassmv{i}}}  \grasssym{]}  \grassnt{t}  \colon  \grassnt{B}  $
  is given by the morphism
  \begin{mathpar}
    \bigotimes_i
    \interp{   \grassnt{r_{\grassmv{i}}}   \sigma_{\grassmv{i}}   \mid  \mathsf{N}_{\grassmv{i}}  \odot  \Delta_{\grassmv{i}}  }_{\mathsf{m}}
    \iso
    \bigotimes_i
    F^{\mode m_i}_{\mode m}
    \interp{   \grassnt{r_{\grassmv{i}}}   \sigma_{\grassmv{i}}   \mid  \mathsf{N}_{\grassmv{i}}  \odot  \Delta_{\grassmv{i}}  }_{\mathsf{m}_{\grassmv{i}}}
    \xrightarrow{\bigotimes_i F^{\mode m_i}_{\mode m}(r_i \cdot e_i)}
    \interp{  \rho  \mid  \mathsf{M}  \odot  \Gamma }_{\mathsf{m}}
    \xrightarrow{t}
    B
  \end{mathpar}
  or, equivalently:
  \begin{mathpar}
    \bigotimes_i
    \bigotimes_j
    F^{\mode n_{ij}}_{\mode m}(r_i q_{ij} \at A_{ij})
    \iso
    \bigotimes_i
    F^{\mode m_i}_{\mode m}
    \bigotimes_j
    F^{\mode n_{ij}}_{\mode m_i}(r_i q_{ij} \at A_{ij})
    \xrightarrow{\bigotimes_i F^{\mode m_i}_{\mode m}(r_i \cdot e_i)}
    \bigotimes_i
    F^{\mode m_i}_{\mode m} (r_i \at A_i)
    \xrightarrow{t}
    B
  \end{mathpar}
\end{lemma}

Before the proof, we begin with a lemma.

\begin{lemma}
  \label{alem:scalar-mult-comp}
  In the situation of the above lemma,
  suppose that $ s \in R_{\mode m} $ and
  write $ f $ for the composite
  \begin{mathpar}
    \bigotimes_{ij}
    F^{\mode n_{ij}}_{\mode m}(r_i q_{ij} \at A_{ij})
    \iso
    \bigotimes_i
    F^{\mode m_i}_{\mode m}
    \bigotimes_j
    F^{\mode n_{ij}}_{\mode m_i}(r_i q_{ij} \at A_{ij})
    \xrightarrow{\bigotimes_i F^{\mode m_i}_{\mode m}(r_i \cdot e_i)}
    \bigotimes_i
    F^{\mode m_i}_{\mode m} (r_i \at A_i)
    \xrightarrow{t}
    B.
  \end{mathpar}
  The morphism $ s \cdot f $ is equal to the composite
  \begin{mathpar}
    \bigotimes_{ij}
    F^{\mode n_{ij}}_{\mode m}(s r_i q_{ij} \at A_{ij})
    \iso
    \bigotimes_i
    F^{\mode m_i}_{\mode m}
    \bigotimes_j
    F^{\mode n_{ij}}_{\mode m_i}(sr_i q_{ij} \at A_{ij})
    \xrightarrow{
      \bigotimes_i F^{\mode m_i}_{\mode m}(sr_i \cdot e_i)
    }
    \bigotimes_i F^{\mode m_i}_{\mode m} (sr_i \at A_i)
    \xrightarrow{(s \cdot t)}
    (s \at B)
  \end{mathpar}
\end{lemma}

\begin{proof}
  Consider the following diagram
  in which unlabelled arrows are assembled from $ \mu $, $ \delta $ and $ \tau $,
  to ``move scalars out of the tensor product,''
  and the vertical arrows labeled ``$ \sim $'' are induced by the isomorphisms
  $ F^{\mode n_{ij}}_{\mode m}
  \cong F^{\mode m_i}_{\mode m} \circ F^{\mode n_{ij}}_{\mode m_i} $
  and the strong monoidal structures on the functors $ F^{\mode m_i}_{\mode m} $.
  \[
	\begin{tikzcd}
      \bigotimes_{ij}F^{\mode n_{ij}}_{\mode m}
      (s r_i q_{ij} \at A_{ij})
      \arrow[d, "\sim"']
      \arrow[r]
      &
      s \at \bigotimes_{ij}F^{\mode n_{ij}}_{\mode m}
      (r_i q_{ij} \at A_{ij})
      \arrow[d, "\sim"]
      \\
      \bigotimes_i F^{\mode m_i}_{\mode m}
      \bigotimes_j F^{\mode n_{ij}}_{\mode m_i}
      (s r_i q_{ij} \at A_{ij})
      \arrow[r]
      \arrow[d]
      &
      s \at
      \bigotimes_i F^{\mode m_i}_{\mode m}
      \bigotimes_j F^{\mode n_{ij}}_{\mode m_i}
      (r_i q_{ij} \at A_{ij})
      \arrow[d]
      \\
      \bigotimes_i F^{\mode m_i}_{\mode m}
      \big(
      s r_i \at
      \bigotimes_j F^{\mode n_{ij}}_{\mode m_i}
      (q_{ij} \at A_{ij})
      \big)
      \arrow[r]
      \arrow[d, "\bigotimes_i F^{\mode m_i}_{\mode m}(sr_i \at e_i)"']
      &
      s \at
      \bigotimes_i F^{\mode m_i}_{\mode m}
      \big(
      r_i \at
      \bigotimes_j F^{\mode n_{ij}}_{\mode m_i}
      (q_{ij} \at A_{ij})
      \big)
      \arrow[d, "s \at \bigotimes_i F^{\mode m_i}_{\mode m}(r_i \at e_i)"]
      \\
      \bigotimes_i F^{\mode m_i}_{\mode m} (s r_i \at A_i)
      \arrow[r]
      &
      s \at \bigotimes_i F^{\mode m_i}_{\mode m}(r_i \at A_i)
      \arrow[d, "s \at t"]
      \\
      &
      s \at B
    \end{tikzcd}
  \]
  The composite of the top and right edge is $ s \cdot f $,
  while the composite of the left and bottom edge is the morphism from the claim.
  The bottom square commutes by naturality of $ \mu $, $ \delta $ and $ \tau $.
  The middle square commutes by the coherence conditions on $ \mu $;
  in fact those conditions imply that all morphisms assembled from $ \delta $,
  $ \mu $ and $ \tau $ beginning at the top left of that square
  at the bottom right of that square are equal.
  The top square commutes by the naturality of the involved transformations,
  using that $ \mu $ is a monoidal natural transformation.
  Thus the diagram commutes, concluding the proof.
\end{proof}

\begin{proof}[Of Lemma \ref{alem:subst-comp}]
  \newcommand*{\rname}[1]{\normalfont\rulename{grass}{term}{#1}\bf}
  \setcounter{case}{0}
  The proof of Theorem \ref{athm:multi-subst}
  gives a constructive transformation of the derivations of $ t $ and the $ e_i $
  into a derivation of $ \grasssym{[}  \grassnt{e_{\grassmv{i}}}  \grasssym{/}  \grassmv{x_{\grassmv{i}}}  \grasssym{]}  \grassnt{t} $.
  We proceed by induction on the derivation of $ t $,
  inspecting the proof of Theorem \ref{athm:multi-subst}.

  \begin{case}[\rname{var}]\rm
    In this case we must verify that any morphism
    \[
      \interp{  \rho  \mid  \mathsf{M}  \odot  \Gamma }_{\mode m} \xrightarrow{t} B
    \]
    is equal to
    \[
      F^{\mode m}_{\mode m} \interp{    1    \rho   \mid  \mathsf{M}  \odot  \Gamma }_{\mode m}
      \xrightarrow{F^{\mode m}_{\mode m}(1 \cdot t)}
      F^{\mode m}_{\mode m}(1 \at B)
      \iso
      1 \at B
      \xrightarrow{\epsilon}
      B
    \]
    which follows from the following commutative diagram
    \[
      \begin{tikzcd}
        \bigotimes_i
        F^{\mode m_i}_{\mode m} (1 \cdot r_i \at A_i)
        \arrow[d, equals]
        \arrow[r, "\bigotimes_i F^{\mode m_i}_{\mode m}\delta"]
        &
        \bigotimes_i
        F^{\mode m_i}_{\mode m} (1 \at r_i \at A_i)
        \arrow[ld, "\bigotimes_i F^{\mode m_i}_{\mode m} \epsilon" description]
        \arrow[r, "\bigotimes_i \mu"]
        &
        \bigotimes_i
        1 \at F^{\mode m_i}_{\mode m} (r_i \at A_i)
        \arrow[ld, "\bigotimes_i \epsilon" description]
        \arrow[r, "\tau"]
        &
        1 \at \bigotimes_i
        F^{\mode m_i}_{\mode m} (r_i \at A_i)
        \arrow[ld, "\epsilon" description]
        \arrow[d, "1 \at e"]
        \\
        \bigotimes_i
        F^{\mode m_i}_{\mode m} (r_i \at A_i)
        \arrow[r, equals]
        &
        \bigotimes_i
        F^{\mode m_i}_{\mode m} (r_i \at A_i)
        \arrow[r, equals]
        &
        \bigotimes_i
        F^{\mode m_i}_{\mode m} (r_i \at A_i)
        \arrow[d, "e"']
        &
        1 \at B
        \arrow[ld, "\epsilon"]
        \\
        &&
        B
      \end{tikzcd}
    \]
  \end{case}

  \begin{case}[\rname{weak}]\rm
    Consider a term
    $  \sigma  \mid  \mathsf{N}   \odot   \Delta  \vdash_{ \mathsf{n} }  \grassnt{e}  \colon  \grassnt{T}  $
    with categorical interpretation
    \[
      \interpCtx{\sigma}{\mode N}{\Delta}{\mode n} =
      \bigotimes_j F^{\mode n_j}_{\mode n} (s_j \at Y_j)
      \xrightarrow{e}
      T
    \]
    Consider the following morphism:
    \begin{mathpar}
      \bigotimes_i
      F^{\mode m_i}_{\mode m}
      \interpCtx{\sigma_i}{\mode N_i}{\Delta_i}{\mode m_i}
      \ox
      F^{\mode n}_{\mode m}
      \Big(\bigotimes_j F^{\mode n_j}_{\mode n}(0s_j \at Y_j) \Big)
      \\
      \xrightarrow{
        \bigotimes_i F^{\mode m_i}_{\mode m}(r_i\cdot e_i)
        \ox
        F^{\mode n}_{\mode m}(0 \cdot e)
      }
      \interpCtx{\rho}{\mode M}{\Gamma}{\mode m} \ox F^{\mode n}_{\mode m}(0 \at T)
      \xrightarrow{t \ox F^{\mode n}_{\mode m} w_T} B \ox F^{\mode n}_{\mode m}(I) \iso B
    \end{mathpar}
    This morphism is the morphism we claim realizes weakening by a variable of type
    $ T $ and then substituting the term $ e $ for that variable.
    On the other hand, the proof of Theorem \ref{athm:multi-subst},
    constructs this term as a sequence of weakenings,
    which has the following interpretation:
    \begin{mathpar}
      \bigotimes_i
      F^{\mode m_i}_{\mode m}
      \interpCtx{\sigma_i}{\mode N_i}{\Delta_i}{\mode m_i}
      \ox
      \bigotimes_j
      F^{\mode n_j}_{\mode m}(0\cdot s_j \at Y_j)
      \\
      \xrightarrow{
        \bigotimes_i F^{\mode m_i}_{\mode m}(r_i\cdot e_i)
        \ox
        \bigotimes_j F^{\mode n_j}_{\mode m}(w_{Y_j})
      }
      \interpCtx{\rho}{\mode M}{\Gamma}{\mode m} \ox \bigotimes_j F^{\mode n_j}_{\mode m}I
      \xrightarrow{t \ox \bigotimes_j u^{-1}}
      B \ox \bigotimes_j I
      \iso B
    \end{mathpar}
    Our task is therefore to prove that the two composites above are equal.
    This boils down to showing the commutation of the diagram below,
    in which the composition across the left edge is the definition of $ 0 \cdot e $.
    The top square commutes by the coherence conditions on exponential actions,
    the second square due to the coherence conditions on $ \mu $,
    the third square due to the fact that $ w $ is a monoidal natural transformation
    and the last square by naturality of $ w $.

    \begin{mathpar}
      \begin{tikzcd}[column sep = 3cm]
        \bigotimes_j F^{\mode n_j}_{\mode n}(0s_j \at Y_j)
        \arrow[d, "\bigotimes_j F^{\mode n_j}_{\mode n}(\delta)"']
        \arrow[r, "\bigotimes_j F^{\mode n_j}_{\mode n}(w)"]
        &
        \bigotimes_j I
        \arrow[d, equals]
        \\
        \bigotimes_j F^{\mode n_j}_{\mode n}(0 \at s_j \at Y_j)
        \arrow[d, "\bigotimes_{j} \mu"']
        \arrow[r, "\bigotimes_j F^{\mode n_j}_{\mode n}(w)"]
        &
        \bigotimes_j I
        \arrow[d, equals]
        \\
        \bigotimes_j 0 \at F^{\mode n_j}_{\mode n}(s_j \at Y_j)
        \arrow[d, "\tau"']
        \arrow[r, "\bigotimes_j w"]
        &
        \bigotimes_j I
        \arrow[d, "\sim"]
        \\
        0 \at \bigotimes_j F^{\mode n_j}_{\mode n}(s_j \at Y_j)
        \arrow[d, "0 \at e"']
        \arrow[r, "w"]
        &
        I
        \arrow[d, equals]
        \\
        0 \at T
        \arrow[r, "w"]
        &
        I
      \end{tikzcd}
    \end{mathpar}
  \end{case}

  \begin{case}[\rname{sub}]\rm
    Suppose we have
    $ \rho' = r_1', \mathellipsis, r_k' \ge r_1, \mathellipsis, r_k = \rho $.
    Consider the diagram
    \begin{mathpar}
      \begin{tikzcd}
        \bigotimes_{i, j} F^{\mode n_{ij}}_{\mode m}(r_i q_{ij} \at A_{ij})
        \arrow[d, "\sim"']
        &
        \bigotimes_{i, j} F^{\mode n_{ij}}_{\mode m}(r_i' q_{ij} \at A_{ij})
        \arrow[d, "\sim"]
        \arrow[l]
        \\
        \bigotimes_i
        F^{\mode m_i}_{\mode m}
        \bigotimes_j
        F^{\mode n_{ij}}_{\mode m_i}(r_i q_{ij} \at A_{ij})
        \arrow[d,
        "\bigotimes_i F^{\mode m_i}_{\mode m}
        (\tau \circ \bigotimes_j (\mu \circ F^{\mode n_{ij}}_{\mode m_i} (\delta)))
        "']
        &
        \bigotimes_i
        F^{\mode m_i}_{\mode m}
        \bigotimes_j
        F^{\mode n_{ij}}_{\mode m_i}(r_i' q_{ij} \at A_{ij})
        \arrow[d,
        "\bigotimes_i F^{\mode m_i}_{\mode m}
        (\tau \circ \bigotimes_j (\mu \circ F^{\mode n_{ij}}_{\mode m_i} (\delta)))
        "]
        \arrow[l]
        \\
        \bigotimes_i
        F^{\mode m_i}_{\mode m}
        \big(
        r_i \at
        \bigotimes_j
        F^{\mode n_{ij}}_{\mode m_i}(q_{ij} \at A_{ij})
        \big)
        \arrow[d, "\bigotimes_i F^{\mode m_i}_{\mode m}(r_i \at e_i)"']
        &
        \bigotimes_i
        F^{\mode m_i}_{\mode m}
        \big(
        r_i' \at
        \bigotimes_j
        F^{\mode n_{ij}}_{\mode m_i}(q_{ij} \at A_{ij})
        \big)
        \arrow[l]
        \arrow[d, "\bigotimes_i F^{\mode m_i}_{\mode m}(r_i' \at e_i)"]
        \\
        \bigotimes_i F^{\mode m_i}_{\mode m}(r_i \at A_i)
        \arrow[d, "t"']
        &
        \bigotimes_i F^{\mode m_i}_{\mode m}(r_i' \at A_i)
        \arrow[l]
        \arrow[dl, "t"]
        \\
        B
      \end{tikzcd}
    \end{mathpar}
    Where all horizontal arrows are induced by $  \rho'  \ge  \rho  $.
    By the inductive hypothesis,
    the composite across the left edge is the interpretation
    of the term $ \grasssym{[}  \grassnt{e_{\grassmv{i}}}  \grasssym{/}  \grassmv{x_{\grassmv{i}}}  \grasssym{]}  \grassnt{t} $.
    Since the interpretation of the \rulename{grass}{term}{sub} rule is given
    by precomposing with the action of $  \rho'  \ge  \rho  $,
    the composite of the left edge with the topmost horizontal morphism is the
    interpetation of applying \rulename{grass}{term}{sub} to the the interpretation
    of $ \grasssym{[}  \grassnt{e_{\grassmv{i}}}  \grasssym{/}  \grassmv{x_{\grassmv{i}}}  \grasssym{]}  \grassnt{t} $.
    The composite across the right edge of the diagram is the morphism from our claim,
    so we must show that this diagram commutes.
    The top two squares commute by the naturality of the involved vertical morphisms,
    the third square commutes by the functoriality of the $ F^{\mode m_i}_{\mode m} $.
    The bottom triangle commutes by definition of the interpretation of the
    \rulename{grass}{term}{sub} rule.
  \end{case}

  \begin{case}[\rname{cont}]\rm
    Suppose the derivation is
    \[
      \inferrule*
      {
         \rho  \grasssym{,}  \grassnt{q_{{\mathrm{1}}}}  \grasssym{,}  \grassnt{q_{{\mathrm{2}}}}  \mid  \mathsf{M}  \grasssym{,}  \mathsf{n}  \grasssym{,}  \mathsf{n}   \odot   \Gamma  \grasssym{,}  \grassmv{y_{{\mathrm{1}}}}  \grasssym{:}  \grassnt{T}  \grasssym{,}  \grassmv{y_{{\mathrm{2}}}}  \grasssym{:}  \grassnt{T}  \vdash_{ \mathsf{m} }  \grassnt{t}  \colon  \grassnt{B} 
      }
      {
         \rho  \grasssym{,}  \grassnt{q_{{\mathrm{1}}}}  \grasssym{+}  \grassnt{q_{{\mathrm{2}}}}  \mid  \mathsf{M}  \grasssym{,}  \mathsf{n}  \grasssym{,}  \mathsf{n}   \odot   \Gamma  \grasssym{,}  \grassmv{z}  \grasssym{:}  \grassnt{T}  \vdash_{ \mathsf{m} }  \grasssym{[}  \grassmv{z}  \grasssym{/}  \grassmv{y_{{\mathrm{1}}}}  \grasssym{,}  \grassmv{z}  \grasssym{/}  \grassmv{y_{{\mathrm{2}}}}  \grasssym{]}  \grassnt{t}  \colon  \grassnt{B} 
      }
    \]
    and we are given a term $  \sigma  \mid  \mathsf{N}   \odot   \Delta  \vdash_{ \mathsf{n} }  \grassnt{e}  \colon  \grassnt{T}  $.
    Consider the diagram
    \begin{center}
      \begin{tikzcd}[every cell/.append style={align=center}]
        \mlnode{
          $ \bigotimes_i F^{\mode m_i}_{\mode m}
          \interpCtx{r_i \sigma_i}{\mode N_i}{\Delta_i}{\mode m_i}
          $
          \\
          $
          \ox
          F^{\mode n}_{\mode m} \interpCtx{q_1 \sigma}{\mode N}{\Delta}{\mode n}
          \ox
          F^{\mode n}_{\mode m} \interpCtx{q_2 \sigma}{\mode N}{\Delta}{\mode n}
          $
        }
        \arrow[dd, "\bigotimes_i F^{\mode m_i}_{\mode m}(r_i \cdot e_i)
        \ox F^{\mode n}_{\mode m} (q_1 \cdot e)
        \ox F^{\mode n}_{\mode m} (q_2 \cdot e)"']
        &
        $ \bigotimes_i F^{\mode m_i}_{\mode m}
        \interpCtx{r_i \sigma_i}{\mode N_i}{\Delta_i}{\mode m_i}
        \ox
        F^{\mode n}_{\mode m} \interpCtx{(q_1 + q_2) \sigma}{\mode N}{\Delta}{\mode n}
        $
        \arrow[d, "\bigotimes_{i} F^{\mode m_i}_{\mode m}(r_i \cdot e_i) \ox F^{\mode n}_{\mode m}((q_1 + q_2) \cdot e)"]
        \arrow[l, dashed]
        \\
        &
        $ \bigotimes_i F^{\mode m_i}_{\mode m} (r_i \at A_i)
        \ox F^{\mode n}_{\mode m}((q_1 + q_2) \at T) $
        \arrow[d, "\id \ox F^{\mode n}_{\mode m}(c)"]
        \\
        $
        \bigotimes_{i} F^{\mode m_i}_{\mode m}(r_i \at A_i)
        \ox F^{\mode n}_{\mode m} (q_1 \at T) \ox F^{\mode n}_{\mode m} (q_2 \at T)
        $
        \arrow[d, "t"']
        &
        $ \bigotimes_i F^{\mode m_i}_{\mode m} (r_i \at A_i)
        \ox F^{\mode n}_{\mode m}(q_1 \at T \ox q_2 \at T) $
        \arrow[l, "\id \ox m^{-1}"']
        \arrow[ld, "t \circ (\id \ox m^{-1})"]
        \\
        $ B $
      \end{tikzcd}
    \end{center}
    By the inductive hypothesis, the composite across the left edge is the
    interpretation of the term $ \grasssym{[}  \grassnt{e_{\grassmv{i}}}  \grasssym{/}  \grassmv{x_{\grassmv{i}}}  \grasssym{,}  \grassnt{e}  \grasssym{/}  \grassmv{y_{{\mathrm{1}}}}  \grasssym{,}  \grassnt{e}  \grasssym{/}  \grassmv{y_{{\mathrm{2}}}}  \grasssym{]}  \grassnt{t} $,
    while the composition across the right edge is the morphism we claim
    to be the interpretation of the term $ \grasssym{[}  \grassnt{e_{\grassmv{i}}}  \grasssym{/}  \grassmv{x_{\grassmv{i}}}  \grasssym{,}  \grassnt{e}  \grasssym{/}  \grassmv{z}  \grasssym{]}  \grasssym{[}  \grassmv{z}  \grasssym{/}  \grassmv{y_{{\mathrm{1}}}}  \grasssym{,}  \grassmv{z}  \grasssym{/}  \grassmv{y_{{\mathrm{2}}}}  \grasssym{]}  \grassnt{t} $.
    The proof of Theorem \ref{athm:multi-subst} gives the interpretation
    of this term as the composite of the left edge with the dashed arrow,
    where the dashed arrow is a sequence of exchanges
    followed by a sequence of contractions.
    So, we need to show that the above diagram commutes.
    It suffices to show that the following diagram commutes:
    \[
      \begin{tikzcd}
        \interpCtx{q_1 \sigma}{\mode N}{\Delta}{\mode n}
        \ox
        \interpCtx{q_2 \sigma}{\mode N}{\Delta}{\mode n}
        \arrow[d, "q_1 \cdot e \ox q_2 \cdot e"]
        &
        \interpCtx{(q_1 + q_2) \sigma}{\mode N}{\Delta}{\mode n}
        \arrow[l, dashed]
        \arrow[d, "(q_1 + q_2) \cdot e"]
        \\
        (q_1 \at T) \ox (q_2 \at T)
        &
        \arrow[l, "c"]
        (q_1 + q_2) \at T
      \end{tikzcd}
    \]
    Suppose that $ \interp{  \grasssym{(}  \grassnt{q_{{\mathrm{1}}}}  \grasssym{+}  \grassnt{q_{{\mathrm{2}}}}  \grasssym{)}   \sigma   \mid  \mathsf{N}  \odot  \Delta } _{\mode n} $
    is the object $ \bigotimes_i F^{\mode n_i}_{\mode n}((q_1 + q_2)s_i \at X_i) $
    and consider the diagram in Figure \ref{afig:subst-comp-cont}.
    Its left vertical edge is the dashed arrow, and the bottom edge is
    $ q_1 \cdot e \ox q_2 \cdot e $.
    The composite across the top/right edges is $ c \circ (q_1 + q_2) $.
    On that diagram's left side, the top square commutes by the coherence conditions
    on graded actions,
    the other two squares commute by the naturality of the vertical morphisms,
    and the triangle commutes by definition.
    On the right side, the top area commutes by Lemma \ref{lem:mu-coherence}.
    The middle area commutes since $ c $ is a monoidal natural transformation,
    and the bottom area commutes because $ c $ is natural.

    \begin{figure}
      \[
        \begin{tikzcd}[column sep = 0em]
          &
          \bigotimes_i F^{\mode n_i}_{\mode n}((q_1 + q_2) \at s_i \at X_i)
          \arrow[dd, "\bigotimes_i F^{\mode n_i}_{\mode n} c"']
          \arrow[dddddr, "\bigotimes_{i} \mu",
          bend left]
          \\
          \bigotimes_i F^{\mode n_i}_{\mode n}((q_1 + q_2)s_i \at X_i)
          \arrow[ur, "\bigotimes_i F^{\mode n_i}_{\mode n} \delta"]
          \arrow[dd, "\bigotimes_i F^{\mode n_i}_{\mode n} c"']
          \\
          &
          \bigotimes_i F^{\mode n_i}_{\mode n}(q_1 \at s_i \at X_i \ox q_2 \at s_i \at X_i)
          \arrow[dd, "\bigotimes_i m^{-1}"']
          \\
          \bigotimes_i F^{\mode n_i}_{\mode n}(q_1s_i \at X_i \ox q_2s_i \at X_i)
          \arrow[ur, "\bigotimes_i F^{\mode n_i}_{\mode n}(\delta \ox \delta)"]
          \arrow[dd, "\bigotimes_i m^{-1}"']
          \\
          &
          \bigotimes_i
          F^{\mode n_i}_{\mode n}(q_1 \at s_i \at X_i) \ox
          F^{\mode n_i}_{\mode n}(q_2 \at s_i \at X_i)
          \arrow[dd, "\bigotimes_i \mu \ox \mu"']
          \\
          \bigotimes_i
          F^{\mode n_i}_{\mode n} (q_1s_i \at X_i) \ox
          F^{\mode n_i}_{\mode n}(q_2s_i \at X_i)
          \arrow[ur,
          "\bigotimes_i F^{\mode n_i}_{\mode n} \delta
          \ox F^{\mode n_i}_{\mode n} \delta
          "]
          \arrow[dd, "\sim"']
          \arrow[dr,
          "\bigotimes_i (\mu \circ F^{\mode n_i}_{\mode n} \delta)
          \ox (\mu \circ F^{\mode n_i}_{\mode n} \delta)
          "']
          &
          &
          \bigotimes_i (q_1 + q_2) \at F^{\mode n_i}_{\mode n}(s_i \at X_i)
          \arrow[dl,
          "\bigotimes_i c"',
          start anchor=west,
          end anchor={[xshift=-4em]north east}
          ]
          \arrow[ddd, "\tau"]
          \\
          &
          \bigotimes_i
          q_1 \at F^{\mode n_i}_{\mode n}(s_i \at X_i) \ox
          q_2 \at F^{\mode n_i}_{\mode n}(s_i \at X_i)
          \arrow[dd, "\sim"']
          \\
          \bigotimes_{j = 1, 2}
          \bigotimes_i
          F^{\mode n_i}_{\mode n} (q_js_i \at X_i)
          \arrow[dr, "\bigotimes_{j = 1, 2} \bigotimes_i \mu \circ F^{\mode n_i}_{\mode n}\delta"']
          \\
          &
          \bigotimes_{j = 1, 2} \bigotimes_i q_j \at F^{\mode n_i}_{\mode n} (s_i \at X_i)
          \arrow[dd, "\bigotimes_{j = 1, 2} \tau"]
          &
          (q_1 + q_2) \at \bigotimes_i F^{\mode n_i}(s_i \at X_i)
          \arrow[ddl, "c",
          start anchor={[xshift=2em]south west},
          end anchor={[xshift=-2em]north east}
          ]
          \arrow[d, "(q_1 + q_2) \at e"]
          \\
          &&
          (q_1 + q_2) \at T
          \arrow[d, "c"]
          \\
          &
          \bigotimes_{j = 1, 2} q_j \at \bigotimes_i F^{\mode n_i}_{\mode n}(s_i \at X_i)
          \arrow[r, "q_1 \at e \ox q_2 \at e"]
          &
          q_1 \at T \ox q_2 \at T
        \end{tikzcd}
      \]
      \caption{Diagram for case \textsc{Cont}.
        The arrows labeled ``$ \sim $'' are isomorphisms
        permuting the tensor product.
      }
      \label{afig:subst-comp-cont}
    \end{figure}
  \end{case}

  \begin{case}[\rname{dropI}]\rm
    Suppose we have $ \mathsf{n}  \le  \mathsf{m} $ and $ s \in R_{\mode m} $.
    Starting at the morphism
    \[
      \bigotimes_{ij}
      F^{\mode n_{ij}}_{\mode m}(r_i q_{ij} \at A_{ij})
      \iso
      \bigotimes_i
      F^{\mode m_i}_{\mode m}
      \bigotimes_j
      F^{\mode n_{ij}}_{\mode m_i}(r_i q_{ij} \at A_{ij})
      \xrightarrow{\bigotimes_i F^{\mode m_i}_{\mode m}(r_i \cdot e_i)}
      \bigotimes_i
      F^{\mode m_i}_{\mode m} (r_i \at A_i)
      \xrightarrow{t}
      B
    \]
    which we will denote by $ f $,
    and applying \rulename{grass}{term}{dropI} to it, we obtain the morphism
    \[
      F^{\mode m}_{\mode n}\bigotimes_{ij}
      F^{\mode n_{ij}}_{\mode m}(s r_i q_{ij} \at A_{ij})
      \xrightarrow{F^{\mode m}_{\mode n} (s \cdot f)}
      F^{\mode m}_{\mode n}(s \at B)
    \]
    We claim that this morphism is equal to
    \begin{mathpar}
      F^{\mode m}_{\mode n}\bigotimes_{ij}
      F^{\mode n_{ij}}_{\mode m}(s r_i q_{ij} \at A_{ij})
      \iso
      F^{\mode m}_{\mode n}
      \bigotimes_i
      F^{\mode m_i}_{\mode m}
      \bigotimes_j
      F^{\mode n_{ij}}_{\mode m_i}(sr_i q_{ij} \at A_{ij})
      \and
      \xrightarrow{
        F^{\mode m}_{\mode n}\bigotimes_i F^{\mode m_i}_{\mode m}(sr_i \cdot e_i)
      }
      F^{\mode m}_{\mode n} \bigotimes_i F^{\mode m_i}_{\mode m} (sr_i \at A_i)
      \xrightarrow{F^{\mode m}_{\mode n}(s \cdot t)}
      F^{\mode m}_{\mode n} (s \at B)
    \end{mathpar}
    which is immediate from Lemma \ref{alem:scalar-mult-comp}.
  \end{case}

  \begin{case}[\rname{dropE}]\rm
    Let $ \mathsf{l}  \le  \mathsf{n}  \le  \mathsf{m} $.
    By the inductive hypothesis we start with morphisms
    \begin{mathpar}
      \bigotimes_i \interp{  \grassnt{r_{\grassmv{i}}}   \sigma_{\grassmv{i}}   \mid  \mathsf{N}_{\grassmv{i}}  \odot  \Delta_{\grassmv{i}} }_{\mathsf{n}}
      \iso
      \bigotimes_{i} F^{\mode m_i}_{\mode n}
      \interp{  \grassnt{r_{\grassmv{i}}}   \sigma_{\grassmv{i}}   \mid  \mathsf{N}_{\grassmv{i}}  \odot  \Delta_{\grassmv{i}} }_{\mathsf{m}_{\grassmv{i}}}
      \xrightarrow{\bigotimes_i F^{\mode m_i}_{\mode n}(r_i \cdot e_i)}
      \bigotimes_i F^{\mode m_i}_{\mode n}(r_i \at A_i)
      \xrightarrow{t}
      F^{\mode m}_{\mode n} (q \at A)
    \end{mathpar}
    and
    \begin{mathpar}
      \bigotimes_i \interp{  \grassnt{r'_{\grassmv{i}}}   \sigma'_{\grassmv{i}}   \mid  \mathsf{N}'_{\grassmv{i}}  \odot  \Delta'_{\grassmv{i}} }_{\mathsf{l}}
      \ox F^{\mode m}_{\mode l} (q \at A)
      \iso
      \bigotimes_{i} F^{\mode m'_i}_{\mode l}
      \interp{  \grassnt{r'_{\grassmv{i}}}   \sigma'_{\grassmv{i}}   \mid  \mathsf{N}'_{\grassmv{i}}  \odot  \Delta'_{\grassmv{i}} }_{\mathsf{m}'_{\grassmv{i}}}
      \ox F^{\mode m}_{\mode l} (q \at A)
      \and
      \xrightarrow{\bigotimes_i F^{\mode m'_i}_{\mode l}(r'_i \cdot e'_i) \ox \id}
      \bigotimes_i F^{\mode m'_i}_{\mode l}(r'_i \at A'_i)
      \ox F^{\mode m}_{\mode l} (q \at A)
      \xrightarrow{t'}
      B
    \end{mathpar}
    Applying the rule \rulename{grass}{term}{dropE} to these two morphisms produces
    the morphism
    \begin{align*}
      &\
        \bigotimes_{i} F^{\mode m'_i}_{\mode l}
        \interp{  \grassnt{r'_{\grassmv{i}}}   \sigma'_{\grassmv{i}}   \mid  \mathsf{N}'_{\grassmv{i}}  \odot  \Delta'_{\grassmv{i}} }_{\mathsf{m}'_{\grassmv{i}}}
        \ox
        \bigotimes_{i}
        F^{\mode m_i}_{\mode l}
        \interp{  \grassnt{r_{\grassmv{i}}}   \sigma_{\grassmv{i}}   \mid  \mathsf{N}_{\grassmv{i}}  \odot  \Delta_{\grassmv{i}} }_{\mathsf{m}_{\grassmv{i}}}
      \\
      \iso
      &\
        \bigotimes_{i} F^{\mode m'_i}_{\mode l}
        \interp{  \grassnt{r'_{\grassmv{i}}}   \sigma'_{\grassmv{i}}   \mid  \mathsf{N}'_{\grassmv{i}}  \odot  \Delta'_{\grassmv{i}} }_{\mathsf{m}'_{\grassmv{i}}}
        \ox
        F^{\mode n}_{\mode l}
        \bigotimes_{i}
        F^{\mode m_i}_{\mode n}
        \interp{  \grassnt{r_{\grassmv{i}}}   \sigma_{\grassmv{i}}   \mid  \mathsf{N}_{\grassmv{i}}  \odot  \Delta_{\grassmv{i}} }_{\mathsf{m}_{\grassmv{i}}}
      \\
      \xrightarrow{\id \ox
      F^{\mode n}_{\mode l}(
      t \circ
      \bigotimes_i
      F^{\mode m_i}_{\mode n}
      (r_i \cdot e_i)
      )}
      &\
        \bigotimes_{i} F^{\mode m'_i}_{\mode l}
        \interp{  \grassnt{r'_{\grassmv{i}}}   \sigma'_{\grassmv{i}}   \mid  \mathsf{N}'_{\grassmv{i}}  \odot  \Delta'_{\grassmv{i}} }_{\mathsf{m}'_{\grassmv{i}}}
        \ox F^{\mode n}_{\mode l}F^{\mode m}_{\mode n}(q \at A)
      \\
      \xrightarrow{\bigotimes_i F^{\mode m'_i}_{\mode l}(r'_i \cdot e'_i) \ox \id}
      &\
        \bigotimes_i F^{\mode m'_i}_{\mode l}(r'_i \at A'_i)
        \ox F^{\mode n}_{\mode l} F^{\mode m}_{\mode n} (q \at A)
      \\
      \iso
      &\
        \bigotimes_i F^{\mode m'_i}_{\mode l}(r'_i \at A'_i)
        \ox F^{\mode m}_{\mode l} (q \at A)
        \xrightarrow{t'} B
    \end{align*}
    and we claim that this morphism is equal to
    \begin{align*}
      &\
        \bigotimes_{i} F^{\mode m'_i}_{\mode l}
        \interp{  \grassnt{r'_{\grassmv{i}}}   \sigma'_{\grassmv{i}}   \mid  \mathsf{N}'_{\grassmv{i}}  \odot  \Delta'_{\grassmv{i}} }_{\mathsf{m}'_{\grassmv{i}}}
        \ox
        \bigotimes_{i}
        F^{\mode m_i}_{\mode l}
        \interp{  \grassnt{r_{\grassmv{i}}}   \sigma_{\grassmv{i}}   \mid  \mathsf{N}_{\grassmv{i}}  \odot  \Delta_{\grassmv{i}} }_{\mathsf{m}_{\grassmv{i}}}
      \\
      \xrightarrow{
      \bigotimes_i F^{\mode m'_i}_{\mode l}(r'_i \cdot e'_i)
      \ox
      \bigotimes_i F^{\mode m_i}_{\mode l}(r_i \cdot e_i)
      }
      &\
        \bigotimes_i F^{\mode m'_i}_{\mode l}(r'_i \at A'_i)
        \ox
        \bigotimes_i F^{\mode m_i}_{\mode l}(r_i \at A_i)
      \\
      \iso
      &\
        \bigotimes_i F^{\mode m'_i}_{\mode l}(r'_i \at A'_i)
        \ox
        F^{\mode n}_{\mode l}
        \bigotimes_i F^{\mode m_i}_{\mode n}(r_i \at A_i)
      \\
      \xrightarrow{\id \ox F^{\mode n}_{\mode l}(t)}
      &\
        \bigotimes_i F^{\mode m'_i}_{\mode l}(r'_i \at A'_i)
        \ox
        F^{\mode n}_{\mode l} F^{\mode m}_{\mode n} (q \at A)
      \\
      \iso
      &\
        \bigotimes_i F^{\mode m'_i}_{\mode l}(r'_i \at A'_i)
        \ox F^{\mode m}_{\mode l} (q \at A)
        \xrightarrow{t'} B
    \end{align*}
    These morphisms are readily verified to be equal.
  \end{case}

  \begin{case}[\rname{raiseI}]\rm
    Let $ \mathsf{m}  \le  \mathsf{n} $.
    By the inductive hypothesis, we begin with a morphism
    \[
      \bigotimes_i F^{\mode m_i}_{\mode m}
      \interp{   \grassnt{r_{\grassmv{i}}}   \sigma_{\grassmv{i}}   \mid  \mathsf{N}_{\grassmv{i}}  \odot  \Delta_{\grassmv{i}} }_{\mathsf{m}_{\grassmv{i}}}
      \xrightarrow{\bigotimes_i F^{\mode m_i}_{\mode m}(r_i \cdot e_i)}
      \bigotimes_i F^{\mode m_i}_{\mode m}(r_i \at A_i)
      \xrightarrow{t}
      B
    \]
    applying the rule \rulename{grass}{term}{raiseI} to it yields the morphism
    \begin{align*}
      & \
        \bigotimes_{i}
        F^{\mode m_i}_{\mode n} \interp{   \grassnt{r_{\grassmv{i}}}   \sigma_{\grassmv{i}}   \mid  \mathsf{N}_{\grassmv{i}}  \odot  \Delta_{\grassmv{i}}  }_{\mode m_i}
      \\
      \xrightarrow{\text{unit}}
      & \
        G^{\mode n}_{\mode m} F^{\mode n}_{\mode m}
        \bigotimes_{i}
        F^{\mode m_i}_{\mode n} \interp{   \grassnt{r_{\grassmv{i}}}   \sigma_{\grassmv{i}}   \mid  \mathsf{N}_{\grassmv{i}}  \odot  \Delta_{\grassmv{i}}  }_{\mode m_i}
      \\
      \iso
      & \
        G^{\mode n}_{\mode m}
        \bigotimes_{i}
        F^{\mode n}_{\mode m}
        F^{\mode m_i}_{\mode n} \interp{   \grassnt{r_{\grassmv{i}}}   \sigma_{\grassmv{i}}   \mid  \mathsf{N}_{\grassmv{i}}  \odot  \Delta_{\grassmv{i}}  }_{\mode m_i}
      \\
      \iso
      & \
        G^{\mode n}_{\mode m}
        \bigotimes_{i}
        F^{\mode m_i}_{\mode m}
        \interp{   \grassnt{r_{\grassmv{i}}}   \sigma_{\grassmv{i}}   \mid  \mathsf{N}_{\grassmv{i}}  \odot  \Delta_{\grassmv{i}}  }_{\mode m_i}
      \\
      \xrightarrow{
      G^{\mode m}_{\mode n} \bigotimes_i F^{\mode m_i}_{\mode m}(r_i \cdot e_i)
      }
      & \
        G^{\mode n}_{\mode m}\bigotimes_i F^{\mode m_i}_{\mode m}(r_i \at A_i)
        \xrightarrow{G^{\mode m}_{\mode n} t}
        G^{\mode m}_{\mode n} B
    \end{align*}
    We claim that this morphism is equal to
    \begin{align*}
      & \
        \bigotimes_i
        F^{\mode m_i}_{\mode n} \interp{   \grassnt{r_{\grassmv{i}}}   \sigma_{\grassmv{i}}   \mid  \mathsf{N}_{\grassmv{i}}  \odot  \Delta_{\grassmv{i}}  }_{\mode m_i}
      \\
      \xrightarrow{\bigotimes_i F^{\mode m_i}_{\mode n}(r_i \cdot e_i)}
      &\
        \bigotimes_i F^{\mode m_i}_{\mode n} (r_i \at A_i)
      \\
      \xrightarrow{\text{unit}}
      &\
        G^{\mode n}_{\mode m} F^{\mode m}_{\mode n}
        \bigotimes_i F^{\mode m_i}_{\mode n} (r_i \at A_i)
      \\
      \iso
      &\
        G^{\mode n}_{\mode m}
        \bigotimes_i F^{\mode m}_{\mode n} F^{\mode m_i}_{\mode n} (r_i \at A_i)
      \\
      \iso
      &\
        G^{\mode n}_{\mode m}
        \bigotimes_i F^{\mode m_i}_{\mode m} (r_i \at A_i)
        \xrightarrow{G^{\mode n}_{\mode m} t} G^{\mode n}_{\mode m} B
    \end{align*}
    The following diagram commutes, since the horizontal morphisms are natural
    transformations, concluding this case.
    \[
      \begin{tikzcd}
        \bigotimes_i
        F^{\mode m_i}_{\mode n} \interp{   \grassnt{r_{\grassmv{i}}}   \sigma_{\grassmv{i}}   \mid  \mathsf{N}_{\grassmv{i}}  \odot  \Delta_{\grassmv{i}}  }_{\mode m_i}
        \arrow[r, "\text{unit}"]
        \arrow[d, "\bigotimes_i F^{\mode m_i}_{\mode n}(r_i \cdot e_i)"]
        &
        G^{\mode n}_{\mode m} F^{\mode n}_{\mode m}
        \bigotimes_{i}
        F^{\mode m_i}_{\mode n} \interp{   \grassnt{r_{\grassmv{i}}}   \sigma_{\grassmv{i}}   \mid  \mathsf{N}_{\grassmv{i}}  \odot  \Delta_{\grassmv{i}}  }_{\mode m_i}
        \arrow[r, "\sim"]
        &
        G^{\mode n}_{\mode m}
        \bigotimes_{i}
        F^{\mode m_i}_{\mode m} \interp{   \grassnt{r_{\grassmv{i}}}   \sigma_{\grassmv{i}}   \mid  \mathsf{N}_{\grassmv{i}}  \odot  \Delta_{\grassmv{i}}  }_{\mode m_i}
        \arrow[d, "G^{\mode n}_{\mode m}
        \bigotimes_i F^{\mode m_i}_{\mode m}(r_i \cdot e_i)"]
        \\
        \bigotimes_i F^{\mode m_i}_{\mode n} (r_i \at A_i)
        \arrow[r, "\text{unit}"']
        &
        G^{\mode n}_{\mode m} F^{\mode m}_{\mode n}
        \bigotimes_i F^{\mode m_i}_{\mode n} (r_i \at A_i)
        \arrow[r, "\sim"']
        &
        G^{\mode n}_{\mode m}
        \bigotimes_i F^{\mode m_i}_{\mode m} (r_i \at A_i)
      \end{tikzcd}
    \]
  \end{case}

  \begin{case}[\rname{raiseE}]\rm
    Let $ \mathsf{n}  \le  \mathsf{m} $.
    By the inductive hypothesis we begin with a morphism
    \[
      \bigotimes_i
      \interp{   \grassnt{r_{\grassmv{i}}}   \sigma_{\grassmv{i}}   \mid  \mathsf{N}_{\grassmv{i}}  \odot  \Delta_{\grassmv{i}} }_{\mode m}
      \iso
      \bigotimes_i F^{\mode m_i}_{\mode m}
      \interp{   \grassnt{r_{\grassmv{i}}}   \sigma_{\grassmv{i}}   \mid  \mathsf{N}_{\grassmv{i}}  \odot  \Delta_{\grassmv{i}} }_{\mode m_i}
      \xrightarrow{\bigotimes_i F^{\mode m_i}_{\mode m} (r_i \cdot e_i)}
      \bigotimes_i
      F^{\mode m_i}_{\mode m} (r_i \at A_i)
      \xrightarrow{t}
      G^{\mode m}_{\mode n} B
    \]
    Applying the rule \rulename{grass}{term}{dropE} to this morphism yields
    \begin{align*}
      \bigotimes_i
      \interp{   \grassnt{r_{\grassmv{i}}}   \sigma_{\grassmv{i}}   \mid  \mathsf{N}_{\grassmv{i}}  \odot  \Delta_{\grassmv{i}} }_{\mode n}
      &
        \iso
        F^{\mode m}_{\mode n}
        \bigotimes_i
        \interp{   \grassnt{r_{\grassmv{i}}}   \sigma_{\grassmv{i}}   \mid  \mathsf{N}_{\grassmv{i}}  \odot  \Delta_{\grassmv{i}} }_{\mode m}
        \iso
        F^{\mode m}_{\mode n}
        \bigotimes_i
        F^{\mode m_i}_{\mode m}
        \interp{   \grassnt{r_{\grassmv{i}}}   \sigma_{\grassmv{i}}   \mid  \mathsf{N}_{\grassmv{i}}  \odot  \Delta_{\grassmv{i}} }_{\mode m_i}
      \\
      &
        \xrightarrow{
        F^{\mode m}_{\mode n} \bigotimes_i F^{\mode m_i}_{\mode m} (r_i \cdot e_i)
        }
        F^{\mode m}_{\mode n}
        \bigotimes_i
        F^{\mode m_i}_{\mode m}
        (r_i \at A_i)
        \xrightarrow{F^{\mode m}_{\mode n}t}
        F^{\mode m}_{\mode n}G^{\mode m}_{\mode n} B
        \xrightarrow{\text{counit}} B
    \end{align*}
    We claim that this is equal to the morphism
    \begin{align*}
      \bigotimes_i
      \interp{   \grassnt{r_{\grassmv{i}}}   \sigma_{\grassmv{i}}   \mid  \mathsf{N}_{\grassmv{i}}  \odot  \Delta_{\grassmv{i}} }_{\mode n}
      &
        \iso
        \bigotimes_i
        F^{\mode m_i}_{\mode n}
        \interp{   \grassnt{r_{\grassmv{i}}}   \sigma_{\grassmv{i}}   \mid  \mathsf{N}_{\grassmv{i}}  \odot  \Delta_{\grassmv{i}} }_{\mode m_i}
      \\
      &
        \xrightarrow{\bigotimes_i F^{\mode m_i}_{\mode n}(r_i \cdot e_i)}
        \bigotimes_i F^{\mode m_i}_{\mode n} (r_i \at A_i)
        \iso
        F^{\mode m}_{\mode n}
        \bigotimes_i
        F^{\mode m_i}_{\mode m} (r_i \at A_i)
        \xrightarrow{F^{\mode m}_{\mode n}(t)}
        F^{\mode m}_{\mode n} G^{\mode m}_{\mode n} B
        \xrightarrow{\text{counit}} B
    \end{align*}
    So it is sufficient to see that the following diagram commutes
    \[
      \begin{tikzcd}
        \bigotimes_i
        \interp{   \grassnt{r_{\grassmv{i}}}   \sigma_{\grassmv{i}}   \mid  \mathsf{N}_{\grassmv{i}}  \odot  \Delta_{\grassmv{i}} }_{\mode n}
        \arrow[r, "\sim"]
        \arrow[d, "\sim"]
        &
        \bigotimes_i
        F^{\mode m}_{\mode n}
        \interp{   \grassnt{r_{\grassmv{i}}}   \sigma_{\grassmv{i}}   \mid  \mathsf{N}_{\grassmv{i}}  \odot  \Delta_{\grassmv{i}} }_{\mode m}
        \arrow[r, "\sim"]
        \arrow[d, "\sim"]
        &
        F^{\mode m}_{\mode n}
        \bigotimes_i
        \interp{   \grassnt{r_{\grassmv{i}}}   \sigma_{\grassmv{i}}   \mid  \mathsf{N}_{\grassmv{i}}  \odot  \Delta_{\grassmv{i}} }_{\mode m}
        \arrow[d, "\sim"]
        \\
        \bigotimes_i
        F^{\mode m_i}_{\mode n}
        \interp{   \grassnt{r_{\grassmv{i}}}   \sigma_{\grassmv{i}}   \mid  \mathsf{N}_{\grassmv{i}}  \odot  \Delta_{\grassmv{i}} }_{\mode m_i}
        \arrow[r, "\sim"]
        \arrow[d, "\bigotimes_i F^{\mode m_i}_{\mode n}(r_i \cdot e_i)"]
        &
        \bigotimes_i
        F^{\mode m}_{\mode n}
        F^{\mode m_i}_{\mode m}
        \interp{   \grassnt{r_{\grassmv{i}}}   \sigma_{\grassmv{i}}   \mid  \mathsf{N}_{\grassmv{i}}  \odot  \Delta_{\grassmv{i}} }_{\mode m_i}
        \arrow[r, "\sim"]
        &
        F^{\mode m}_{\mode n}
        \bigotimes_i
        F^{\mode m_i}_{\mode m}
        \interp{   \grassnt{r_{\grassmv{i}}}   \sigma_{\grassmv{i}}   \mid  \mathsf{N}_{\grassmv{i}}  \odot  \Delta_{\grassmv{i}} }_{\mode m_i}
        \arrow[d, "F^{\mode m}_{\mode n}\bigotimes_i F^{\mode m_i}_{\mode m}(r_i \cdot e_i)"]
        \\
        \bigotimes_i F^{\mode m_i}_{\mode n} (r_i \at A_i)
        \arrow[r, "\sim"]
        &
        \bigotimes_i F^{\mode m}_{\mode n} F^{\mode m_i}_{\mode m} (r_i \at A_i)
        \arrow[r, "\sim"]
        &
        F^{\mode m}_{\mode n} \bigotimes_i F^{\mode m_i}_{\mode m} (r_i \at A_i)
      \end{tikzcd}
    \]
    The bottom rectangle commutes by the naturality of the horizontal morphisms,
    while the top rectangles commute since the system of natural isomorphisms
    $ F^{\mode l_2}_{\mode l_3} \circ F^{\mode l_1}_{\mode l_2}
    \iso F^{\mode l_1}_{\mode l_3} $ is coherent.
  \end{case}

  The remaining cases are straightforward using \ref{alem:scalar-mult-comp}.
  We demonstrate this using the case for function application as an example.
  \begin{case}[\rname{arrowE}]\rm
    We begin with morphisms
    \begin{mathpar}
      t_1 \from
      \bigotimes_i \interp{   \grassnt{r_{\grassmv{i}}}   \sigma_{\grassmv{i}}   \mid  \mathsf{N}_{\grassmv{i}}  \odot  \Delta_{\grassmv{i}} }_{\mode m}
      \iso
      \bigotimes_i F^{\mode m_i}_{\mode m}
      \interp{   \grassnt{r_{\grassmv{i}}}   \sigma_{\grassmv{i}}   \mid  \mathsf{N}_{\grassmv{i}}  \odot  \Delta_{\grassmv{i}} }_{\mode m_i}
      \xrightarrow{\bigotimes_i F^{\mode m_i}_{\mode m} (r_i \cdot e_i)}
      \bigotimes_i F^{\mode m_i}_{\mode m}(r_i \at A_i)
      \xrightarrow{t}
      B
      \and
      t_2 \from
      \bigotimes_i
      \interp{   \grassnt{r'_{\grassmv{i}}}   \sigma'_{\grassmv{i}}   \mid  \mathsf{N}'_{\grassmv{i}}  \odot  \Delta'_{\grassmv{i}} }_{\mode n}
      \iso
      \bigotimes_i F^{\mode m'_i}_{\mode n}
      \interp{   \grassnt{r'_{\grassmv{i}}}   \sigma'_{\grassmv{i}}   \mid  \mathsf{N}'_{\grassmv{i}}  \odot  \Delta'_{\grassmv{i}} }_{\mode m'_i}
      \xrightarrow{\bigotimes_i F^{\mode m'_i}_{\mode n} (r'_i \cdot e'_i)}
      \bigotimes_i F^{\mode m'_i}_{\mode n}(r'_i \at A'_i)
      \xrightarrow{t'}
      F^{\mode m}_{\mode n} (q \at B) \lto B'
    \end{mathpar}
    applying \rulename{grass}{term}{arrowE} to them yields the morphism
    \begin{align*}
      & \
        \bigotimes_i
        \interp{   \grassnt{r'_{\grassmv{i}}}   \sigma'_{\grassmv{i}}   \mid  \mathsf{N}'_{\grassmv{i}}  \odot  \Delta'_{\grassmv{i}} }_{\mode n}
        \ox
        \bigotimes_i \interp{    \grassnt{q}   \grassnt{r_{\grassmv{i}}}    \sigma_{\grassmv{i}}   \mid  \mathsf{N}_{\grassmv{i}}  \odot  \Delta_{\grassmv{i}} }_{\mode n}
      \\
      \iso
      & \
        \bigotimes_i F^{\mode m'_i}_{\mode n}
        \interp{   \grassnt{r'_{\grassmv{i}}}   \sigma'_{\grassmv{i}}   \mid  \mathsf{N}'_{\grassmv{i}}  \odot  \Delta'_{\grassmv{i}} }_{\mode m'_i}
        \ox
        F^{\mode m}_{\mode n} \bigotimes_i F^{\mode m_i}_{\mode m}
        \interp{    \grassnt{q}   \grassnt{r_{\grassmv{i}}}    \sigma_{\grassmv{i}}   \mid  \mathsf{N}_{\grassmv{i}}  \odot  \Delta_{\grassmv{i}} }_{\mode m_i}
      \\
      \xrightarrow{t_2 \ox F^{\mode m}_{\mode n}(q \cdot t_1)}
      & \
        ((q \at B) \lto B') \ox (q \at B)
        \xrightarrow{\textit{ev}}
        B'
    \end{align*}
    We claim that this morphism is equal to
    \begin{align*}
      & \
        \bigotimes_i
        \interp{   \grassnt{r'_{\grassmv{i}}}   \sigma'_{\grassmv{i}}   \mid  \mathsf{N}'_{\grassmv{i}}  \odot  \Delta'_{\grassmv{i}} }_{\mode m}
        \ox
        \bigotimes_i \interp{    \grassnt{q}   \grassnt{r_{\grassmv{i}}}    \sigma_{\grassmv{i}}   \mid  \mathsf{N}_{\grassmv{i}}  \odot  \Delta_{\grassmv{i}} }_{\mode m}
      \\
      \iso
      & \
        \bigotimes_i F^{\mode m'_i}_{\mode m}
        \interp{   \grassnt{r'_{\grassmv{i}}}   \sigma'_{\grassmv{i}}   \mid  \mathsf{N}'_{\grassmv{i}}  \odot  \Delta'_{\grassmv{i}} }_{\mode m'_i}
        \ox
        \bigotimes_i F^{\mode m_i}_{\mode m}
        \interp{    \grassnt{q}   \grassnt{r_{\grassmv{i}}}    \sigma_{\grassmv{i}}   \mid  \mathsf{N}_{\grassmv{i}}  \odot  \Delta_{\grassmv{i}} }_{\mode m_i}
      \\
      \xrightarrow{
      (\bigotimes_i F^{\mode m'_i}_{\mode m} (r'_i \cdot e'_i))
      \ox
      (\bigotimes_i F^{\mode m_i}_{\mode m} ((qr_i) \cdot e_i))
      }
      & \
        \bigotimes_i F^{\mode m'_i}_{\mode m}(r'_i \at A'_i)
        \ox
        \bigotimes_i F^{\mode m_i}_{\mode m}((qr_i) \at A_i)
      \\
      \xrightarrow{t' \ox (q \cdot t)}
      &
        ((q \at B) \lto B') \ox (q \at B)
        \xrightarrow{\textit{ev}} B'
    \end{align*}
    But this follows immediately from Lemma \ref{alem:scalar-mult-comp}.
  \end{case}
\end{proof}

\newpage
\subsection{Semantic soundness of beta-conversion}
\label{asec:beta-semantic-soundness}

\begin{theorem}
  $ \beta $-conversion is sound with respect to the categorical semantics:
  If $  \rho  \mid  \mathsf{M}  \odot  \Gamma  \vdash_{ \mathsf{m} }  \grassnt{t_{{\mathrm{1}}}}  \equiv_\beta  \grassnt{t_{{\mathrm{2}}}}  \colon  \grassnt{A}  $,
  then the interpretations of
  $  \rho  \mid  \mathsf{M}   \odot   \Gamma  \vdash_{ \mathsf{m} }  \grassnt{t_{{\mathrm{1}}}}  \colon  \grassnt{A}  $ and
  $  \rho  \mid  \mathsf{M}   \odot   \Gamma  \vdash_{ \mathsf{m} }  \grassnt{t_{{\mathrm{2}}}}  \colon  \grassnt{A}  $ are equal in any categorical model.
\end{theorem}

\begin{proof}
  \setcounter{case}{0}
  We proceed by induction,
  inspecting the proof of Theorem \ref{athm:beta-syntactic-soundness},
  and verifying that the derivations provided there have the same interpretations.
  We will only consider the cases where the $ \beta $-conversion
  happens as a result of an introduction rule being used against an elimination rule,
  since the other cases are straightforward.

  \begin{case}[Unit]\rm
    In this case we have a morphism
    $ \interp{  \rho  \mid  \mathsf{M}  \odot  \Gamma }_{\mode m} \xrightarrow{e} I_{\mode m} $
    and we must show that this morphism is equal to
    \[
      \interp{  \rho  \mid  \mathsf{M}  \odot  \Gamma }_{\mode m}
      \iso
      \interp{  \rho  \mid  \mathsf{M}  \odot  \Gamma }_{\mode m} \ox \interp{  \emptyset  \mid  \emptyset  \odot  \emptyset }_{\mode m}
      \xrightarrow{e \ox (\iota \circ (q \cdot \id))}
      I_{\mode m} \ox I_{\mode m}
      \iso
      I_{\mode m}.
    \]
    But $ q \cdot \id = \iota^{-1} : I_{\mode m} \to q \at I_{\mode m} $,
    from which the claim follows.
  \end{case}

  \begin{case}[Pair]\rm
    In this case we are given morphisms
    $ \interp{  \sigma_{\grassmv{i}}  \mid  \mathsf{N}_{\grassmv{i}}  \odot  \Delta_{\grassmv{i}} }_{\mode m} \xrightarrow{e_i} A_i $
    (for $ i \in \{1, 2\} $)
    and
    \[
      \interp{  \rho  \mid  \mathsf{M}  \odot  \Gamma }_{\mode n}
      \ox F^{\mode m}_{\mode n} (q \at A_1)
      \ox F^{\mode m}_{\mode n} (q \at A_2)
      \xrightarrow{t} B
    \]
    The interpretation of $  \mathbin{\mathsf{let} } _{@  \grassnt{q} }  \grasssym{(}  \grassmv{x_{{\mathrm{1}}}}  \grasssym{,}  \grassmv{x_{{\mathrm{2}}}}  \grasssym{)}  =  \grasssym{(}  \grassnt{e_{{\mathrm{1}}}}  \grasssym{,}  \grassnt{e_{{\mathrm{2}}}}  \grasssym{)}   \mathbin{\mathsf{in} }   \grassnt{t}  $ is
    \begin{align*}
      &\
        \interp{  \rho  \mid  \mathsf{M}  \odot  \Gamma }_{\mode n}
        \ox F^{\mode m}_{\mode n}
        \interp{   \grassnt{q}   \grasssym{(}  \sigma_{{\mathrm{1}}}  \grasssym{,}  \sigma_{{\mathrm{2}}}  \grasssym{)}   \mid  \mathsf{N}_{{\mathrm{1}}}  \grasssym{,}  \mathsf{N}_{{\mathrm{2}}}  \odot  \Delta_{{\mathrm{1}}}  \grasssym{,}  \Delta_{{\mathrm{2}}} }_{\mode m}
      \\
      \to
      &\
        \interp{  \rho  \mid  \mathsf{M}  \odot  \Gamma }_{\mode n}
        \ox F^{\mode m}_{\mode n}
        q \at \interp{  \grasssym{(}  \sigma_{{\mathrm{1}}}  \grasssym{,}  \sigma_{{\mathrm{2}}}  \grasssym{)}  \mid  \mathsf{N}_{{\mathrm{1}}}  \grasssym{,}  \mathsf{N}_{{\mathrm{2}}}  \odot  \Delta_{{\mathrm{1}}}  \grasssym{,}  \Delta_{{\mathrm{2}}} }_{\mode m}
      \\
      \iso
      &\
        \interp{  \rho  \mid  \mathsf{M}  \odot  \Gamma }_{\mode n}
        \ox F^{\mode m}_{\mode n}
        q \at (
        \interp{  \sigma_{{\mathrm{1}}}  \mid  \mathsf{N}_{{\mathrm{1}}}  \odot  \Delta_{{\mathrm{1}}} }_{\mode m}
        \ox
        \interp{  \sigma_{{\mathrm{2}}}  \mid  \mathsf{N}_{{\mathrm{2}}}  \odot  \Delta_{{\mathrm{2}}} }_{\mode m}
        )
      \\
      \xrightarrow{\id \ox F^{\mode m}_{\mode n} q \at (e_1 \ox e_2)}
      & \
        \interp{  \rho  \mid  \mathsf{M}  \odot  \Gamma }_{\mode n}
        \ox F^{\mode m}_{\mode n}
        q \at (A_1 \ox A_2)
      \\
      \iso
      &\
        \interp{  \rho  \mid  \mathsf{M}  \odot  \Gamma }_{\mode n}
        \ox F^{\mode m}_{\mode n} (q \at A_1)
        \ox F^{\mode m}_{\mode n} (q \at A_2)
      \\
      \xrightarrow{t} &\ B
    \end{align*}
    where the unlabeled morphism is constructed from $ \mu $ and $ \tau $.
    On the other hand the interpretation of $ [e_1/x_1, e_2/x_2] t $
    is given by Lemma \ref{alem:subst-comp} as:
    \begin{align*}
      &\
        \interp{  \rho  \mid  \mathsf{M}  \odot  \Gamma }_{\mode n}
        \ox F^{\mode m}_{\mode n}
        \interp{   \grassnt{q}   \grasssym{(}  \sigma_{{\mathrm{1}}}  \grasssym{,}  \sigma_{{\mathrm{2}}}  \grasssym{)}   \mid  \mathsf{N}_{{\mathrm{1}}}  \grasssym{,}  \mathsf{N}_{{\mathrm{2}}}  \odot  \Delta_{{\mathrm{1}}}  \grasssym{,}  \Delta_{{\mathrm{2}}} }_{\mode m}
      \\
      \iso
      &\
        \interp{  \rho  \mid  \mathsf{M}  \odot  \Gamma }_{\mode n}
        \ox F^{\mode m}_{\mode n}
        \interp{   \grassnt{q}   \sigma_{{\mathrm{1}}}   \mid  \mathsf{N}_{{\mathrm{1}}}  \odot  \Delta_{{\mathrm{1}}} }_{\mode m}
        \ox F^{\mode m}_{\mode n}
        \interp{   \grassnt{q}   \sigma_{{\mathrm{2}}}   \mid  \mathsf{N}_{{\mathrm{2}}}  \odot  \Delta_{{\mathrm{2}}} }_{\mode m}
      \\
      \to
      &\
        \interp{  \rho  \mid  \mathsf{M}  \odot  \Gamma }_{\mode n}
        \ox F^{\mode m}_{\mode n} q\at
        \interp{  \sigma_{{\mathrm{1}}}  \mid  \mathsf{N}_{{\mathrm{1}}}  \odot  \Delta_{{\mathrm{1}}} }_{\mode m}
        \ox F^{\mode m}_{\mode n} q\at
        \interp{  \sigma_{{\mathrm{2}}}  \mid  \mathsf{N}_{{\mathrm{2}}}  \odot  \Delta_{{\mathrm{2}}} }_{\mode m}
      \\
      \xrightarrow{\id
      \ox F^{\mode m}_{\mode n}(q \at e_1)
      \ox F^{\mode m}_{\mode n}(q \at e_2)}
      &\
        \interp{  \rho  \mid  \mathsf{M}  \odot  \Gamma }_{\mode n}
        \ox F^{\mode m}_{\mode n} (q \at A_1)
        \ox F^{\mode m}_{\mode n} (q \at A_2)
      \\
      \xrightarrow{t} &\ B
    \end{align*}
    where the unlabeled arrow is constructed from $ \mu $ and $ \tau $.
    These morphisms are easily verified to be equal.
  \end{case}
  \begin{case}[Function]\rm
    We are given morphisms
    \begin{mathpar}
      \interp{  \rho  \mid  \mathsf{M}  \odot  \Gamma }_{\mode n} \ox F^{\mode m}_{\mode n}(q \at A)
      \xrightarrow{t}
      B
      \and
      \interp{  \sigma  \mid  \mathsf{N}  \odot  \Delta }_{\mode m} \xrightarrow{e} A
    \end{mathpar}
    By the formula for transposing morphism along adjunctions,
    the intepretation $ \grasssym{(}  \lambda  \grassmv{x}  \grasssym{.}  \grassnt{t}  \grasssym{)} $ is
    \[
      \interp{  \rho  \mid  \mathsf{M}  \odot  \Gamma }_{\mode n}
      \xrightarrow{\text{unit}}
      F^{\mode m}_{\mode n}(q \at A) \lto
      (\interp{  \rho  \mid  \mathsf{M}  \odot  \Gamma }_{\mode n} \ox F^{\mode m}_{\mode n}(q \at A))
      \xrightarrow{F^{\mode m}_{\mode n}(q \at A) \lto t}
      F^{\mode m}_{\mode n}(q \at A) \lto B
    \]
    and thus the interpretation $ \grasssym{(}  \lambda  \grassmv{x}  \grasssym{.}  \grassnt{t}  \grasssym{)} \, \grassnt{e} $ is
    \[
      \interp{  \rho  \mid  \mathsf{M}  \odot  \Gamma }_{\mode n}
      \ox
      F^{\mode m}_{\mode n}\interp{   \grassnt{q}   \sigma   \mid  \mathsf{N}  \odot  \Delta }_{\mode m}
      \xrightarrow{((F^{\mode m}_{\mode n}(q \at A) \lto t) \circ \text{unit})
        \ox F^{\mode m}_{\mode n} (q \cdot e)}
      (F^{\mode m}_{\mode n}(q \at A) \lto B) \ox F^{\mode m}_{\mode n}(q \at A)
      \xrightarrow{\textit{ev}} B
    \]
    We claim that this morphism is equal to
    \[
      \interp{  \rho  \mid  \mathsf{M}  \odot  \Gamma }_{\mode n}
      \ox
      F^{\mode m}_{\mode n}\interp{   \grassnt{q}   \sigma   \mid  \mathsf{N}  \odot  \Delta }_{\mode m}
      \xrightarrow{\id \ox F^{\mode m}_{\mode n}(q \cdot e)}
      \interp{  \rho  \mid  \mathsf{M}  \odot  \Gamma }_{\mode n}
      \ox
      F^{\mode m}_{\mode n}(q \at A)
      \xrightarrow{t}
      B
    \]
    To see that these morphisms are equal, consider the following diagram:
    \[
      \begin{tikzcd}[column sep = -4em]
        \interp{  \rho  \mid  \mathsf{M}  \odot  \Gamma }_{\mode n}
        \ox
        F^{\mode m}_{\mode n}\interp{   \grassnt{q}   \sigma   \mid  \mathsf{N}  \odot  \Delta }_{\mode m}
        \arrow[dd, "\text{unit} \ox \id"']
        \arrow[rd, "\id \ox F^{\mode m}_{\mode n}(q \cdot e)"]
        \\
        &
        \interp{  \rho  \mid  \mathsf{M}  \odot  \Gamma }_{\mode n}
        \ox
        F^{\mode m}_{\mode n}(q \at A)
        \arrow[dd, "\text{unit} \ox \id"]
        \arrow[rddd, "\id", bend left=35]
        \\
        \mlnode{
          $ F^{\mode m}_{\mode n}(q \at A) \lto
          (\interp{  \rho  \mid  \mathsf{M}  \odot  \Gamma }_{\mode n} \ox F^{\mode m}_{\mode n}(q \at A)) $
          \\
          $ \ox F^{\mode m}_{\mode n}\interp{   \grassnt{q}   \sigma   \mid  \mathsf{N}  \odot  \Delta }_{\mode m} $
        }
        \arrow[rd, "\id \ox F^{\mode m}_{\mode n}(q \cdot e)"]
        \arrow[dd, "(F^{\mode m}_{\mode n}(q \at A) \lto t) \ox \id"']
        \\
        &
        \mlnode{
          $ F^{\mode m}_{\mode n}(q \at A) \lto
          (\interp{  \rho  \mid  \mathsf{M}  \odot  \Gamma }_{\mode n} \ox F^{\mode m}_{\mode n}(q \at A)) $
          \\
          $ \ox F^{\mode m}_{\mode n}(q \at A) $
        }
        \arrow[dd, "(F^{\mode m}_{\mode n}(q \at A) \lto t) \ox \id"']
        \arrow[rd, "\textit{ev}"]
        \\
        (F^{\mode m}_{\mode n}(q \at A) \lto B)
        \ox
        F^{\mode m}_{\mode n}\interp{   \grassnt{q}   \sigma   \mid  \mathsf{N}  \odot  \Delta }_{\mode m}
        \arrow[rd, "\id \ox F^{\mode m}_{\mode n}(q \cdot e)"']
        &&
        (\interp{  \rho  \mid  \mathsf{M}  \odot  \Gamma }_{\mode n} \ox F^{\mode m}_{\mode n}(q \at A))
        \arrow[dd, "t"]
        \\
        &
        (F^{\mode m}_{\mode n}(q \at A) \lto B) \ox F^{\mode m}_{\mode n}(q \at A)
        \arrow[rd, "\textit{ev}"']
        \\
        &&
        B
      \end{tikzcd}
    \]
    The left/bottom edge is the former morphism,
    and the top/right edge is the latter morphism.
    The top left and bottom right squares commute by naturality of of the unit and counit,
    while the bottom left square commutes by functoriality of $ - \ox - $.
    Finally the trianlge commutes by one of triangle identities for adjunctions.
  \end{case}

  \begin{case}[Drop]\rm
    We are given $ \mathsf{l}  \le  \mathsf{n}  \le  \mathsf{m} $, and morphisms
    \begin{mathpar}
      \interp{  \rho  \mid  \mathsf{M}  \odot  \Gamma }_{\mode m} \xrightarrow{e} A
      \and
      \interp{  \sigma  \mid  \mathsf{N}  \odot  \Delta }_{\mode l}
      \ox F^{\mode m}_{\mode l} (q \at A)
      \xrightarrow{t}
      B
    \end{mathpar}
    The categorical interpretation
    $  \mathbin{\mathsf{let} } _{@  \grassnt{q} }   \operatorname{\downarrow} _{ \mathsf{n}  \le  \mathsf{m} }  \grassmv{x}   =   \operatorname{\downarrow} ^{ \grassnt{q} }_{ \mathsf{n}  \le  \mathsf{m} }  \grassnt{e}    \mathbin{\mathsf{in} }   \grassnt{t}  $ is
    \begin{align*}
      &\
        \interp{  \sigma  \mid  \mathsf{N}  \odot  \Delta }_{\mode l}
        \ox
        \interp{   \grassnt{q}   \rho   \mid  \mathsf{M}  \odot  \Gamma }_{\mode l}
      \\
      \iso
      &\
        \interp{  \sigma  \mid  \mathsf{N}  \odot  \Delta }_{\mode l}
        \ox
        F^{\mode n}_{\mode l}F^{\mode m}_{\mode n}
        \interp{   \grassnt{q}   \rho   \mid  \mathsf{M}  \odot  \Gamma }_{\mode m}
      \\
      \xrightarrow{F^{\mode n}_{\mode l} F^{\mode m}_{\mode n} (q \cdot e)}
      &\
        \interp{  \sigma  \mid  \mathsf{N}  \odot  \Delta }_{\mode l}
        \ox F^{\mode n}_{\mode l}F^{\mode m}_{\mode n}(q \at A)
      \\
      \iso
      &\
        \interp{  \sigma  \mid  \mathsf{N}  \odot  \Delta }_{\mode l}
        \ox F^{\mode m}_{\mode l} (q \at A)
      \\
      \xrightarrow{t}
      &\
        B
    \end{align*}
    On the other hand, by Lemma \ref{alem:subst-comp},
    the interpretation of $ \grasssym{[}  \grassnt{e}  \grasssym{/}  \grassmv{x}  \grasssym{]}  \grassnt{t} $ is
    \begin{align*}
      &\
        \interp{  \sigma  \mid  \mathsf{N}  \odot  \Delta }_{\mode l}
        \ox
        \interp{   \grassnt{q}   \rho   \mid  \mathsf{M}  \odot  \Gamma }_{\mode l}
      \\
      \iso
      &\
        \interp{  \sigma  \mid  \mathsf{N}  \odot  \Delta }_{\mode l}
        \ox
        F^{\mode m}_{\mode l}
        \interp{   \grassnt{q}   \rho   \mid  \mathsf{M}  \odot  \Gamma }_{\mode m}
      \\
      \xrightarrow{F^{\mode m}_{\mode l}(q \cdot e)}
      &\
        \interp{  \sigma  \mid  \mathsf{N}  \odot  \Delta }_{\mode l}
        \ox F^{\mode m}_{\mode l} (q \at A)
      \\
      \xrightarrow{t}
      &\
        B
    \end{align*}
    And it is straightforward verify that these morphisms are equal.
  \end{case}
  \begin{case}[Raise]\rm
    We are given $ \mathsf{m}  \le  \mathsf{n}  \le  \mathsf{M} $ and a morphism
    \[
      \interp{  \rho  \mid  \mathsf{M}  \odot  \Gamma }_{\mode m} \xrightarrow{t} A
    \]
    We have a commutative diagram
    \[
      \begin{tikzcd}
        \interp{  \rho  \mid  \mathsf{M}  \odot  \Gamma }_{\mode m}
        \arrow[r, "\sim"]
        &
        F^{\mode n}_{\mode m} \interp{  \rho  \mid  \mathsf{M}  \odot  \Gamma }_{\mode n}
        \arrow[r, "F^{\mode n}_{\mode m}(\text{unit})"]
        \arrow[dr, equals]
        &
        F^{\mode n}_{\mode m}
        G^{\mode n}_{\mode m}
        F^{\mode n}_{\mode m}
        \interp{  \rho  \mid  \mathsf{M}  \odot  \Gamma }_{\mode n}
        \arrow[r, "\sim"]
        \arrow[d]
        &
        F^{\mode n}_{\mode m}
        G^{\mode n}_{\mode m}
        \interp{  \rho  \mid  \mathsf{M}  \odot  \Gamma }_{\mode m}
        \arrow[r, "F^{\mode n}_{\mode m} G^{\mode m}_{\mode n}t"]
        \arrow[d]
        &
        F^{\mode n}_{\mode m} G^{\mode n}_{\mode m} A
        \arrow[d]
        \\
        & &
        F^{\mode n}_{\mode m}
        \interp{  \rho  \mid  \mathsf{M}  \odot  \Gamma }_{\mode n}
        \arrow[r, "\sim"]
        &
        \interp{  \rho  \mid  \mathsf{M}  \odot  \Gamma }_{\mode m}
        \arrow[r, "t"]
        &
        A
      \end{tikzcd}
    \]
    where unlabeled vertical arrows are the counit of
    $ F^{\mode m}_{\mode n} \dashv G^{\mode m}_{\mode n} $.
    The composition across the top and right edges is the interpretation
    of $  \operatorname{\uparrow}^{-1} _{ \mathsf{m}  \le  \mathsf{n} }  \grasssym{(}   \operatorname{\uparrow} _{ \mathsf{m}  \le  \mathsf{n} }  \grassnt{t}   \grasssym{)}  $,
    and this diagram shows that that morphism is equal to $ t $,
    concluding this case.
  \end{case}
\end{proof}

\newpage
\subsection{Semantic soundness of eta-conversion}
\label{asec:eta-semantic-soundness}

\begin{theorem}
  $ \eta $-conversion is sound with respect to the categorical semantics:
  If $  \rho  \mid  \mathsf{M}  \odot  \Gamma  \vdash_{ \mathsf{m} }  \grassnt{t_{{\mathrm{1}}}}  \equiv_\eta  \grassnt{t_{{\mathrm{2}}}}  \colon  \grassnt{A}  $,
  then the interpretations of
  $  \rho  \mid  \mathsf{M}   \odot   \Gamma  \vdash_{ \mathsf{m} }  \grassnt{t_{{\mathrm{1}}}}  \colon  \grassnt{A}  $ and
  $  \rho  \mid  \mathsf{M}   \odot   \Gamma  \vdash_{ \mathsf{m} }  \grassnt{t_{{\mathrm{2}}}}  \colon  \grassnt{A}  $ are equal in any categorical model.
\end{theorem}

\begin{proof}
  \setcounter{case}{0}%
  We have five cases to consider.
  \begin{case}[Pair]\rm
    In this case we are given a morphism
    $
    \bigotimes_i F^{\mode m_i}_{\mode m}(q_i \at B_i)
    = \interp{  \sigma  \mid  \mathsf{N}  \odot  \Delta  }_{\mode m}
    \xrightarrow{e} A_1 \ox A_2
    $
    and we must show that it is equal to the composite
    \[
      \interp{  \sigma  \mid  \mathsf{N}  \odot  \Delta  }_{\mode m}
      \iso
      F^{\mode m}_{\mode m}
      \interp{  \sigma  \mid  \mathsf{N}  \odot  \Delta  }_{\mode m}
      \xrightarrow{m^{-1}\circ F^{\mode m}_{\mode m}(\tau^{-1} \circ (1 \cdot e))}
      F^{\mode m}_{\mode m} (1 \at A_1) \ox F^{\mode m}_{\mode m} (1 \at A_2)
      \iso (1 \at A_1) \ox (1 \at A_2)
      \xrightarrow{\epsilon \ox \epsilon}
      A_1 \ox A_2
    \]
    with the morphism
    $
      F^{\mode m}_{\mode m} (1 \at A_1) \ox F^{\mode m}_{\mode m} (1 \at A_2)
      \iso (1 \at A_1) \ox (1 \at A_2)
      \xrightarrow{\epsilon \ox \epsilon}
      A_1 \ox A_2
    $
    being the interpretation of
    $   1   \grasssym{,}   1   \mid  \mathsf{m}  \grasssym{,}  \mathsf{m}   \odot   \grassmv{x_{{\mathrm{1}}}}  \grasssym{:}  \grassnt{A_{{\mathrm{1}}}}  \grasssym{,}  \grassmv{x_{{\mathrm{2}}}}  \grasssym{:}  \grassnt{A_{{\mathrm{2}}}}  \vdash_{ \mathsf{m} }  \grasssym{(}  \grassmv{x_{{\mathrm{1}}}}  \grasssym{,}  \grassmv{x_{{\mathrm{2}}}}  \grasssym{)}  \colon   \grassnt{A_{{\mathrm{1}}}}  \otimes  \grassnt{A_{{\mathrm{2}}}}   $.
    This follows from the following commutative diagram
    \[
      \begin{tikzcd}
        \bigotimes_i F^{\mode m_i}_{\mode m}(q_i \at B_i)
        \arrow[dddd, "e"']
        \arrow[r, "\sim"]
        &
        F^{\mode m}_{\mode m}
        \bigotimes_i F^{\mode m_i}_{\mode m}(1 \cdot q_i \at B_i)
        \arrow[d, equals]
        \arrow[r, "F^{\mode m}_{\mode m} \bigotimes_i \mu \circ F^{\mode m_i}_{\mode m}\delta"]
        &[3em]
        F^{\mode m}_{\mode m}
        \bigotimes_i 1 \at F^{\mode m_i}_{\mode m}(q_i \at B_i)
        \arrow[d, "F^{\mode m}_{\mode m}\tau"]
        \arrow[ld, "F^{\mode m}_{\mode m}\bigotimes_i \epsilon" description]
        \\
        &
        F^{\mode m}_{\mode m} \bigotimes_i F^{\mode m_i}_{\mode m}(q_i \at B_i)
        \arrow[d, "F^{\mode m}_{\mode m}e"']
        &
        F^{\mode m}_{\mode m} 1 \at \bigotimes_i F^{\mode m_i}_{\mode m}(q_i \at B_i)
        \arrow[l, "F^{\mode m}_{\mode m}\epsilon" description]
        \arrow[d, "F^{\mode m}_{\mode m} (1 \at e)"]
        \\
        &
        F^{\mode m}_{\mode m}(A_1 \ox A_2)
        \arrow[d, "m^{-1}"]
        \arrow[ldd, "\sim"']
        &
        F^{\mode m}_{\mode m}(1 \at (A_1 \ox A_2))
        \arrow[d, "F^{\mode m}_{\mode m}(\tau^{-1})"]
        \arrow[l, "F^{\mode m}_{\mode m}\epsilon" description]
        \\
        &
        F^{\mode m}_{\mode m}(A_1) \ox F^{\mode m}_{\mode m}(A_2)
        \arrow[ld, "\sim"]
        &
        F^{\mode m}_{\mode m}((1 \at A_1) \ox (1 \at A_2))
        \arrow[ul, "F^{\mode m}_{\mode m}(\epsilon \ox \epsilon)" description]
        \arrow[d, "m^{-1}"]
        \\
        A_1 \ox A_2
        &
        (1 \at A_1) \ox (1 \at A_2)
        \arrow[l, "\epsilon \ox \epsilon"]
        &
        F^{\mode m}_{\mode m}(1 \at A_1) \ox F^{\mode m}_{\mode m}(1 \at A_2)
        \arrow[ul,
        "F^{\mode m}_{\mode m}(\epsilon) \ox F^{\mode m}_{\mode m}(\epsilon)" description]
        \arrow[l, "\sim"]
      \end{tikzcd}
    \]
  \end{case}

  \begin{case}[Unit]\rm
    Similar to the Pair case, using the fact that
    $ \iota \from I \to 1 \at I $
    and
    $ \epsilon \from 1 \at I \to I $ are inverse morphisms.
  \end{case}

  \begin{case}[Raise]\rm
    Given a morphism $ e \from \interp{ \rho  \mid  \mathsf{M}  \odot  \Gamma }_{\mode n} \to G^{\mode n}_{\mode m} A $,
    applying \rulenames{grass}{term}{raiseE} and \rulenames{grass}{term}{raiseI}
    yields the morphism
    \[
      \interp{ \rho  \mid  \mathsf{M}  \odot  \Gamma }_{\mode n}
      \xrightarrow{\text{unit}}
      G^{\mode n}_{\mode m}F^{\mode n}_{\mode m} \interp{ \rho  \mid  \mathsf{M}  \odot  \Gamma }_{\mode n}
      \xrightarrow{G^{\mode n}_{\mode m}F^{\mode n}_{\mode m} e}
      G^{\mode n}_{\mode m}F^{\mode n}_{\mode m} G^{\mode n}_{\mode m} A
      \xrightarrow{G^{\mode n}_{\mode m}\text{counit}}
      G^{\mode n}_{\mode m}A
    \]
    and we must show that this morphism is equal to $ e $, which is staraightforward.
  \end{case}

  For the final two cases we make the following observation:
  Write $ v_A : F^{\mode m}_{\mode m}(1 \at A) \iso 1 \at A \xrightarrow{\epsilon} A $.
  Then for any $ q $, the composite
  $ q \at A \iso F^{\mode m}_{\mode m}(q \at A) \xrightarrow{q \cdot v_A} q \at A $
  is the identity.
  This follows from the diagram
  \[
	\begin{tikzcd}
      q \at A
      \arrow[r, "\sim"]
      \arrow[rd, equals]
      &
      F^{\mode m}_{\mode m}(q \at A)
      \arrow[d, "\sim"]
      \arrow[r, "F^{\mode m}_{\mode m}\delta"]
      &
      F^{\mode m}_{\mode m}(q \at 1 \at A)
      \arrow[r, "\mu"]
      \arrow[d, "\sim"]
      &
      q \at F^{\mode m}_{\mode m}(1 \at A)
      \arrow[r, "q \at v_A"]
      \arrow[d, "\sim"]
      &
      q \at 1 \at A
      \\
      &
      q \at A
      \arrow[r, "\delta"']
      &
      q \at 1 \at A
      \arrow[r, equals]
      &
      q \at 1 \at A
      \arrow[ur, "q \at \epsilon"']
    \end{tikzcd}
  \]
  The composite across the bottom edge is the identity
  by the laws for colax monoidal functors.
  The right square commutes by Lemma \ref{lem:ell-mu-induced-coherence} applied
  to the isomorphism $ G^{\mode m}_{\mode m} \iso \id $.

  \begin{case}[Drop]\rm
    In this case we are given a morphism
    $ e \from \interp{  \rho  \mid  \mathsf{M}  \odot  \Gamma }_{\mode n} \to F^{\mode m}_{\mode n}(q \at A) $.
    We must show that it is equal to the composite
    \[
      \interp{  \rho  \mid  \mathsf{M}  \odot  \Gamma }_{\mode n}
      \iso
      F^{\mode n}_{\mode n}
      \interp{  \rho  \mid  \mathsf{M}  \odot  \Gamma }_{\mode n}
      \xrightarrow{F^{\mode n}_{\mode n} e}
      F^{\mode n}_{\mode n} F^{\mode m}_{\mode n} (q \at A)
      \iso
      F^{\mode m}_{\mode n} (q \at A)
      \iso
      F^{\mode m}_{\mode n} F^{\mode m}_{\mode m} (q \at A)
      \xrightarrow{F^{\mode m}_{\mode n}(q \cdot v_A)}
      F^{\mode m}_{\mode n} (q \at A)
    \]
  but this is straightforward by the above observation.
  \end{case}

  \begin{case}[Functions]\rm
    In this case we are given a morphism
    $ e \from \interp{  \rho  \mid  \mathsf{M}  \odot  \Gamma }_{\mode n} \to F^{\mode m}_{\mode n}(q \at A) \lto B $.
    Applying the rule \rulename{grass}{term}{funE}
    to this morphism and $ v_A $ produces the morphism
    \[
      \interp{  \rho  \mid  \mathsf{M}  \odot  \Gamma }_{\mode n}
      \ox F^{\mode m}_{\mode n}(q \at A)
      \iso
      \interp{  \rho  \mid  \mathsf{M}  \odot  \Gamma }_{\mode n}
      \ox F^{\mode m}_{\mode n}F^{\mode m}_{\mode m}(q \at A)
      \xrightarrow{e \ox F^{\mode m}_{\mode n}(q \cdot v_A)}
      (F^{\mode m}_{\mode n}(q \at A) \lto B) \ox F^{\mode m}_{\mode n}(q \at A)
      \xrightarrow{\textit{ev}}
      B
    \]
    By the above observation this morphism is equal to
    $ \textit{ev} \circ (e \ox F^{\mode m}_{\mode n}(q \at A)) $,
    that is, the transpose of $ e $ along the adjunction
    $ (-) \ox F^{\mode m}_{\mode n}(q \at A) \dashv F^{\mode m}_{\mode n}(q \at A) \lto (-) $.
    Applying \rulename{grass}{term}{funI} to this morphism therefore yields $ e $,
    which is what we needed to show.
  \end{case}
\end{proof}

\fi
\end{document}